\documentclass[11pt]{article}
\usepackage{fullpage}

\usepackage[T1]{fontenc}
\DeclareFontShape{T1}{cmr}{m}{scit}{<->ssub * cmr/m/sc}{}

\usepackage[tbtags]{amsmath}
\usepackage{amssymb}
\usepackage{amsthm}
\usepackage{mathtools}
\usepackage{dsfont}

\usepackage{graphicx}
\usepackage{subcaption}
\usepackage[percent]{overpic}
\usepackage{tikz}
\usetikzlibrary{decorations.pathreplacing,calc,shapes.geometric,backgrounds}

\usepackage[svgnames]{xcolor}
\usepackage{tcolorbox}
\usepackage{mdframed}

\tcbset{
    colback=blue!3!white,
    colframe=blue!40!black,
    boxrule=0.8pt
}

\usepackage{marginnote}
\usepackage{appendix}
\usepackage{enumitem}

\usepackage[linesnumbered,ruled,vlined]{algorithm2e}

\SetAlgoSkip{smallskip}
\SetKwComment{Comment}{\# }{}
\SetEndCharOfAlgoLine{}
\SetKwInOut{Input}{Input}
\SetKwInOut{Output}{Output}
\SetKwFor{ForEach}{for}{}{end for}
\SetKwFor{While}{while}{}{end while}
\SetKwIF{If}{ElseIf}{Else}{if}{:}{elif}{else:}{end if}
\SetKwRepeat{Do}{do}{while}

\usepackage[round]{natbib}
\usepackage[colorlinks=true,linkcolor=DarkBlue,citecolor=DarkBlue]{hyperref}
\usepackage[capitalize]{cleveref}

\AddToHook{cmd/appendix/before}{
    \crefalias{section}{appendix}
    \crefalias{subsection}{appendix}
}

\newcommand{\notelist}{}

\NewDocumentCommand{\addnote}{mm}{
    \expandafter\gdef\expandafter\notelist\expandafter{
        \notelist
        \noindent \hyperlink{#2}{#1}\par
    }
    \hypertarget{#2}{#1}
}
\definecolor{greenbox}{rgb}{0.8,1,0.8}
\definecolor{redbox}{rgb}{1,0.8,0.8}

\DeclareMathOperator{\dist}{dist}

\DeclareMathOperator{\lca}{LCA}
\DeclareMathOperator{\depth}{depth}
\DeclareMathOperator{\Avg}{Avg}

\DeclareMathOperator{\poly}{poly}
\DeclareMathOperator{\Poi}{Poisson}

\DeclareMathOperator{\sign}{sign}

\newcommand*{\eps}{\ensuremath{\epsilon}}

\newcommand*{\err}{\text{err}}

\newcommand{\inner}[1]{\langle #1\rangle}
\newcommand{\set}[1]{\left\{#1\right\}}
\newcommand{\deltaest}{\hat{\delta}_{S}}

\newcommand{\prob}[1]{\Pr\!\left(#1\right)}

\newcommand{\Exp}[1]{\mathbb{E}\!\left[#1\right]}
\newcommand{\Var}[1]{\mathrm{Var}\!\left(#1\right)}

\newcommand{\ind}[1]{\mathds{1}_{\left\{#1\right\}}}

\newcommand{\Left}{{\mathcal L}}
\newcommand{\Right}{{\mathcal R}}

\newtheorem{result}{Theorem}

\newtheorem{theorem}{Theorem}[section]
\newtheorem{claim}[theorem]{Claim}
\newtheorem{corollary}[theorem]{Corollary}
\newtheorem{lemma}[theorem]{Lemma}
\newtheorem{definition}[theorem]{Definition}
\newtheorem{problem}[theorem]{Problem}

\newtheorem*{question*}{Main Question}

\title{Optimal Phylogenetic Reconstruction from Sampled Quartets}
\author{
Dionysis Arvanitakis\thanks{D.~Arvanitakis and K.~Makarychev
were supported by NSF Awards CCF-1955351 and EECS-2216970, and in part by grants from the NSF (DMS-2235451) and the Simons Foundation (MPS-NITMB-00005320) to the NSF-Simons National Institute for Theory and Mathematics in Biology (NITMB). 
}
\\Northwestern University
\and
Vaggos Chatziafratis\thanks{V. Chatziafratis and Y. Luo were supported by a UC Santa Cruz start-up grant and by Hellman's fellowship. Part of this work was done while visiting Archimedes Research Unit, Athens, Greece.}\\UC Santa Cruz
\and
Yiyuan Luo\footnotemark[2]\\UC Santa Cruz
\and 
Konstantin Makarychev\footnotemark[1]\\Northwestern University
}

\date{}
\begin{document}
\maketitle
\thispagestyle{empty}

\begin{abstract}
    Quartet Reconstruction, the task of recovering a single phylogenetic tree from smaller trees on four species called \textit{quartets}, is a well-studied problem in theoretical computer science with far-reaching connections to statistics, graph theory and biology. Given a random sample containing $m$ noisy quartets, labeled according to an unknown ground-truth tree $T$ on $n$ taxa, we want to \textit{learn the tree structure of $T$ with small generalization error,} i.e., to output a tree $\widehat T$ that is \textit{close} to $T$ in terms of quartet distance and can predict the classification of unseen quartets. Unfortunately, the empirical risk minimizer corresponds to the $\mathsf{NP}$-hard problem of finding a tree that maximizes agreements with the sampled quartets, and earlier works in approximation algorithms gave $(1-\eps)$-approximation schemes (PTAS) for dense instances with $m=\Theta(n^4)$ quartets, or for $m=\Theta(n^2\log n)$ quartets \textit{randomly} sampled from $T$.

    Prior to our work, it was unknown how many samples are information-theoretically required to learn the tree, and whether there is an efficient reconstruction algorithm. We present optimal results for reconstructing an unknown phylogenetic tree $T$ from a random sample of $m=\Theta(n)$ quartets, potentially corrupted under the standard Random Classification Noise (RCN) model. This matches the $\Omega(n)$ lower bound required for any meaningful tree reconstruction, as for $m=o(n)$, large parts of $T$ cannot be recovered, and  \textit{exact} tree reconstruction ($\eps=0$) requires $\Omega(n^3)$ quartets. Our contribution is twofold: first, we give a tree reconstruction algorithm that, not only achieves a $(1-\eps)$-approximation for Quartet Reconstruction, but most importantly \textit{recovers} a tree close to $T$ in quartet distance; second, we show a new $\Theta(n)$ bound on the Natarajan dimension of phylogenies (an analog of VC dimension in multiclass classification), which may be of independent interest. Coupled together, these imply that our reconstructed tree $\widehat T$ will generalize well to unseen quartets. Our analysis relies on a new \textit{Quartet-based Embedding and Detection ($\mathrm{QED}$)} procedure, that repeatedly identifies and removes well-clustered subtrees from the (unknown) ground-truth $T$ via semidefinite programming.
\end{abstract}

\newpage
\tableofcontents
\thispagestyle{empty}
\newpage
\setcounter{page}{1}

\section{Introduction}\label{sec:intro}
How can we reconstruct a phylogenetic tree from partial information on how different species are related? This is a classic question in computational biology and a central topic in the field of Phylogenetics, where a major task is to recover the Tree of Life for living and extinct species~\citep{colonius1981tree,maddison2007tree,jansson2016improved}. Beyond biology, similar questions for finding hierarchical tree representations of data from labeled samples (often called a \textit{hierarchical clustering}) naturally arise across disciplines, ranging from linguistics with the evolution of languages~\citep{gray2009language,nakhleh2005comparison,walker2012cultural}, to computer science and statistics, as datasets and networks tend to organize hierarchically~\citep{ward1963hierarchical,ravasz2003hierarchical,clauset2008hierarchical,nickel2017poincare}.

A common approach for reconstructing a tree works by first obtaining small subtrees on overlapping sets of taxa, most commonly on 4 species called \textit{quartets}, and then, trying to aggregate/amalgamate these subtrees into one big tree over the whole set~\citep{strimmer1996quartet,jansson2005rooted,byrka2010new,jansson2016improved}. One of the most basic problems here is the $\mathsf{NP}$-complete problem of Quartet Reconstruction~\citep{steel1992complexity,jiang1998orchestrating,jiang2001polynomial}:
\begin{problem}[Quartet Reconstruction, \textsc{MaxQuartetRec}]\label{prob:rec}
    Given a set $L$ of $n$ items, and a sample of $m$ quartets of the form $ab|cd$ with $a,b,c,d\in  L$, Quartet Reconstruction aims to find an unrooted tree $T$ with $n$ leaves labeled from $L$, so that the path $a,b$ is vertex-disjoint from the path $c,d$ in $T$, for all given quartets (in this case we say quartet $ab|cd$ is satisfied). The corresponding optimization problem of maximizing satisfied quartets is denoted \textsc{MaxQuartetRec}.
\end{problem}

Quartet-based reconstruction has played a central role in phylogenetics, both in theory and in practice, see survey by~\cite{chor1998quartets} and standard textbooks by~\cite{semple2003phylogenetics,felsenstein2004inferring}. A quartet tree, i.e., a phylogenetic tree\footnote{In other words, a quartet tree is an unrooted tree with $4$ leaves, where every internal node is of degree 3.} on 4 elements $a,b,c,d$, constitutes the basic informative building unit in phylogenetic applications~\citep{avni2015weighted,avni2019new}, and lies at the heart of popular quartet-puzzling methods~\citep{strimmer1996quartet,yang2013quartet,reaz2014accurate}. Recently, assuming Unique Games~\citep{khot2002power}, it was shown that \textsc{MaxQuartetRec} is $\mathsf{NP}$-hard to approximate beyond a $\tfrac13$-approximation factor~\citep{chatziafratis2023triplet}. Note that there are only 3 distinct quartets on $a,b,c,d$, namely $ab|cd,ac|bd,ad|bc$, so a $\tfrac13$-approximation is also obtained by the trivial algorithm that outputs a random tree on $n$ leaves (each quartet is equally likely), which means that \textsc{MaxQuartetRec} is approximation resistant in the worst-case~\citep{haastad2001some,hast2005beating,GHMRC11}.

\subsection{Learning Trees from Noisy Samples}
\label{sec:Intro:LearnTree}

The approximation of \textsc{MaxQuartetRec} (Problem~\ref{prob:rec}) has been instrumental in tree reconstruction approaches~\citep{strimmer1996quartet,bryant2001constructing,yang2013quartet,reaz2014accurate}, and can be seen as a (hierarchical) clustering objective. For example, early work by~\cite{jiang1998orchestrating,jiang2001polynomial} designed a PTAS for fully-dense instances of $m=\binom{n}{4}$ quartets, and later~\cite{snir2011linear,snir2012reconstructing} significantly improved the dependence on the sample size to $m=\Theta(n^2\log n)$ random samples. However, the true goal of most phylogenetic applications is to \textit{recover} a meaningful ground-truth tree on $n$ species of interest, i.e., to reconstruct the quartet relationships on them, and not just approximate a proxy objective. This is of course a well-known issue with all clustering objectives~\citep{balcan2009approximate,ackerman2009clusterability,bilu2012stable,daniely2012clustering,makarychev2014bilu,angelidakis2017algorithms,ben2018clustering}, where an objective function (e.g., $k$-means, etc.) serves only as a means to an end for recovering ground-truth solutions, with various stability notions being proposed to avoid worst-case impossibility results, see book by~\cite{roughgarden2021beyond}.

\paragraph{Reconstruction of an unknown tree $T$.} In most applications, indeed there is a meaningful ground-truth tree $T$ to be recovered, often called a latent tree~\citep{anandkumar2011spectral,kandiros2023learning}, and the $m$ given quartets are obtained by some noisy classification on a random sample $(a,b,c,d)$ at the leaves (see App.~\ref{app:fig}, Fig.~\ref{fig:quartet_label}), e.g., by an experiment on the genomes of 4 species, or by asking an expert~\citep{wu2008reconstructing,vikram2016interactive,emamjomeh2018adaptive,ghoshdastidar2019foundations}. Quartets offer several advantages, as they do not make assumptions on numerical values of pairwise distances (see Chapter 12 in~\cite{felsenstein2004inferring}), or on the ancestry of species from triplets~\citep{aho1981inferring,byrka2010new,bryant2000computing,wu2008reconstructing,emamjomeh2018adaptive}. In fact, in statistics, quartets played a major role in the resolution of Steel's conjecture for the Cavender-Farris-Neyman/Ising mutation model:~\cite{mossel2004phase,daskalakis2006optimal,daskalakis2011evolutionary} designed ``quartet tests'' that avoided the bias of sampled sequences of characters at the leaves of a ground-truth tree, obtaining the optimal reconstruction bounds for the mutation probability. Finally, earlier works study PAC-learning of evolutionary trees under various stochastic models for evolution, such as the two-state general Markov model~\citep{ambainis1997nearly,farach1999efficient,cryan2001evolutionary}.

Here, our focus is on reconstructing a tree $\widehat T$ that is \textit{close} to $T$, based on $m$ noisy \textit{quartet} samples from $T$. Before we define the problem of \textit{Tree Learning}, we need a notion of \textit{distance} between trees, captured by the classic notion of quartet distance~\citep{felsenstein2004inferring,brodal2013efficient,dudek2019computing}:

\begin{definition} (Quartet Distance on Trees, $\mathrm{dist}(\widehat T,T)$) Given two trees $\widehat T,T$ on the same set $L$ of $n$ leaves, the quartet distance $\mathrm{dist}(\widehat T,T)$ is the fraction out of $\binom{n}{4}$ quartets on which the two trees disagree, i.e., the fraction of induced quartet topologies that are different in $\widehat T$ compared to $T$.
\end{definition}

 Since every tree uniquely defines a set of $\binom{n}{4}$ quartets, quartet distance is a standard measure of similarity between phylogenetic trees, focusing on the branching patterns of quartet trees on the $n$ taxa~\citep{grunewald2007closure,brodal2013efficient}. It ranges between $\approx\frac{2}{3}$ %$\Theta(n^4)$
if two trees widely disagree, and zero, if they are identical. In fact, the maximum quartet distance has been studied in combinatorics: a conjecture of~\cite{bandelt1986reconstructing} states that the maximum quartet distance between two trees is always at most $\frac23+o(1)$, and the currently best-known bound is $0.69 + o(1)$ by~\cite{alon2016maximum}. Furthermore, quartet distance has deep connections with graph theory; as it turns out, computing the quartet distance is equivalent (up to polylogarithmic factors) to counting $4$-cycles~\citep{dudek2019computing}, leading to surprising connections with fine-grained complexity and conjectures made by~\cite{yuster1997finding,dahlgaard2017finding}. In this paper, we develop algorithms for recovery in quartet distance from \textit{noisy} quartets.

\paragraph{Random Classification Noise (RCN) Model.} Let $\mathcal{T}$ be the class of all phylogenetic trees on the $n$ leaves labeled from $L$. Moreover, let $T \in\mathcal T$ be any ground-truth tree whose structure we want to recover via a tree $\widehat T$, such that $\mathrm{dist}(\widehat T,T)\le \eps$, for some target accuracy $\eps$. Notice that being close in quartet distance means that such a tree $\widehat{T}$ \textit{generalizes} well, i.e., it would be successful even at predicting unseen quartets. Towards this, we follow the standard Random Classification Noise (RCN) model under the uniform distribution~\citep{angluin1988learning,natarajan2013learning}, where we are given a random sample of $m$ quartets, each of which is generated as follows:

\begin{itemize}
    \item Draw 4 distinct elements $(a,b,c,d)$ uniformly at random from the set $L$ of leaves;
    \item With probability $1-\eta$, we get the correct quartet label 
    $$T(a,b,c,d) \in \{\{ab|cd\},\{ac|bd\},\{ad|bc\}\};$$
    \item Otherwise, w.p. $\eta$, we get a random wrong quartet label, that disagrees with $T(a,b,c,d)$.
\end{itemize}

A quartet $t=ab|cd$ drawn from the RCN model with error $\eta$ is denoted as $t\sim \mathcal{RCN}(\eta)$ ($\eta$, or an upper bound to it, is known to the algorithm). Notice that since there are 3 labels for a quartet, $\eta=\tfrac23$ corresponds to a uniformly random quartet, so from now on $\eta<\tfrac23$. The RCN noise model is quite standard as it captures for instance the aforementioned scenarios where labeled examples are provided by a noisy expert or experiment. The criterion for successfully learning the ground-truth $T$ is \textit{generalization}:
%\\aiNote{53}{There is a contradiction about whether $\eta$ is known: the text says “$\eta$ is unknown to the algorithm” but the inline author note says “known”. Please clarify the assumption, since it affects the stated guarantees and sample complexity dependence on $\eta$.}

\begin{problem}[Phylogenetic Tree Learning]\label{prob:learn} Let $\mathcal{T}$ be the class of all phylogenetic trees on $n$ leaves. Given small constants $\eps,\delta>0$, using as few sampled quartets as possible, efficiently find a tree $\widehat T\in \mathcal T$, that generalizes well to unseen quartets. That is, with probability $1-\delta$, find a tree $\widehat T\in \mathcal T$ so the probability of error on 4 random elements $(a,b,c,d)$ is small compared to the best tree $T\in \mathcal T$, or equivalently,\footnote{This is in the realizable setting. In the agnostic setting, samples consist of tuples $(a,b,c,d)$ and a quartet label $q\in \{\{ab|cd\},\{ac|bd\},\{ad|bc\}\}$, with the desired guarantee being $\prob{\widehat{T}(a,b,c,d)\neq q}\leq \min_{T\in \mathcal{T}}\prob{T(a,b,c,d)\neq q}+\epsilon$.},
    \[
        \Pr [\widehat T(a,b,c,d) \neq T(a,b,c,d)] \le \eps
    \]
    where $T(a,b,c,d) \in \{\{ab|cd\},\{ac|bd\},\{ad|bc\}\}$ denotes the unique quartet satisfied by tree $T\in\mathcal T$.
\end{problem}

Observe that this is the standard generalization framework of PAC-learning~\citep{shalev2014understanding}; here our algorithms will be robust, so we also allow input samples to be noisy, and the reconstructed output tree $\widehat T$ will still have small quartet distance from the ground-truth $T$, $\mathrm{dist}(\widehat T,T)\le\eps$.  \textsc{MaxQuartetRec} (Problem~\ref{prob:rec}) corresponds to the Empirical Risk Minimizer (ERM) on the sample of size $m$; indeed, approximating the \textsc{MaxQuartetRec} objective can be seen as a prerequisite step towards recovery and it has been extensively studied in theoretical computer science~\citep{jiang1998orchestrating,jiang2001polynomial,byrka2010new,snir2011linear,snir2012reconstructing}.  However, learning the structure of $T$ is more challenging and the main question behind our work is the following:

\begin{question*}[Sampling Complexity and Efficient Tree Learning]\label{q:main}
    \textit{How many (noisy) quartet samples $m$ are necessary to learn the underlying ground-truth tree $T$? Most importantly, is there a polynomial-time algorithm that reconstructs $T$ (up to generalization error $\eps$)?}
\end{question*}

\subsection{Our Results: Optimal Tree Reconstruction with \texorpdfstring{$m=\Theta(n)$}{m=Θ(n)} Quartets}
We present two results, one on the sampling complexity of tree learning (Theorem~\ref{th:dimension}), and one on algorithmically learning the structure of a ground-truth tree (Theorem~\ref{th:ptas}), via a new \textit{Quartet-based Embedding and Detection} ($\mathrm{QED}$) procedure based on semidefinite programming, which also implies a PTAS for \textsc{MaxQuartetRec} from $m=O(n)$ samples (see~\hyperref[sec:tech-overview]{Technical Overview}).

\paragraph{Sampling Complexity.} In PAC-learning, the Natarajan dimension~\citep{natarajan1989learning,haussler1995generalization} captures the sampling complexity for learning a set of functions, generalizing the VC-dimension from binary to multi-class settings~\citep{shalev2014understanding,daniely2015multiclass,avdiukhin2023tree}. Natarajan dimension is suitable for our tree learning purposes, since 4 elements $a,b,c,d$ can be classified in three different ways by a ground-truth tree $T$, namely $T(a,b,c,d) \in \{\{ab|cd\},\{ac|bd\},\{ad|bc\}\}$, but we defer precise definitions of Natarajan shattering to Section~\ref{sec:nat-dimension}.  Our first result is a tight bound on the Natarajan dimension of learning trees (information-theoretically) from quartets:

\begin{tcolorbox}
    \begin{theorem}[Sampling Complexity of Tree Learning]\label{th:dimension}
        The Natarajan dimension of learning phylogenetic trees on $n$ leaves  from quartets is $\Theta(n)$, hence the sampling complexity of Problem~\ref{prob:learn} is $m=\Theta_{\epsilon, \delta}(n)$ quartets.
    \end{theorem}
\end{tcolorbox}

This means that, in principle, the reconstruction of the ground-truth tree $T$ is possible with only $m=\Theta(n)$ samples. Notice here,
this result holds for any distribution, not necessarily the uniform. However, previous results on \textit{efficient} PTAS algorithms required much larger samples of size $m=\Theta(n^4)$, later improved to $m=\Theta(n^2\log n)$~\citep{jiang1998orchestrating,jiang2001polynomial,snir2011linear,snir2012reconstructing}. Observe that since there are $|\mathcal{T}|\le (2n-2)^{2n-4}$ leaf-labeled trees on $n$ leaves, a direct upper bound for the Natarajan dimension is $O(n\log n)$. To improve the bound, we do a quartet-preserving embedding of a phylogenetic tree $T$, by associating each quadruple $v_i,v_j,v_k,v_l$ of leaves to a degree-$8$ polynomial. The sign of the polynomial reveals what the correct split for $v_i,v_j,v_k,v_l$ is according to $T$, and then we bound the number of sign-patterns of  polynomials~\citep{warren1968lower,alon1985geometrical}. Up to constant factors, our result also recovers the bound in \cite{avdiukhin2023tree} for the rooted triplet consistency problem, which is solvable in polynomial time~\citep{aho1981inferring,byrka2010new,emamjomeh2018adaptive}. Unfortunately, contrary to triplets, checking satisfiability of quartets is $\mathsf{NP}$-complete~\citep{steel1992complexity}.

\paragraph{Algorithm for Tree Learning.} Even in the absence of noise ($\eta = 0$), a trivial lower bound for any meaningful tree reconstruction to be possible in terms of quartet distance is $m = \Omega(n)$, as otherwise, with $m = o(n)$, most of the ground-truth tree would not even be observed in the sample. Our main result is a polynomial-time algorithm that recovers an unknown ground-truth hierarchy $T$ from $m=\Theta(n)$ quartet samples on its leaves $L$, matching the above information-theoretic bound, and that is robust against noise from $\mathcal{RCN}(\eta)$ model above. Our main technical ingredient is a tree reconstruction algorithm that we call \textit{Quartet-based Embedding and Detection} ($\mathrm{QED}$) that repeatedly detects well-clustered subtrees of $T$ via semidefinite programming, then removes them, and ultimately combines them into a final tree $\widehat T$ with small quartet distance from $T$. In particular, our algorithm implies a new PTAS for \textsc{MaxQuartetRec} with $m=\Theta(n)$ noisy quartets improving the previously-best sample size of $m=\Theta(n^2\log n)$~\citep{jiang1998orchestrating,jiang2001polynomial,snir2011linear,snir2012quartet}.
\begin{tcolorbox}
    \begin{theorem}[Tree Learning with $m=\Theta(n)$ quartets]\label{th:ptas}
There exists an algorithm with the following guarantee. For every $\epsilon>0$ and unknown ground-truth tree $T$ on $n$ items, given a uniformly random sample of $m=c\cdot n$ quartets consistent with $T$ (where $c=c(\epsilon)$ is a constant depending only on $\epsilon$), the algorithm outputs a tree $\widehat{T}$ with $\dist(\widehat{T}, T)\leq \epsilon$ with probability $1-o(1)$. The algorithm runs in time polynomial in $n$. 

\smallskip

Moreover, the guarantee continues to hold in the presence of noise. Specifically, when $m=c(\epsilon,\eta)\cdot n$ sampled quartets are corrupted under the $\mathcal{RCN}(\eta)$ model (with $\eta<2/3$), the algorithm still guarantees $\dist(\widehat T,T)\le\epsilon$. This yields a $(1-\epsilon)$-approximation scheme (PTAS) for \textsc{MaxQuartetRec}. 
\end{theorem}
\end{tcolorbox}

Here, the sampling complexity is polynomial on both $\frac{1}{\eps}$ and $\frac{1}{1-\frac{3}{2}\eta}$, and the runtime of the algorithm is polynomial in $n$, specifically the runtime is $O(f(\epsilon)\poly(n))$, where $f(\epsilon)=\left(\frac{1}{\epsilon}\right)^{O(\frac{1}{\epsilon})}$. Notice that no restrictions are imposed on the topology of the ground-truth $T$ which could be adversarially chosen. Even without noise ($\eta=0$), allowing some reconstruction error $\eps>0$ is reasonable, as for the \textit{exact} tree reconstruction problem ($\eps=0$), at least  $\Omega(n^3)$ quartet samples are required; in fact, there is a query lower bound of $\Omega(n^3)$ against any \textit{non-adaptive} query algorithm (future queries cannot depend on previous answers), let alone our setting where algorithms only have access to a randomly drawn sample (see proof in Appendix~\ref{app:query-complexity}). For the case where adaptivity is allowed, $O(n\log n)$ queries suffice~\citep{wu2008reconstructing,emamjomeh2018adaptive}. Furthermore, even for the noiseless setting ($\eta=0$), our result significantly improves prior works which gave a PTAS with $m=\Theta(n^2\log n)$ randomly sampled quartets labeled correctly by $T$~\citep{snir2011linear,snir2012reconstructing}. Observe that if errors were completely adversarial, any algorithm would suffer quartet distance $\mathrm{dist}(\widehat T,T) > \eta$ and would not generalize. Earlier works for \textsc{MaxQuartetRec}, gave a PTAS with much larger sample size of $m=\Theta(n^4)$ quartets~\citep{jiang1998orchestrating,jiang2001polynomial,snir2011linear}, using techniques from~\cite{arora1995polynomial,arora2002new} for dense graph problems. Moreover, recall that for worst-case sample distributions, hardness-of-approximation results prevent us from approximating \textsc{MaxQuartetRec} beyond the trivial baseline of $\tfrac13$ achieved by a random tree~\citep{chatziafratis2023triplet}.

Taken together, Theorems~\ref{th:dimension},~\ref{th:ptas} answer the \hyperref[q:main]{Main Question} for learning trees (Problem~\ref{prob:learn}).

\subsection{Further Related Work}
Quartets have played an important role in phylogenomics, statistics and machine learning for recovering latent tree structures~\citep{daskalakis2006optimal,daskalakis2011evolutionary,anandkumar2011spectral,avni2015weighted,avni2019new,kandiros2023learning}, and also in logic and algebra due to closure operations on trees~\citep{grunewald2007closure,bodirsky2010complexity,bodirsky2012complexity,bodirsky2017complexity}. More technically distant works on recovering trees from \textit{pairwise} information on the leaves have also been widely-studied in computer science and metric embeddings. For example, a related problem arising in phylogenetics is given an $n\times n$ matrix describing dissimilarity on $n$ taxa, find the tree metric or ultrametric that best fits the matrix~\citep{agarwala1998approximability,ailon2005fitting,charikar2024improved,cohen2024fitting,cohen2025fitting}. Algorithms for community detection in the Stochastic Block Model have been extensively studied in the literature; see e.g.,~\cite{boppana1987eigenvalues,feige2001heuristics,mcsherry2001spectral,coja2010graph,massoulie2014community,
makarychev2012approximation,
makarychev2014constant,
makarychev2016learning,
   mossel2015consistency,feldman2015subsampled,
    guedon2016community,
abbe2018community}. For hierarchical clustering, these settings were studied by~\cite{cohen2017hierarchical,cohen2019hierarchical,manghiuc2021hierarchical} in beyond-worst-case models with a planted tree, where they show a constant-factor approximation for the similarity-based objective introduced by~\cite{dasgupta2016cost}. These works differ from our setting as they do not address quartets.

\paragraph{Roadmap.} Section~\ref{sec:tech-overview} provides the technical overview for our reconstruction algorithm, and Section~\ref{sec:tree-reconstruction} provides the details of the \emph{clustering}, \emph{reconstruction}, and \emph{verification} procedures assuming black box access to $\mathrm{QED}$ procedure. Section~\ref{sec:community-detection} describes in detail the $\mathrm{QED}$ procedure: we formulate and analyze the SDP for recovering the leaves of a subtree of $T$. Finally, Section~\ref{sec:nat-dimension} provides the proof for the bound on the Natarajan dimension. Throughout the paper, we use $\binom{S}{k}$ to denote the set of $k$-subsets of $S$. $T$ denotes the ground-truth tree on $L$ leaves, and $\widehat T$ is the reconstructed tree. We use $\mathcal{Q}(T)$ or simply $\mathcal{Q}$ for the set of $\binom{n}{4}$ quartets that $T$ satisfies. Trees can be rooted arbitrarily, as per the definition of quartets, this will not affect satisfaction of quartets or quartet distance. For a rooted tree $T$, we use $T_u$ for the subtree rooted at a vertex $u\in T$, and $T_{uv}$ for the subtree rooted at the lowest common ancestor of $u$ and $v$, $\lca(u,v)$.

\section{Technical Overview}\label{sec:tech-overview}

In this section, we provide an overview of the tree reconstruction algorithm.
We assume that there exists a ground-truth tree $T$ with $n$ leaves.
The algorithm receives a collection of $m = O(n)$ random quartets $(ab|cd)$.
Each quartet is obtained by sampling four distinct leaves $a,b,c,d$ uniformly at random from the set of all 4-tuples of leaves for which the relation $(ab|cd)$ holds in the ground-truth tree $T$.
To model noise, each quartet is independently replaced with one of the two unsatisfied quartets, $(ac|bd)$ or $(ad|bc)$, each with probability $\eta/2$ (see Section~\ref{sec:Intro:LearnTree} for details).
Thus, $\eta$ represents the probability that a sampled quartet is misclassified, and $\eta = 0$ corresponds to the noiseless setting.

Our goal is to reconstruct the underlying ground-truth tree $T$ from this collection of (possibly noisy) quartet constraints.
Given such samples, we aim to recover a tree $\widehat{T}$ whose structure is close to $T$ in the quartet distance.
We show that, with high probability, the reconstructed tree $\widehat{T}$ differs from $T$ on at most an $\eps$-fraction of quartets while using only polynomial time and a number of samples linear in $n$ (see Theorem~\ref{th:ptas}).

In this overview, we focus on the noiseless classification setting, where $\eta = 0$.
Note that even when $\eta = 0$, we cannot directly use an algorithm for finding a satisfying assignment, since this problem is $\mathsf{NP}$-complete and hence no polynomial-time algorithm exists (unless $\mathsf{P} = \mathsf{NP}$).
Later, we will build on the ideas presented in this section when describing the algorithm that applies to the case $\eta \in [0, \tfrac23)$.

The ground-truth tree $T$ is initially unrooted.
For convenience, we introduce a root $r$ by subdividing one of the edges of $T$ so that the two subtrees adjacent to $r$ contain between one-third and two-thirds of the leaves; that is, their sizes differ by at most a factor of $2{:}1$.
From this point onward, we view $T$ as a rooted tree.
Note that the quartet relation $(ab|cd)$ is independent of the choice of the root: it evaluates to true if and only if the paths between $a$ and $b$ and between $c$ and $d$ are vertex-disjoint.

The starting point of our discussion is the following observation: the probability that, for a given pair of vertices $(a,b)$, there exists a quartet constraint $(ab|cd)$ or $(cd|ab)$ in the random sample is slightly higher when the vertices $a$ and $b$ are close to each other rather than far apart.
More precisely, this probability is proportional to
\begin{equation}\label{eq:prob-edge-ab}
    1 - \frac{2|T_{ab}|}{n} + O\Big((|T_{ab}|/n)^2\Big),
\end{equation}
provided that $|T_{ab}|/n$ is sufficiently small, where $T_{ab}$ denotes the subtree rooted at the least common ancestor of $a$ and $b$.

To take advantage of this observation, we construct an auxiliary graph $G$ on the leaves of $T$.
For each quartet constraint $(ab|cd)$ in our random sample, we add either the edge $ab$ or $cd$ to $G$, each with probability $1/2$.
In this way, we obtain a random graph $G$ on the vertex set $L$ (the set of leaves of the tree), where the presence of all edges is independent, and the probability that an edge between $a$ and $b$ appears depends on the positions of $a$ and $b$ in the ground-truth tree $T$.

The random graph arises from a stochastic model similar to the Hierarchical Stochastic Block Model introduced by~\cite*{cohen2019hierarchical}.
Hence, one might consider applying the algorithms developed by~\cite{cohen2019hierarchical} in our setting.
Unfortunately, this is not possible for the following reasons:
(i) the general algorithm in the paper by~\cite{cohen2019hierarchical} provides only a constant-factor approximation, which is insufficient for our purposes; and
(ii) the probability bound~\eqref{eq:prob-edge-ab} holds only for leaves $a$ and $b$ that are sufficiently close to each other.

Nevertheless, we can use the graph $G$ to extract information about the structure of the underlying tree.
Building on techniques for learning communities in the Stochastic Block Model with $k$ communities and adversarial noise \citep*{makarychev2016learning},
we develop a procedure that identifies a set of leaves belonging to a subtree of size $\delta n$, where $\delta$ is a sufficiently small constant.

The idea behind this learning step is as follows.
Let $u$ be a vertex in the tree, and let $L_u$ denote the set of leaves in the subtree rooted at $u$.
Suppose $|L_u| \approx \delta n$.
Then, the probability that two leaves $a,b \in L_u$ are connected by an edge in $G$ is slightly higher than the probability of an edge $ac$, where $a \in L_u$ and $c \notin L_u$.
However, note that
(i) for leaves $c$ that are close to $L_u$, the probability of the edge $ac$ may be very close to that of $ab$; and
(ii) for leaves $c$ that are very far from $a$, this probability can actually be quite large. Thus, we cannot identify the set $L_u$ exactly but can only recover it approximately: we find a set $R_u$ such that
\[
    |L_u \triangle R_u| \leq O(\delta^2 n).
\]

Towards identifying such a set $R_u$, we employ a semidefinite program (SDP),
specified in Eqs.~(\ref{eq:sdpobj})-(\ref{eq:sdp3}), as a community
detection procedure. This SDP embeds the graph into a Euclidean space, with the hope that
the vectors corresponding to vertices within $S$ form a cluster --
that is, they are tightly grouped together
and well separated from the remaining vertices. However, as discussed above, the probability that
an edge between two vertices $u,v \in S$ is present is nearly identical to the probability of an edge
between $u \in S$ and a vertex $w \notin S$ when $w$ is sufficiently close to $u$.
Moreover, these probabilities may depend on the unknown topology of the underlying  tree $T$.
To account for this, we introduce a buffer set $B$ surrounding $S$,
consisting of intermediate vertices that separate $S$ from the remainder of the
ground-truth tree~$T$.
Our analysis shows that $|B| \le O(\delta^2 n)$.
Finally, to ensure that the SDP does not return an excessively large set
that could deviate from the intended tree structure of~$T$,
we include an appropriate \emph{spreading constraint}
(see Eq.~(\ref{eq:sdp3})).

Equipped with this recovery procedure, we can approximately recover the leaves of all subtrees $T_u$ containing approximately $\delta n$ leaves.
However, this information alone is not sufficient to approximately reconstruct the entire tree, since many, or even most, leaves may not belong to such subtrees $T_u$.

To address this, we consider an iterative algorithm.
The algorithm maintains a set of active leaves, denoted by $A_i$.
Initially, all leaves are active, i.e., $A_0 = L$.
At iteration $i$, the algorithm approximately identifies the set of leaves belonging to the subtree rooted at some vertex $u_i$ of $T$ that are still active at the beginning of that iteration, that is, the set $L_{u_i} \cap A_i$.
Denote the identified set by $R_i$.
We then remove $R_i$ from $A_i$ by setting
$
A_{i+1} = A_i \setminus R_i,
$
and proceed to the next iteration.

After running this iterative process, we obtain a partition of the leaves into clusters $R_1, \dots, R_k$.
Each cluster has size approximately $\delta n$, and the total number of clusters satisfies $k \approx 1 / \delta$.
(For certain technical reasons, we need to recover the left and right parts of $T$ separately; however, we omit this and other low-level details from this overview.)
Each set $R_i$ approximately corresponds to a set $L_{u_i}\cap A_i$, which denotes the leaves active at step $i$ in the subtree rooted at some vertex $u_i$.
Let $U = \{u_1, \dots, u_k\}$.

We now define sets $S_i$ as follows: each $S_i$ consists of all leaves $x$ in $T$ such that the unique shortest path from $x$ to $u_i$ does not contain any other vertex $u_j \in U$.
We show that each $R_i$ approximates not only $L_{u_i}\cap A_i$ but also $S_i$, that is,
\[
    \sum_{i=1}^k |R_i \triangle S_i| \leq \varepsilon n,
\]
for some small $\varepsilon > 0$.
Importantly, while the sets $R_i$ and $A_i$ depend on the entire execution of our algorithm, the partition $S_1, \dots, S_k$ is uniquely determined by the set $U$ of vertices $u_1, \dots, u_k$.

We now reconstruct the tree by first \emph{guessing} the relative positions of the nodes $u_1, \dots, u_k$ in the ground-truth tree $T$.
We then build a binary skeleton consisting of the vertices $u_1, \dots, u_k$ and at most $k$ additional Steiner vertices. We finally attach each set $R_i$ to the corresponding vertex $u_i$.

At a high level, this reconstruction procedure is similar to the approach of~\cite{snir2012reconstructing}, who refer to such a skeleton as a \emph{model} tree (see also~\cite*{jiang1998orchestrating,jiang2001polynomial,chatziafratis2023triplet}).
However, our method differs in several key respects.
Rather than \emph{guessing} the model in advance, we first identify the vertices $u_1, \dots, u_k$ using the recovery procedure and only then \emph{guess} their relative positions.
This strategy both simplifies the reconstruction process and makes it more robust to misclassified nodes produced by the recovery step.
In particular, it ensures that we never need to recover sets $L_{u_i}\cap A_i$ that are very small i.e., of size less than $\delta n$.

We next show that the reconstructed tree $\widehat{T}$ is close to the ground-truth tree $T$ in terms of the quartet distance.
The argument proceeds in two steps:
(a) if we were to attach the true sets $S_i$ to the reconstructed skeleton, the resulting tree would provide a good approximation of $T$; and
(b) since the recovered sets $R_i$ are close to the corresponding $S_i$, the same guarantee holds for our reconstructed tree $\widehat{T}$.

Throughout the reconstruction process, the algorithm makes several nondeterministic choices, which may lead to multiple candidate trees $\widehat{T}$.
We know that at least one of these candidates is a good approximation of the ground-truth tree $T$, but we do not initially know which one.
To identify a near-correct tree, we evaluate all candidate trees on the random instance of quartets $Q$ and select the one that achieves the highest consistency.
This validation step is justified by a uniform convergence argument based on our bound on the Natarajan dimension of the quartet consistency problem (see Theorem~\ref{thm:natarajan_quart}, Section~\ref{sec:nat-dimension}).

\section{Tree Reconstruction}\label{sec:tree-reconstruction}

In this section, we present the \emph{Quartet-based Embedding and Detection (QED)} tree reconstruction algorithm and prove Theorem~\ref{th:ptas}.
The construction follows the technical overview provided in Section~\ref{sec:tech-overview}.
Recall that we assume the ground-truth tree $T$ is rooted, and that both its left and right subtrees contain at least one third of all leaves.

The main building block of the algorithm is the community detection procedure. In Section~\ref{sec:algo_grothendieck} we describe a procedure to generate a graph based on the quartet set $Q$ and introduce the notion of a \emph{good graph} $G$, defined as a graph that satisfies the conditions of Lemma~\ref{lem:Actual_to_expected_bound}.
Lemma~\ref{lem:Actual_to_expected_bound} shows that, with probability exponentially close to~$1$ (specifically, with probability at least
$1 - \operatorname{poly}(1/\delta) \exp(-1.5n)$) over the randomness in the quartet set and the internal randomness of the algorithm, the graph constructed is good. 
Throughout this section, we shall assume that~$G$ is a good graph. We fix a sufficiently small parameter $\delta = \Theta(\epsilon)$, where $\epsilon$ denotes the target accuracy.

Below, we use the notation $T_u$ for the subtree rooted at vertex $u$ and $L_u$ for the set of leaves of $T_u$.
We prove the following theorem
in Section~\ref{sec:community-detection} (see Theorem~\ref{thm:qed-guarantee}).

\begin{theorem}[QED Procedure]
    \label{thm:QED}
    There exists a randomized polynomial-time algorithm that, given a random sample of quartets $Q$ for an unknown tree $T$, a subset of leaves $A \subseteq L$, and a side indicator $\sigma \in \{\Left, \Right\}$ (where $\Left$ and $\Right$ denote the left and right sides, respectively), returns a set $R \subseteq A$ such that $|R| \in [\delta n, 2\delta n]$.

    The algorithm satisfies the following guarantee.
    With probability at least $p_{\delta}$ (over the internal randomness of the algorithm),
    there exists a vertex $u$ in the $\sigma$-subtree of $T$ such that
    \[
        |R \triangle (L_u \cap A)| \le C \delta^2 n,
    \]
    where $p_{\delta} > 0$ is a constant depending only on $\delta$,
    provided that the following conditions hold:
    (i) $|Q| \ge \alpha_{\delta,\eta} n$ and $G$ is a good graph (see above);
    (ii) at least $48\delta n$ leaves of $A$ lie in the $\sigma$-side of $T$; and
    (iii) at least $n/4$ leaves of $A$ lie in the opposite subtree of $T$,
    where $\alpha_{\delta,\eta}$ is a constant depending only on $\delta$ and $\eta$.
\end{theorem}

We postpone the detailed discussion of the QED procedure until Section~\ref{sec:community-detection} and now describe the overall tree reconstruction algorithm. The algorithm proceeds in three main phases: the \emph{clustering step}, the \emph{reconstruction step}, and the \emph{verification step}.

\begin{algorithm}[ht]\label{fig:alg:clustering}
\caption{\textsc{ClusteringStep}$(Q, L, \rho_{\Left}, \rho_{\Right}, \delta)$}
\label{fig:alg:clustering-step}
\Input{Quartet sample $Q$, leaves $L$, estimates $(\rho_{\Left}, \rho_{\Right})$, parameter $\delta$}
\Output{Clusters $\{R^{\Left}_i\}$, $\{R^{\Right}_i\}$, and $R^{\circ}$}

Let $\delta = \Theta(\epsilon)$ be a sufficiently small parameter.

\ForEach{$\sigma \in \{\Left, \Right\}$\textup{:}}{

    $A_1^{\sigma} = L$ \Comment{Initialize active leaves}
    $i = 1$

    \While{$ \sum_{i'<i}
        |R_{i'}^{\sigma}|
        \leq \rho_{\sigma} - 49\delta n$\textup{:}}{
        $R^{\sigma}_i = \textsc{QED}(Q, A_i^{\sigma}, \sigma)$
        \Comment{Apply QED on the current side}

        $A^{\sigma}_{i+1} = A^{\sigma}_i \setminus R^{\sigma}_i$ \Comment{Remove clustered leaves}
        $i = i + 1$
    }
    $k_{\sigma}=i-1$ \Comment{Number of clusters for $\sigma$}
}

$R^{\circ} = L \setminus
(\bigcup_{\sigma,i} R^{\sigma}_i)$
\Comment{Place all unassigned leaves into $R^{\circ}$}

\Return $\{R^{\Left}_1, \dots, R^{\Left}_{k_{\Left}}\}$,
        $\{R^{\Right}_1, \dots, R^{\Right}_{k_{\Right}}\}, R^{\circ}$
\end{algorithm}

\subsection{Clustering Step}
The algorithm begins by guessing the sizes  of the left and right subtrees of the unknown tree $T$ up to an additive error of $\delta n$.
We denote this guess by $(\rho_{\Left}, \rho_{\Right})$.
This guess is correct with probability at least $\delta$; that is, with probability at least $\delta$, we have
$\big|\rho_{\Left} - |L_{r_{\Left}}|\big| \le \delta n$,
where $r_{\Left}$ is the left child of the root $r$, and $L_{r_{\Left}}$ denotes the set of leaves in the subtree rooted at $r_{\Left}$.

Having fixed the estimate $(\rho_{\Left}, \rho_{\Right})$, the algorithm recursively partitions the leaves on the left and right sides of the tree.
The algorithm first selects a side $\sigma \in \{\Left, \Right\}$; throughout, we use $\Left$ and $\Right$ to denote the left and right sides, respectively. It then iteratively partitions the leaves on the $\sigma$-side of the tree.
The algorithm maintains a set of active leaves $A^{\sigma}_i \subseteq L$, consisting of the leaves that remain unclustered at the beginning of iteration~$i$.
In each iteration, it invokes the QED procedure (Theorem \ref{thm:QED}) with the $\sigma$-side indicator and the current active set $A^{\sigma}_i$.
The QED procedure returns a subset $R^{\sigma}_i \subseteq A^{\sigma}_i$ of size between $\delta n$ and $2\delta n$.
After obtaining $R^{\sigma}_i$, the algorithm removes these leaves from the active set, setting
\[
    A^{\sigma}_{i+1} = A^{\sigma}_i \setminus R^{\sigma}_i,
\]
and repeats the process.
This iterative procedure continues as long as the total number of leaves in the recovered sets $R^{\sigma}_1, \dots, R^{\sigma}_i$ remains below $\rho_{\sigma} - 49\delta n$.

The clustering step is applied independently to the left and right subtrees of the tree, producing collections
$R^{\Left}_1, \dots, R^{\Left}_{k_{\Left}}$ and
$R^{\Right}_1, \dots, R^{\Right}_{k_{\Right}}$.
We note that the recovered clusters on the left and right sides may overlap slightly, and a small fraction of leaves may remain unassigned.
We place these unassigned leaves into a separate set $R^{\circ}$.
The pseudocode for the clustering step is provided in Algorithm~\ref{fig:alg:clustering-step}.

We now state and prove the main theorem concerning the clustering step.
To this end, we first introduce the following definition (see Figure~\ref{fig:phi-partition}).

\begin{definition}\label{def:phi_partition}
    Let $T$ be a rooted tree, and let $U$ be a subset of its vertices that includes the root.
    Define the mapping $\phi_U : V(T) \to U$ as follows: for each vertex $v \in V(T)$,
    $\phi_U(v)$ is the first vertex in $U$ encountered on the unique path from $v$ to the root.
    Then, for each $u \in U$, define
    \[
        S_U(u) = \{v \in L : \phi_U(v) = u\}.
    \]
    In other words, $S_U(u)$ is the set of leaves $v$ such that the path from $v$ to $u$ does not contain any other vertices from $U$.
\end{definition}

\begin{figure}[ht]
    \centering
    \begin{tikzpicture}[
        dot/.style={circle,fill,inner sep=1.6pt},
        uset/.style={draw,fill=black,regular polygon,regular polygon sides=4,inner sep=2.2pt},
        >=stealth
    ]
        % ---- coordinates ----
        % Root (U-vertex)
        \coordinate (r)   at (0, 4.2);
        % Level 1
        \coordinate (a)   at (-2.4, 3.2);   % left child of root (internal)
        \coordinate (b)   at ( 2.4, 3.2);   % right child of root (internal)
        % Level 2
        \coordinate (u1)  at (-3.2, 2.2);   % U-vertex in left subtree
        \coordinate (c)   at (-1.6, 2.2);   % internal
        \coordinate (u2)  at ( 1.6, 2.2);   % U-vertex in right subtree
        \coordinate (d)   at ( 3.2, 2.2);   % internal
        % Level 3
        \coordinate (e)   at (-3.6, 1.2);   % internal under u1
        \coordinate (f)   at (-2.8, 1.2);   % internal under u1
        \coordinate (g)   at ( 1.2, 1.2);   % internal under u2
        \coordinate (h)   at ( 2.0, 1.2);   % internal under u2
        % Leaves
        \coordinate (l1)  at (-3.9, 0.2);
        \coordinate (l2)  at (-3.3, 0.2);
        \coordinate (l3)  at (-3.1, 0.2);
        \coordinate (l4)  at (-2.5, 0.2);
        \coordinate (l5)  at (-1.9, 0.2);
        \coordinate (l6)  at (-1.3, 0.2);
        \coordinate (l7)  at ( 0.9, 0.2);
        \coordinate (l8)  at ( 1.5, 0.2);
        \coordinate (l9)  at ( 1.7, 0.2);
        \coordinate (l10) at ( 2.3, 0.2);
        \coordinate (l11) at ( 2.9, 0.2);
        \coordinate (l12) at ( 3.5, 0.2);

        % ---- background regions for S_U sets ----
        \begin{scope}[on background layer]
            % S_U(u1) — leaves l1..l4
            \fill[blue!12, rounded corners=6pt]
                (-4.2, -0.15) rectangle (-2.2, 0.55);
            % S_U(u2) — leaves l7..l10
            \fill[red!12, rounded corners=6pt]
                (0.6, -0.15) rectangle (2.6, 0.55);
            % S_U(u3) — leaves l5,l6
            \fill[olive!12, rounded corners=6pt]
                (-2.15, -0.15) rectangle (-1.05, 0.55);
            % S_U(r) — leaves l11,l12
            \fill[green!12, rounded corners=6pt]
                (2.65, -0.15) rectangle (3.75, 0.55);
        \end{scope}

        % ---- edges ----
        \draw (r) -- (a);
        \draw (r) -- (b);
        \draw (a) -- (u1);
        \draw (a) -- (c);
        \draw (b) -- (u2);
        \draw (b) -- (d);
        \draw (u1) -- (e);
        \draw (u1) -- (f);
        \draw (e) -- (l1);
        \draw (e) -- (l2);
        \draw (f) -- (l3);
        \draw (f) -- (l4);
        \draw (c) -- (l5);
        \draw (c) -- (l6);
        \draw (u2) -- (g);
        \draw (u2) -- (h);
        \draw (g) -- (l7);
        \draw (g) -- (l8);
        \draw (h) -- (l9);
        \draw (h) -- (l10);
        \draw (d) -- (l11);
        \draw (d) -- (l12);

        % ---- nodes ----
        % U-vertices (squares)
        \node[uset,label={above:$r$}] at (r)  {};
        \node[uset,label={left:$u^{\Left}_1$}] at (u1) {};
        \node[uset,label={left:$u^{\Right}_1$}] at (u2) {};
        \node[uset,label={left:$u^{\Left}_2$}] at (a) {};

        % Internal nodes (dots)
        \foreach \p in {b,c,d,e,f,g,h}{\node[dot] at (\p) {};}

        % Leaves (dots)
        \foreach \p in {l1,l2,l3,l4,l5,l6,l7,l8,l9,l10,l11,l12}{
            \node[dot] at (\p) {};
        }

        % ---- phi_U arrows: all leaves (gray) ----
        % S_U(u1): l1,l2,l3,l4 -> u1
        \draw[->, dashed, thin, gray]
            (l1) .. controls ++(-0.3, 0.5) and ++(-0.3, -0.2) .. ([xshift=-4pt]u1);
        \draw[->, dashed, thin, gray]
            (l2) .. controls ++(-0.15, 0.6) and ++(-0.1, -0.4) .. ([xshift=-1.5pt]u1);
        \draw[->, dashed, thin, gray]
            (l3) .. controls ++(0.1, 0.6) and ++(0.1, -0.4) .. ([xshift=1.5pt]u1);
        \draw[->, dashed, thin, gray]
            (l4) .. controls ++(0.3, 0.5) and ++(0.3, -0.2) .. ([xshift=4pt]u1);
        % S_U(u2): l7,l8,l9,l10 -> u2
        \draw[->, dashed, thin, gray]
            (l7) .. controls ++(-0.3, 0.5) and ++(-0.3, -0.2) .. ([xshift=-4pt]u2);
        \draw[->, dashed, thin, gray]
            (l8) .. controls ++(-0.1, 0.6) and ++(-0.1, -0.4) .. ([xshift=-1.5pt]u2);
        \draw[->, dashed, thin, gray]
            (l9) .. controls ++(0.1, 0.6) and ++(0.1, -0.4) .. ([xshift=1.5pt]u2);
        \draw[->, dashed, thin, gray]
            (l10) .. controls ++(0.3, 0.5) and ++(0.3, -0.2) .. ([xshift=4pt]u2);
        % S_U(u3): l5,l6 -> a
        \draw[->, dashed, thin, gray]
            (l5) .. controls ++(-0.5, 1.2) and ++(0.3, -0.6) .. ([xshift=2pt]a);
        \draw[->, dashed, thin, gray]
            (l6) .. controls ++(0.3, 1.2) and ++(1.0, -0.6) .. ([xshift=4pt]a);
        % S_U(r): l11,l12 -> r
        \draw[->, dashed, thin, gray]
            (l11) .. controls ++(0.1, 2.2) and ++(1.0, -0.8) .. ([xshift=2pt]r);
        \draw[->, dashed, thin, gray]
            (l12) .. controls ++(0.3, 2.2) and ++(3.2, -0.8) .. ([xshift=4pt]r);

        % ---- S_U labels (braces) ----
        \draw[decorate, decoration={brace, amplitude=5pt, mirror}]
            (-4.1, -0.25) -- (-2.3, -0.25)
            node[midway, below=6pt, font=\small] {$S_U(u^{\Left}_1)$};

        \draw[decorate, decoration={brace, amplitude=5pt, mirror}]
            (0.7, -0.25) -- (2.5, -0.25)
            node[midway, below=6pt, font=\small] {$S_U(u^{\Right}_1)$};

        \draw[decorate, decoration={brace, amplitude=5pt, mirror}]
            (-2.1, -0.25) -- (-1.1, -0.25)
            node[midway, below=6pt, font=\small] {$S_U(u^{\Left}_2)$};

        \draw[decorate, decoration={brace, amplitude=5pt, mirror}]
            (2.7, -0.25) -- (3.7, -0.25)
            node[midway, below=6pt, font=\small] {$S_U(r)$};

        % ---- legend ----
        \node[uset, label={right:\small $\in U$}] at (4.5, 4.2) {};
        \node[dot,  label={right:\small $\notin U$}] at (4.5, 3.6) {};
        \draw[->, dashed, thin, gray] (4.2, 2.8) -- (4.8, 2.8)
            node[right, font=\small] {$\phi_U$};

    \end{tikzpicture}
    \caption{Illustration of Definition~\ref{def:phi_partition}.
    Vertices in $U = \{r, u^{\Left}_1, u^{\Left}_2, u^{\Right}_1\}$ are shown as filled squares.
    For each leaf, $\phi_U$ maps it to the nearest ancestor in $U$ 
    (dashed arrows).
    The colored regions indicate the partition of leaves into $S_U(u^{\Left}_1)$, $S_U(u^{\Left}_2)$, $S_U(u^{\Right}_1)$, and $S_U(r)$.}
    \label{fig:phi-partition}
\end{figure}
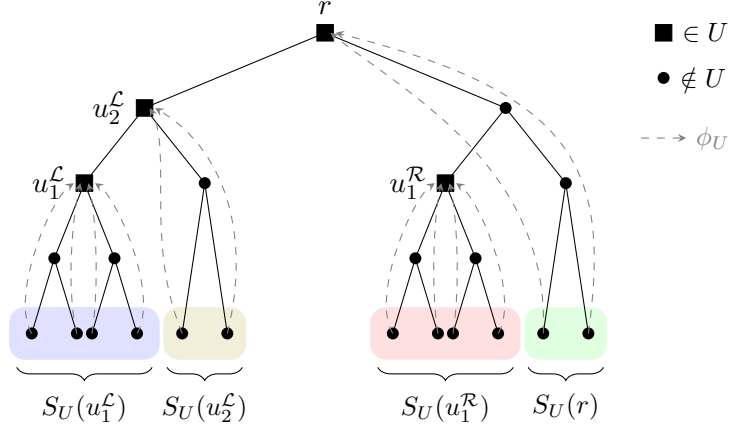

\begin{theorem}[Correctness of the Clustering Step]
    \label{thm:clustering-correctness}
    There exists a constant $p_{\delta}=\delta^{O(\frac{1}{\delta})} > 0$ (the \emph{success probability}) such that,
    with probability at least $p_{\delta}$, the clusters
    $R^{\Left}_1, \dots, R^{\Left}_{k_{\Left}}$,
    $R^{\Right}_1, \dots, R^{\Right}_{k_{\Right}}$, and
    $R^{\circ}$
    produced by the clustering step satisfy the following property.
    There exists a set of vertices
    \[
        U =
        \Big\{
            r,\,
            u^{\Left}_1, \dots, u^{\Left}_{k_{\Left}},
            u^{\Right}_1, \dots, u^{\Right}_{k_{\Right}}
        \Big\}
    \]
    in the ground-truth tree $T$ such that
    \begin{equation}
        \label{eq:thm:clustering}
        \sum_{\sigma \in \{\Left, \Right\}}
        \sum_{i=1}^{k_{\sigma}} \bigl| R^{\sigma}_i \triangle S_U(u^{\sigma}_i) \bigr|
        + \bigl| R^{\circ} \triangle S_U(r) \bigr|
        \le C \delta n,
    \end{equation}
    for some universal constant $C > 0$.
\end{theorem}

\begin{proof}
    The clustering procedures for the left and right subtrees are independent, and thus we will analyze them separately.
    We say that the procedure \emph{succeeds} for side $\sigma \in \{\Left, \Right\}$ if there exist vertices
    $u^{\sigma}_1, \dots, u^{\sigma}_{k_{\sigma}}$, all lying on the $\sigma$-side of the tree $T$, such that
    \begin{equation}
        \label{eq:thm:clustering-one-side}
        \sum_{i=1}^{k_{\sigma}} \bigl| R^{\sigma}_i \triangle
        S_{U^{\sigma}}(u^{\sigma}_i) \bigr|
        \le C' \delta n,
    \end{equation}
    where $U^{\sigma} = \{r, u^{\sigma}_1, \dots, u^{\sigma}_{k_{\sigma}}\}$ and
    $C'$ is a universal constant.

    We show that the algorithm succeeds for each side of the tree
    with a positive probability $p'_{\delta}$ that depends only on $\delta$.
    For clarity of exposition, we focus on the left side of the tree;
    the argument for the right side is entirely analogous. We remind the reader that $|L_{r_{\Left}}|$ denotes the number of leaves
    in the left subtree.

    \begin{lemma}
        \label{lem:one-side-success}
        For a fixed constant $\delta > 0$, there exists a positive constant
        $p'_{\delta} > 0$ such that, conditioned on the event
        $\bigl|\rho_{\Left} - |L_{r_{\Left}}|\bigr| \le \delta n$,
        the clustering procedure applied to the left side of the tree
        succeeds with probability at least $p'_{\delta}=\delta^{O(\frac{1}{\delta})}$.
        In other words, with conditional probability at least $p'_{\delta}$,
        there exist vertices
        $u^{\Left}_1, \dots, u^{\Left}_{k_{\Left}}$
        on the left side of the tree $T$ satisfying~\eqref{eq:thm:clustering-one-side}.
    \end{lemma}
    We first establish the following claim.
    \begin{claim}
        \label{cl:lem:one-side-success}
        For a fixed constant $\delta > 0$, there exists a positive constant
        $p'_{\delta} > 0$ such that, conditioned on the events
        $\bigl|\rho_{\Left} - |L_{r_{\Left}}|\bigr| \le \delta n$ and that the constructed graph $G$ is good (see Section \ref{sec:algo_grothendieck} and earlier discussion),
        all invocations of the QED procedure succeed with probability
        at least $p'_{\delta}=\delta^{O(\frac{1}{\delta})}$.
    \end{claim}

    \begin{proof}[Proof of Claim~\ref{cl:lem:one-side-success}]
        Observe that the number of invocations of the QED procedure is at most $1/\delta$.
        Indeed, each call to QED removes at least $\delta n$ leaves from the active set $A^{\Left}_i$,
        and since there are at most $n$ leaves in total, the procedure can iterate at most $1/\delta$ times.

        Next, we verify that the preconditions required for invoking the QED procedure
        remain valid throughout the process, as long as the algorithm continues to succeed.
        Specifically, we must ensure that, at every iteration~$i$,
        the left side of the tree contains at least $48\delta n$ active leaves
        and the right side contains at least $n/4$ active leaves,
        assuming that all previous QED invocations have succeeded.

        Recall that the algorithm continues clustering the left side as long as
        $
        \sum_{i'<i}
        |R_{i'}^{\sigma}|
        \leq \rho_{\sigma} - 49\delta n.
        $
        Assuming that $\rho_{\sigma}$ was guessed within an additive error of $\delta n$, we have
        \[
            |L_{r_{\Left}}| - \sum_{i'<i}
            |R_{i'}^{\sigma}|
            \ge (\rho_{\sigma} - \delta n) - (\rho_{\sigma} - 49\delta n)
            \ge 48\delta n.
        \]
        Thus, at each iteration, the left side continues to contain enough active leaves
        to satisfy the conditions of Theorem~\ref{thm:QED}.

        Each successful QED call may misclassify or remove at most $C\delta^2 n$ leaves
        from the right side. Therefore, after all iterations, the total number of leaves
        erroneously removed from the right side is at most
        $C\delta^2 n \cdot (1/\delta) = C\delta n$.
        Since the right subtree initially contains at least $n/3$ leaves by assumption on the root placement,
        it follows that, even after accounting for misclassifications,
        the right side still contains at least $n/4$ leaves for sufficiently small~$\delta$.

        Finally, by Theorem~\ref{thm:QED}, each call of the QED procedure succeeds with probability at least $p''_{\delta}=\Omega(\poly(\delta))$, conditioned on the successful completion of all previous calls. Since there are at most $1/\delta$ invocations in total,
        the probability that all of them succeed simultaneously is at least
        $
        p'_{\delta} = (p''_{\delta})^{1/\delta}=\delta^{O(\frac{1}{\delta})} > 0,
        $
        which depends only on~$\delta$.
        This completes the proof of the claim.
    \end{proof}
    \begin{proof}[Proof of Lemma~\ref{lem:one-side-success}]
        By Claim~\ref{cl:lem:one-side-success}, the probability that all invocations of the QED procedure on the left side succeed is at least $p'_{\delta}$.
        We now condition on this event and show that, under this assumption, the clustering algorithm for the left side succeeds.

        At iteration $i$ of the clustering algorithm, the QED procedure returns a set $R_i^{\Left}$.
        Since we assume that the QED procedure succeeds at every step, including iteration $i$, there exists a vertex $u_i^{\Left}$ such that       \[|R_i^{\Left} \triangle (L_{u_i^{\Left}} \cap A_i^{\Left})| \le C \delta^2 n,\]       
        where $L_{u_i^{\Left}} \cap A_i^{\Left}$ is the set of active leaves at the beginning of iteration~$i$ in the subtree rooted at $u_i^{\Left}$.      This implies that at least $\delta n - C \delta^2 n$ leaves belonging to the true subtree
        $L_{u_i^{\Left}}$ are active before iteration $i$, since $|R_i^{\Left}| \ge \delta n$,
        and that at most $C \delta^2 n$ leaves from $L_{u_i^{\Left}}$ remain active after iteration $i$.

        We now make a useful observation. Consider two iterations $i$ and $j$ with $i < j$.
        If $u_i^{\Left}$ were an ancestor of $u_j^{\Left}$,
        then $L_{u_j^{\Left}} \subseteq L_{u_i^{\Left}}$ and hence
        \[
            \delta n - C \delta^2 n
            \le
            |L_{u_j^{\Left}} \cap A_j^{\Left}|
            \le
            |L_{u_i^{\Left}} \cap A_j^{\Left}|
            \le C \delta^2 n.
        \]
        Since the right-hand side is strictly smaller than the left-hand side for sufficiently small $\delta$,
        it follows that $u_i^{\Left}$ cannot be an ancestor of $u_j^{\Left}$ whenever $i < j$.

        We now bound the left-hand side of~\eqref{eq:thm:clustering-one-side} for $\sigma = \Left$.
        Consider a leaf $x$ that belongs to the symmetric difference
        $R_i^{\Left} \triangle S_{U^{\Left}}(u_i^{\Left})$ for some $i$.
        Let $i^*$ be the smallest such index. Then $x$ must belong to the active set $A^{\Left}_{i^*}$,
        since otherwise there would exist an index $j^* < i^*$ such that
        $x \in R_{j^*}^{\Left}$.
        By the minimality of $i^*$, this would also imply that
        $x \in S_{U^{\Left}}(u_{j^*}^{\Left})$.
        Consequently, $x \notin R_{i^*}^{\Left}$ and
        $x \notin S_{U^{\Left}}(u_{i^*}^{\Left})$,
        as the sets $R_{j}^{\Left}$ are pairwise disjoint and the sets
        $S_{U^{\Left}}(u_{j}^{\Left})$ are pairwise disjoint as well.
        This leads to a contradiction with the assumption that
        $x \in R_{i^*}^{\Left} \triangle S_{U^{\Left}}(u_{i^*}^{\Left})$.
        The same argument also shows that $x \notin S_{U^{\Left}}(u_{j}^{\Left})$ for $j < i^*$. We thus conclude that $x \in A^{\Left}_{i^*}$.

        Since 
        $x\in 
        R_{i^*}^{\Left} \triangle S_{U^{\Left}}(u_{i^*}^{\Left})$ 
        and
        $x\in A^{\Left}_{i^*}$, we have
        \[
            x \in
            \bigl(R_{i^*}^{\Left} \triangle S_{U^{\Left}}(u_{i^*}^{\Left})\bigr)\cap
            A^{\Left}_{i^*}
            =
            R_{i^*}^{\Left} \triangle \bigl(S_{U^{\Left}}(u_{i^*}^{\Left})\cap
            A^{\Left}_{i^*}\bigr),
        \]
        where we used 
        $R_{i^*}^{\Left}\subset A_{i^*}^{\Left}$.
        We claim that
        $x\in S_{U^{\Left}}(u_{i^*}^{\Left})\cap A^{\Left}_{i^*}$ if and only if $x\in        L_{u_{i^*}^{\Left}} \cap A_{i^*}^{\Left}$.
        Note that $S_{U^{\Left}}(u_{i^*}^{\Left})\subseteq L_{u_{i^*}^{\Left}}$, and therefore if
        $x\in S_{U^{\Left}}(u_{i^*}^{\Left})\cap A^{\Left}_{i^*} $, then $x\in       L_{u_{i^*}^{\Left}} \cap A_{i^*}^{\Left}$.
For the reverse direction, assume $x\in L_{u_{i^*}^{\Left}}\cap A_{i^*}^{\Left}$
        and consider the path from $x$ to the root.
        Let $u_j^{\Left}$ be the first vertex from $U^{\Left}$ encountered along this path.
        As argued above, $j \le i^*$.
        If $j \ne i^*$, then $j < i^*$, which would imply that
        $x \in S_{U^{\Left}}(u_j^{\Left})$ with $j<i^*$, but this would contradict the minimality of $i^*$.
        
        We thus conclude that
        $x\in R_{i^*}^{\Left} \triangle \bigl(L_{u_{i^*}^{\Left}} \cap A_{i^*}^{\Left}\bigr)$.
        This means that if
        $x\in R_i^{\Left} \triangle S_{U^{\Left}}(u_i^{\Left})$, then
        $x\in R_{i}^{\Left} \triangle \bigl(L_{u_{i}^{\Left}} \cap A_{i}^{\Left}\bigr)$.
        Since any $x$ may belong to $R_i^{\Left} \triangle S_{U^{\Left}}(u_i^{\Left})$
        for at most two distinct values of $i$, we obtain
        \begin{equation*}
            \sum_{i=1}^{k_{\Left}}
            \bigl| R^{\Left}_i \triangle S_{U}(u^{\Left}_i) \bigr|
            \le
            2
            \sum_{i=1}^{k_{\Left}}
            \bigl|R_{i}^{\Left} \triangle \bigl(L_{u_{i}^{\Left}} \cap A_{i}^{\Left}\bigr)\bigr|
            \le 2k_{\Left}\, C\delta^2 n,
        \end{equation*}
        where the last inequality follows from the guarantee of Theorem~\ref{thm:QED}.
        Recalling that $k_{\Left}\le 1/\delta$ completes the proof.
    \end{proof}

    In Lemma~\ref{lem:one-side-success}, we established that, with constant probability,
    the algorithm successfully clusters the left side of the tree,
    conditioned on the event~$\mathcal{E}'$ that the sizes of the left and right subtrees
    were guessed within an additive error of~$\delta n$.
    An identical guarantee holds for the right side.
    The probability of the event~$\mathcal{E}'$ is at least~$\delta$,
    since there are at most $1/\delta$ possible values for the size of the left subtree
    on a grid with step size~$\delta n$ in the range~$[n/3, 2n/3]$.
    Conditioned on~$\mathcal{E}'$, the events that the algorithm successfully clusters
    the left and right sides are independent.
    Consequently, the probability that both sides are successfully clustered
    is at least~$\delta \cdot (p'_{\delta})^2=\delta^{O(\frac{1}{\delta})}$,
    which is a constant depending only on~$\delta$.
    To complete the proof of Theorem~\ref{thm:clustering-correctness},
    we next show that if both subtrees are clustered correctly,
    then the entire tree is clustered successfully.

    \begin{claim}\label{cl:cluster-left-right}
        If the algorithm succeeds for both the left and right sides, then the combined clustering satisfies the desired inequality~\eqref{eq:thm:clustering}.
    \end{claim}
    \begin{proof}
        Suppose that the algorithm succeeds for both the left and right sides.
        Then there exist sets $U^{\Left}$ and $U^{\Right}$, where each vertex
        $u^{\sigma}_i$ lies on side $\sigma \in \{\Left, \Right\}$,
        and inequality~\eqref{eq:thm:clustering-one-side} holds for both sides.
        Define $U = U^{\Left} \cup U^{\Right}$.

        Consider a vertex $u^{\Left}_i \in U^{\Left}$ and the corresponding sets
        $S_{U}(u^{\Left}_i)$ and $S_{U^{\Left}}(u^{\Left}_i)$.
        Both sets contain only leaves from the left side of the tree, since each
        $u^{\Left}_i$ lies on that side and any path from a vertex on the right side to the root
        cannot pass through $u^{\Left}_i$.
        Therefore, no leaf on the right side can be assigned to $u^{\Left}_i$.

        Now consider a leaf $x$ on the left side of the tree.
        The path from $x$ to the root does not pass through any vertex
        $u^{\Right}_j$.
        Hence, $x$ is assigned to $S_{U}(u^{\Left}_i)$
        if and only if $u^{\Left}_i$ is the first vertex on the path from $x$ to the root
        that belongs to $U^{\Left}$.
        This is precisely the same rule used for assigning $x$ to
        $S_{U^{\Left}}(u^{\Left}_i)$.
        We conclude that
        $
        S_{U^{\Left}}(u^{\Left}_i) = S_{U}(u^{\Left}_i)$
        for all $i$, and similarly,
        $S_{U^{\Right}}(u^{\Right}_j) = S_{U}(u^{\Right}_j)
        \quad \text{for all } j
        $.
        The bound~\eqref{eq:thm:clustering-one-side} now implies
        \begin{equation}
            \sum_{i=1}^{k_{\sigma}}
            \bigl| R^{\sigma}_i \triangle S_{U}(u^{\sigma}_i) \bigr|
            \le C' \delta n.
        \end{equation}
        Combining the two bounds for the left and right sides, we obtain
        \[
            \sum_{\sigma \in \{\Left, \Right\}}
            \sum_{i=1}^{k_{\sigma}}
            \bigl| R^{\sigma}_i \triangle S_U(u^{\sigma}_i) \bigr|
            \le 2C' \delta n.
        \]

        To complete the proof of the claim, it remains to bound the term
        $\bigl|R^{\circ} \triangle S_U(r)\bigr|$.
        Recall that $R^{\circ}$ represents the set of leaves that remain unassigned
        after the clustering procedures on both sides have completed.
        By definition,
        \[
            R^{\circ}
            = L \setminus \bigcup_{\sigma \in \{\Left, \Right\}}
            \bigcup_{i=1}^{k_{\sigma}} R^{\sigma}_i,
            \qquad\text{and}\qquad
            S_U(r)
            = L \setminus \bigcup_{\sigma \in \{\Left, \Right\}}
            \bigcup_{i=1}^{k_{\sigma}} S_U(u^{\sigma}_i).
        \]
        Hence,
        \begin{align*}
            \bigl| R^{\circ} \triangle S_U(r) \bigr|
            &=
            \Bigl|
            \Bigl( L \setminus \bigcup_{\sigma,i} R^{\sigma}_i \Bigr)
            \triangle
            \Bigl( L \setminus \bigcup_{\sigma,i} S_U(u^{\sigma}_i) \Bigr)
            \Bigr| =
            \Bigl|
            \Bigl( \bigcup_{\sigma,i} R^{\sigma}_i \Bigr)
            \triangle
            \Bigl(\bigcup_{\sigma,i} S_U(u^{\sigma}_i) \Bigr)
            \Bigr| 
            \\[4pt]
            &\leq
            \Bigl|
            \bigcup_{\sigma,i} \bigl( R^{\sigma}_i \triangle S_U(u^{\sigma}_i) \bigr)
            \Bigr|
            \le
            \sum_{\sigma,i}
            \bigl| R^{\sigma}_i \triangle S_U(u^{\sigma}_i) \bigr|
            \le 2C' \delta n.
        \end{align*}
        Therefore,
        \[
            \sum_{\sigma \in \{\Left, \Right\}}
            \sum_{i=1}^{k_{\sigma}}
            \bigl| R^{\sigma}_i \triangle S_U(u^{\sigma}_i) \bigr|
            + \bigl| R^{\circ} \triangle S_U(r) \bigr|
            \le C \delta n,
        \]
        for some universal constant $C > 0$,
        which completes the proof of the claim.
    \end{proof}
    This concludes the proof of Theorem~\ref{thm:clustering-correctness}.
\end{proof}

\subsection{Reconstruction Step}

We next describe how the algorithm reconstructs the tree from the clusters
obtained in the previous phase. At the end of the clustering step, the algorithm outputs the collections of sets
\[
    R^{\circ},R^{\Left}_1, \dots, R^{\Left}_{k_{\Left}}, \qquad
    R^{\Right}_1, \dots, R^{\Right}_{k_{\Right}}.
\]
If the clustering procedure has succeeded, then by
Theorem~\ref{thm:clustering-correctness}
there exist vertices
\[
    u^{\Left}_1, \dots, u^{\Left}_{k_{\Left}}, \;
    u^{\Right}_1, \dots, u^{\Right}_{k_{\Right}}
\]
in the ground-truth tree $T$ satisfying the bound~\eqref{eq:thm:clustering}.
Of course, the algorithm itself does not know whether the clustering succeeded;
in any case, it proceeds to reconstruct the tree based solely on the obtained clusters.

The reconstruction phase aims to infer the relative positions of the
vertices $u^{\sigma}_i$ within the underlying tree structure.
Let $U$ denote the collection of all such vertices together with the root $r$,
that is,
\[
    U = \{r \}
    \cup \{u^{\Left}_1, \dots, u^{\Left}_{k_{\Left}}\}
    \cup \{u^{\Right}_1, \dots, u^{\Right}_{k_{\Right}}\}.
\]

We define the \emph{skeleton tree} $T_U$ as the tree obtained from $T$
by iteratively removing all leaves not in $U$, and suppressing all vertices
not in $U$ that have exactly one child, until no such vertices remain.
The skeleton tree $T_U$ contains at most $2|U| \le 4\lceil 1/\delta \rceil$ vertices,
since every leaf of $T_U$ belongs to $U$, and every vertex not in $U$ has exactly two children,
which implies that the number of internal (non-$U$) vertices is at most $|U| - 2$.
Thus, the algorithm can guess the skeleton $T_U$ with a small but constant
(depending only on $\delta$) positive probability $p^{T_U}_{\delta}=\delta^{O(\frac{1}{\delta})}$.

% --- Original version ---

After guessing the skeleton $T_U$, the algorithm constructs a tree $\widehat{T}$ based on $T_U$ as follows. For every vertex $u_i^{\sigma} \in V(T_U)$, we create, in an arbitrary manner, a binary rooted tree whose leaves are exactly the vertices of the corresponding set $R_i^{\sigma}$. This tree is then attached as an additional (rightmost) child of $u_i^{\sigma}$.

To ensure that $\widehat{T}$ remains binary, we locally modify the structure at $u_i^{\sigma}$ if necessary. If after the attachment $u_i^{\sigma}$ has three children, we introduce a new vertex as the left child of $u_i^{\sigma}$, and reattach the original two children of $u_i^{\sigma}$ (i.e., those inherited from $T_U$) as children of this new vertex. Consequently, $u_i^{\sigma}$ has exactly two children: the newly created vertex and the root of the attached tree. See Figure~\ref{fig:reconstruction-step}.

Note that all sets $R_{j}^{\Left}$ are pairwise disjoint, and likewise all sets $R_{j}^{\Right}$ are pairwise disjoint.
However, a leaf may belong to both some $R_{j'}^{\Left}$ and $R_{j''}^{\Right}$.
In such cases, the leaf is included in only one of the corresponding subtrees -- either the one attached to $u_{j'}^{\Left}$ or to $u_{j''}^{\Right}$.
This assignment is made arbitrarily.

\begin{figure}[ht]
    \centering
    %
    % ---- Panel (a): Skeleton tree T_U with |C_u| annotations ----
    %
    \begin{subfigure}[c]{0.49\textwidth}
        \centering
        \begin{tikzpicture}[
            scale=0.75,
            dot/.style={circle,fill,inner sep=1.4pt},
            uset/.style={draw,fill=black,regular polygon,regular polygon sides=4,inner sep=2pt},
            cbox/.style={font=\tiny,text=gray!70,draw=gray!40,
                         rounded corners=1.5pt,inner sep=1.5pt,fill=white},
            >=stealth
        ]
        % ---- skeleton coordinates ----
        \coordinate (r)   at (0, 3);
        \coordinate (u2)  at (-2, 1.5);
        \coordinate (u1)  at (-3, 0);
        \coordinate (u3)  at (2, 1.5);
        
        % ---- skeleton edges ----
        \draw[line width=1.5pt] (r)  -- (u2);   % r -> u_2
        \draw[line width=1.5pt] (u2) -- (u1);   % u_2 -> u_1
        \draw[line width=1.5pt] (r)  -- (u3);   % r -> u_3
        
        % ---- skeleton nodes ----
        \node[uset, label={above:$r$}]    at (r)  {};
        \node[uset, label={left:$u^{\Left}_2$}]   at (u2) {};
        \node[uset, label={left:$u^{\Left}_1$}]   at (u1) {};
        \node[uset, label={right:$u^{\Right}_1$}]  at (u3) {};

        \end{tikzpicture}
        \caption{Guessed skeleton $T_U$.}
    \end{subfigure}
    \hfill
    %
    % ---- Panel (b): Reconstructed tree T-hat ----
    %
    \begin{subfigure}[c]{0.49\textwidth}
        \centering
        \begin{tikzpicture}[
            scale=0.75,
            dot/.style={circle,fill,inner sep=1.4pt},
            uset/.style={draw,fill=black,regular polygon,regular polygon sides=4,inner sep=2pt},
            hleaf/.style={circle,inner sep=2.2pt,thick},
            >=stealth
        ]
            % =============================================
            % Correct binary tree structure of T-hat:
            %   r: left child = w (Steiner), right child = rc_r (tree on R^circ)
            %   w: left child = u_2^L, right child = u_1^R
            %   u_2^L: left child = u_1^L, right child = rc2 (tree on R_2^L)
            %   u_1^L: no left child, right child = rcL1 (new Steiner)
            %     rcL1: children are xL, yL (reshuffled pairing on R_1^L)
            %   u_1^R: no left child, right child = rcR1 (new Steiner)
            %     rcR1: children are xR, yR (reshuffled pairing on R_1^R)
            % =============================================

            % ---- coordinates ----
            % Level 0: root
            \coordinate (r) at (0, 6.0);
            % Level 1: r's two children
            \coordinate (w) at (-1.5, 5.0);       % Steiner node (left child of r)
            \coordinate (rcr) at (2.0, 5.0);      % root of R^circ tree (right child of r)
            % Level 2: w's children and R^circ leaves
            \coordinate (u2) at (-3.0, 4.0);     % left child of w
            \coordinate (u3) at (0.0, 4.0);      % right child of w
            \coordinate (l11) at (1.7, 4.0);
            \coordinate (l12) at (2.3, 4.0);
            % Level 3: u_2^L's children and u_1^R's right child
            \coordinate (u1) at (-4.5, 3.0);     % left child of u_2^L
            \coordinate (rc2) at (-1.5, 3.0);     % right child of u_2^L (tree on R_2^L)
            \coordinate (rcR1) at (0.7, 3.0);     % right child of u_1^R (new Steiner, offset right)
            % Level 4: u_1^L's right child, rc2's children, rcR1's children
            \coordinate (rcL1) at (-3.5, 2.0);    % right child of u_1^L (new Steiner, offset right)
            \coordinate (l5)  at (-1.8, 2.0);
            \coordinate (l6)  at (-1.2, 2.0);
            \coordinate (xR) at (0.2, 2.0);       % left child of rcR1
            \coordinate (yR) at (1.2, 2.0);       % right child of rcR1
            % Level 5: rcL1's children (xL, yL), R_1^R leaves
            \coordinate (xL) at (-4.0, 1.0);      % left child of rcL1
            \coordinate (yL) at (-3.0, 1.0);      % right child of rcL1
            \coordinate (l7)  at (-0.1, 1.0);
            \coordinate (l9)  at (0.5, 1.0);
            \coordinate (l8)  at (0.9, 1.0);
            \coordinate (l10) at (1.5, 1.0);
            % Level 6: R_1^L leaves
            \coordinate (l1) at (-4.3, 0.0);
            \coordinate (l3) at (-3.7, 0.0);
            \coordinate (l2) at (-3.3, 0.0);
            \coordinate (l4) at (-2.7, 0.0);

            % ---- background regions (around each cluster's leaves) ----
            \begin{scope}[on background layer]
                % R_1^L: l1,l3,l2,l4 at y=0.0
                \fill[blue!12, rounded corners=5pt]
                    (-4.55, -0.3) rectangle (-2.45, 0.3);
                % R_2^L: l5,l6 at y=2.0
                \fill[olive!12, rounded corners=5pt]
                    (-2.05, 1.7) rectangle (-0.95, 2.3);
                % R_1^R: l7,l9,l8,l10 at y=1.0
                \fill[red!12, rounded corners=5pt]
                    (-0.35, 0.7) rectangle (1.75, 1.3);
                % R^circ: l11,l12 at y=4.0
                \fill[green!12, rounded corners=5pt]
                    (1.45, 3.7) rectangle (2.55, 4.3);
            \end{scope}

            % ---- skeleton backbone (bold) ----
            \draw[line width=1.5pt] (r) -- (w);
            \draw[line width=1.5pt] (w) -- (u2);
            \draw[line width=1.5pt] (w) -- (u3);
            \draw[line width=1.5pt] (u2) -- (u1);

            % ---- cluster subtree edges (normal weight) ----
            % R^circ at r (right child)
            \draw (r) -- (rcr);
            \draw (rcr) -- (l11);
            \draw (rcr) -- (l12);

            % R_2^L at u_2^L (right child)
            \draw (u2) -- (rc2);
            \draw (rc2) -- (l5);
            \draw (rc2) -- (l6);

            % u_1^L -> rcL1 (right child only, no left child)
            \draw (u1) -- (rcL1);
            % rcL1's children (reshuffled binary tree on R_1^L)
            \draw (rcL1) -- (xL);
            \draw (rcL1) -- (yL);
            \draw (xL) -- (l1);
            \draw (xL) -- (l3);
            \draw (yL) -- (l2);
            \draw (yL) -- (l4);

            % u_1^R -> rcR1 (right child only, no left child)
            \draw (u3) -- (rcR1);
            % rcR1's children (reshuffled binary tree on R_1^R)
            \draw (rcR1) -- (xR);
            \draw (rcR1) -- (yR);
            \draw (xR) -- (l7);
            \draw (xR) -- (l9);
            \draw (yR) -- (l8);
            \draw (yR) -- (l10);

            % ---- nodes ----
            % U-vertices (filled squares)
            \node[uset, label={above:$r$}, font=\small] at (r) {};
            \node[uset, label={[yshift=2pt]left:$u^{\Left}_2$}, font=\small] at (u2) {};
            \node[uset, label={left:$u^{\Left}_1$}, font=\small] at (u1) {};
            \node[uset, label={right:$u^{\Right}_1$}, font=\small] at (u3) {};
            % Steiner / internal nodes (filled circles)
            \foreach \p in {w,rcr,rc2,rcL1,rcR1,xL,yL,xR,yR}{\node[dot] at (\p) {};}

            % Cluster leaves
            \foreach \p in {l1,l2,l3,l4}{
                \node[hleaf, draw=blue!60!black, fill=blue!20] at (\p) {};
            }
            \foreach \p in {l5,l6}{
                \node[hleaf, draw=olive!60!black, fill=olive!20] at (\p) {};
            }
            \foreach \p in {l7,l8,l9,l10}{
                \node[hleaf, draw=red!60!black, fill=red!20] at (\p) {};
            }
            \foreach \p in {l11,l12}{
                \node[hleaf, draw=green!60!black, fill=green!20] at (\p) {};
            }

            % ---- R_u labels (braces) ----
            \draw[decorate, decoration={brace, amplitude=5pt, mirror}]
                (-4.4, -0.3) -- (-2.6, -0.3)
                node[midway, below=4pt, font=\small] {$R_1^{\Left}$};

            \draw[decorate, decoration={brace, amplitude=5pt, mirror}]
                (-1.9, 1.7) -- (-1.1, 1.7)
                node[midway, below=4pt, font=\small] {$R_2^{\Left}$};

            \draw[decorate, decoration={brace, amplitude=5pt, mirror}]
                (-0.2, 0.7) -- (1.6, 0.7)
                node[midway, below=4pt, font=\small] {$R_1^{\Right}$};

            \draw[decorate, decoration={brace, amplitude=5pt, mirror}]
                (1.6, 3.7) -- (2.4, 3.7)
                node[midway, below=4pt, font=\small] {$R^{\circ}$};

        \end{tikzpicture}
        \caption{Reconstructed tree $\widehat{T}$.}
    \end{subfigure}
    \caption{Illustration of the reconstruction step.
    (a)~With a small but constant positive probability $\delta^{O(1/\delta)}$, the algorithm correctly guesses the skeleton tree $T_U$ on $U = \{r, u_1^{\Left}, u_2^{\Left}, u_1^{\Right}\}$.
    (b)~For each vertex $u \in U$, the algorithm constructs an arbitrary binary tree on its cluster $R_u$, attaches it as an additional rightmost child of $u$, and resolves any duplicate leaf memberships arbitrarily.}
    \label{fig:reconstruction-step}
\end{figure}
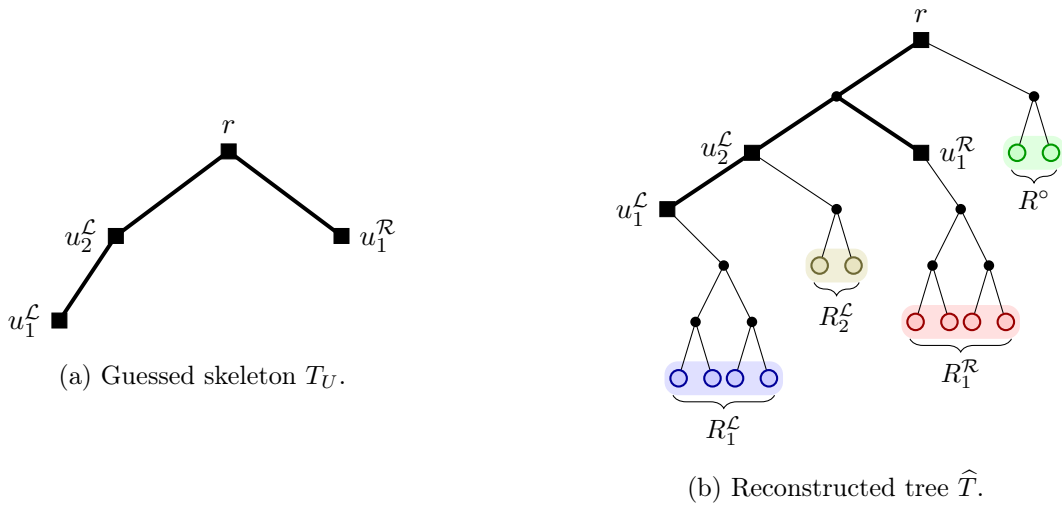

We now show that if the clustering phase succeeds and the reconstruction phase correctly guesses the skeleton,
then the resulting tree $\widehat{T}$ is a good approximation of the true tree $T$.
Specifically, for a randomly chosen 4-tuple $(a, b, c, d)$ of leaves,
the quartet relation $(a b \mid c d)$ takes the same value in $T$ and $\widehat{T}$
with probability $1 - O(\delta)$.
Equivalently, the quartet distance between $T$ and $\widehat{T}$ is $O(\delta)$.

\begin{theorem}
    \label{thm:reconstruction-accuracy}
    Suppose that the clustering phase succeeds and that the reconstruction phase
    correctly guesses the skeleton tree $T_U$.
    Let $\widehat{T}$ denote the reconstructed tree.
    Then, for a uniformly random 4-tuple $(a, b, c, d)$ of leaves,
    the quartet relation $(a b \mid c d)$ takes the same value in $T$ and $\widehat{T}$
    with probability $1 - O(\delta)$ (over the random choice of $(a, b, c, d)$).
\end{theorem}
\begin{proof}
    Consider a uniformly random 4-tuple $(a,b,c,d)$ of leaves.
    We define a \emph{bad event} $\mathcal{E}$ as follows.
    The event $\mathcal{E}$ occurs if either one of the vertices $a, b, c$, or $d$ lies in
    \begin{equation}
        \label{eq:union}
        \bigcup_{\sigma \in \{\Left, \Right\}}
        \bigcup_{i=1}^{k_{\sigma}} \bigl( R^{\sigma}_i \triangle S_U(u^{\sigma}_i) \bigr)
        \;\cup\;
        \bigl( R^{\circ} \triangle S_U(r) \bigr),
    \end{equation}
    or if there exists a distinct triplet $x,y,z \in \{a,b,c,d\}$ such that
    $\phi_U(x) = \phi_U(\lca(y,z))$.

    We now show that $\Pr(\mathcal{E}) = O(\delta)$.
    Indeed, by Theorem~\ref{thm:clustering-correctness} and our assumption that the clustering step succeeded,
    the size of the union in~\eqref{eq:union} is $O(\delta n)$.
    Consequently, the probability that any of $a,b,c,$ or $d$ lies in this set is $O(\delta)$.
    Moreover, since each set $R^{\sigma}_i$ has size at most $2\delta n$,
    every corresponding set $S_U(u^{\sigma}_i)$ also has size $O(\delta n)$.
    The same bound holds for $S_U(r)$.

    Now fix distinct $x,y,z \in \{a,b,c,d\}$.
    The vertex $\phi_U(\lca(y,z))$ equals some $u \in U$, and the event
    $\phi_U(x) = u$ occurs if and only if $x \in S_U(u)$.
    Thus,
    $$\Pr\Bigl(\phi_U(x) = \phi_U(\lca(y,z))\Bigr) = O(\delta).$$
    Applying the union bound over all $24$ distinct ordered triplets $(x,y,z)$
    from $\{a,b,c,d\}$ gives the claimed bound $\Pr(\mathcal{E}) = O(\delta)$.

We now show that if the bad event $\mathcal{E}$ does not occur, then the quartet relation $(ab \mid cd)$ agrees in both $T$ and $\widehat{T}$.
    Recall that for any four distinct leaves $a,b,c,d$, there is exactly one partition of $\{a,b,c,d\}$ into two unordered pairs $\{a,b\}$ and $\{c,d\}$ for which the relation $(ab \mid cd)$ holds. Without loss of generality, assume that $(ab \mid cd)$ is the quartet relation in the true tree $T$. We will show that the same relation holds in the reconstructed tree $\widehat{T}$.
    Equivalently, assuming that the paths from $a$ to $b$ and from $c$ to $d$ are vertex-disjoint in $T$, we prove that they are vertex-disjoint in $\widehat{T}$.

    We denote by $P(x,y)$ the path from $x$ to $y$ in the true tree $T$, and by $\widehat{P}(x,y)$ the corresponding path in the reconstructed tree $\widehat{T}$.
    Lemma~\ref{lem:disjoint-f} in Appendix~\ref{app:mapping-lemma} states that if
    (i) the paths $P(a,b)$ and $P(c,d)$ are vertex-disjoint in $T$, and
    (ii) for every distinct triple $x,y,z \in \{a,b,c,d\}$, we have $\phi_U(x) \ne \phi_U(\mathrm{LCA}(y,z))$,
    then the images $P(\phi_U(a),\phi_U(b))$ and $P(\phi_U(c),\phi_U(d))$ are also vertex-disjoint.
    Since both conditions (i) and (ii) hold in our setting (the former by assumption, and the latter because the bad event $\mathcal{E}$ does not occur),
    we conclude that the paths $P(\phi_U(a),\phi_U(b))$ and $P(\phi_U(c),\phi_U(d))$ in $T$ are vertex-disjoint.

    We now claim that the corresponding paths $\widehat{P}(\phi_U(a),\phi_U(b))$ and $\widehat{P}(\phi_U(c),\phi_U(d))$ in the reconstructed tree $\widehat{T}$
    are also vertex-disjoint.
    It suffices to show that the paths between $\phi_U(a)$ and $\phi_U(b)$ and between $\phi_U(c)$ and $\phi_U(d)$
    remain vertex-disjoint in the skeleton tree $T_U$.
    Indeed, the skeleton construction is obtained from $T$ by iteratively removing degree-one vertices
    and contracting vertices of degree two that do not belong to $U$.
    This process maps every path between two vertices $u', u'' \in U$ in $T$
    to a subpath of the original path connecting $u'$ and $u''$.
    Consequently, it cannot convert two vertex-disjoint paths in $T$
    into intersecting paths in $T_U$.

    Assume for the moment that the vertices $\phi_U(a)$, $\phi_U(b)$, $\phi_U(c)$, and $\phi_U(d)$ are all distinct.
    Consider the path $\widehat{P}(a,b)$ in $\widehat{T}$: it consists of three segments—one from $a$ up to $\phi_U(a)$,
    one along $\widehat{P}(\phi_U(a),\phi_U(b))$ connecting the corresponding cluster representatives in the skeleton,
    and one from $\phi_U(b)$ down to $b$.
    The segments from $a,b,c,d$ to $\phi_U(a)$, $\phi_U(b)$, $\phi_U(c)$, and $\phi_U(d)$, respectively, are pairwise vertex-disjoint,
    since each such path lies entirely within the distinct subtree attached to the right child of its corresponding $\phi_U(\cdot)$.
    Furthermore, the paths $\widehat{P}(\phi_U(a),\phi_U(b))$ and $\widehat{P}(\phi_U(c),\phi_U(d))$ in the skeleton do not intersect,
    as established above.
    Consequently, the full paths $\widehat{P}(a,b)$ and $\widehat{P}(c,d)$ are vertex-disjoint,
    and hence the quartet relation $(ab\mid cd)$ also holds in the reconstructed tree~$\widehat{T}$.

It remains to consider the case where some of the vertices in $\{a,b,c,d\}$ are mapped by $\phi_U$ to the same vertex $u \in U$.  
Since the paths $P(\phi_U(a),\phi_U(b))$ and $P(\phi_U(c),\phi_U(d))$ are vertex-disjoint, we must have 
$\phi_U(a)\neq \phi_U(c)$, $\phi_U(a)\neq \phi_U(d)$, $\phi_U(b)\neq \phi_U(c)$, and $\phi_U(b)\neq \phi_U(d)$. 
Thus, the only possibilities are that $\phi_U(a)=\phi_U(b)=u$ or $\phi_U(c)=\phi_U(d)=u$.  
Assume without loss of generality that $\phi_U(a)=\phi_U(b)=u$. Then $\phi_U(c)\neq u$ and $\phi_U(d)\neq u$. 
The path $\widehat{P}(a,b)$ lies entirely within the subtree attached to the right child of $u$ in $\widehat{T}$. 
Since neither $c$ nor $d$ lies in this subtree, the path $\widehat{P}(c,d)$ does not enter it. 
Consequently, $\widehat{P}(c,d)$ does not intersect $\widehat{P}(a,b)$. 
\end{proof}

\subsection{Verification Step}
So far, we have described the \emph{clustering} and \emph{reconstruction} phases of the algorithm and established that each of these phases succeeds with a small but constant probability depending only on~$\delta$. In particular, the two phases succeed with probability at least $\delta^{O(\frac{1}{\delta})}$.
Conditioned on their success, the resulting tree~$\widehat{T}$ is a good approximation of the true tree~$T$:
specifically, by Theorem~\ref{thm:reconstruction-accuracy}, the quartet distance between~$T$ and~$\widehat{T}$ is~$O(\delta)$. We set $\delta$ to be a sufficiently small constant, $\delta = \Theta(\epsilon)$, so that the quartet distance between~$T$ and~$\widehat{T}$ is at most~$\epsilon/2$ whenever the \emph{clustering} and \emph{reconstruction} phases succeed.

However, the algorithm itself does not know whether the clustering and reconstruction steps have succeeded in a particular run.
To overcome this, the algorithm repeats both phases a sufficiently large number of times (say, $\frac{1}{\delta}^{O(\frac{1}{\delta})}\cdot n$ times),
so that the probability that all attempts to cluster the leaves and reconstruct the tree fail becomes exponentially small.
For each reconstructed tree~$\widehat{T}$, it invokes a \emph{verification step} that estimates the quartet distance between~$T$ and~$\widehat{T}$ using the quartet dataset $Q$, i.e., it calculates the fraction of quartets in $Q $ that tree $\widehat{T}$ fails to satisfy \footnote{Strictly speaking, the algorithm estimates the probability that the quartet label of a sample from $\mathcal{RCN}(\eta)$ is different from the quartet label according to $\widehat{T}$. For example, this probability is $\eta$ for the ground truth tree $T$.}. We denote this empirical estimate by~$\widehat{\mathcal{R}}_{Q}(\widehat{T})$. If this estimate is at most a fixed acceptance threshold $\eta +\frac{3}{4}\eps(1-\frac{3\eta}{2})$, the verifier declares success, and the algorithm outputs the corresponding~$\widehat{T}$ as the final reconstruction.
Otherwise, the algorithm discards the candidate tree and repeats the procedure.

There are two possible scenarios under which the algorithm may fail.
First, with a small probability, the algorithm may fail to obtain a verified tree, in which case it terminates and reports failure.
Second, the verifier may erroneously accept an inaccurate reconstruction.
Observe that the event that the algorithm fails is a subset of  $\mathcal{A}\cup \mathcal{B}$, where $\mathcal{A}$ and $\mathcal{B}$ are defined as follows.
Event $\mathcal{A}$ is the event that the algorithm fails to produce a tree of small distance, at most $\frac{\eps}{2}$, to the ground truth tree after running the \emph{clustering} and \emph{reconstruction} step a polynomial number of times. Note that the probability of that event is $o(1)$.
Event $\mathcal{B}$ is the event that there exists a tree $T'\in \mathcal{T}$ on which the verification step fails. We say that the verification step fails on tree $T'$ if (i) $\dist(T, T')\le\frac{\eps}{2}$ and $\hat{\mathcal{R}}_Q(T')> \eta +\frac{3}{4}\eps(1-\frac{3\eta}{2})$ (the tree is close to $T$ but the estimate is above the threshold) or (ii) $\dist(T,T')>\eps$ and $\hat{\mathcal{R}}_Q(T')\leq  \eta +\frac{3}{4}\eps(1-\frac{3\eta}{2})$ (the tree is far from $T$ but the estimate is at most the threshold). We prove the following theorem, which by the union bound implies the main result, Theorem~\ref{th:ptas}.
\begin{theorem}[Verification guarantee]\label{thm:verification}
    There exists an absolute constant $C$ such that if $Q$ is a quartet dataset with i.i.d. samples from $\mathcal{RCN}(\eta)$ model with $|Q|\geq C\frac{n}{(\eps(1-\frac{3\eta}{2}))^2}$ then the probability, over the randomness in $Q$, that there exists a tree $T'$ for which the verification procedure fails is $o(1)$.
\end{theorem}
\begin{proof}
    Let $\mathcal{R}_{\eta}(T')$ be the probability that tree $T'$ fails to classify correctly a quartet from the $\mathcal{RCN}(\eta)$ model. We can relate this to the quartet distance from the ground truth tree $T$: $\mathcal{R}_{\eta}(T')=(1-\frac{3\eta}{2})\dist(T,T')+\eta$.
    Applying Theorem \ref{thm:natarajan_quart} and the Multiclass Fundamental Theorem (Theorem 29.3 in \cite{shalev2014understanding}), we have that there exist a $C$ such that for every error parameter $\varepsilon>0$, if $Q$ is a set of i.i.d. samples from $\mathcal{RCN}(\eta)$ model with $|Q|\geq C\frac{n}{\varepsilon^2} $ then with probability $1-o(1)$, for every $T'\in \mathcal{T}$: $|\hat{\mathcal{R}}_{Q}(T')-\mathcal{R}(T')|\leq \varepsilon$. We apply this with error parameter $\frac{\eps}{4}\left(1-\frac{3\eta}{2}\right)$ and we have that with probability at least $1-o(1)$, for every $T'\in \mathcal{T}$: $|\hat{\mathcal{R}}_{Q}(T')-\mathcal{R}_{\eta}(T')|\leq \frac{\eps}{4}\left(1-\frac{3\eta}{2}\right)$.
    We show that if this event happens then the verification procedure succeeds on every $T'\in \mathcal{T}$. We divide between cases for $T'$. Recall that the verification step can fail on $T'$ only if $\dist(T,T')\leq \frac{\eps}{2}$ or if $\dist(T,T')>\eps$. Thus, the verification step trivially succeeds on $T'$ if $\frac{\eps}{2}<\dist(T,T')\leq \eps$. If $\dist(T,T')> \eps$ then we have that:
    \[
        \hat{\mathcal{R}}_{Q}(T')\geq \mathcal{R}_{\eta}(T')-\frac{\eps}{4}\left(1-\frac{3\eta}{2}\right)
        >\frac{3\eps}{4}\left(1-\frac{3\eta}{2}\right)+\eta,
    \]
    and the verification procedure succeeds. On the other hand, if $\dist(T,T')\leq \frac{\eps}{2}$ then we have:
    \[
        \hat{\mathcal{R}}_{Q}(T')\leq \mathcal{R}_{\eta}(T')+\left(1-\frac{3\eta}{2}\right)\frac{\eps}{4}
        \leq \frac{3\eps}{4}\left(1-\frac{3\eta}{2}\right)+\eta,
    \]
    and the verification procedure would succeed on $T'$.
\end{proof}

\subsection{Runtime of the Algorithm}
For the runtime of the algorithm, we observe that both the \emph{clustering} and \emph{reconstruction} procedures run in time polynomial in $n$, and the same of course holds for the \emph{verification} procedure. The algorithm then repeats all three of the procedures $\frac{1}{\delta}^{O(\frac{1}{\delta})}\cdot n $ times, thus yielding a runtime of 
\[O\left(\frac{1}{\delta}^{O(\frac{1}{\delta})}\poly(n)\right)=O\left(\frac{1}{\epsilon}^{O(\frac{1}{\epsilon})}\cdot \poly(n)\right).\]
\section{QED Procedure and Analysis}\label{sec:community-detection}

We now present the SBM-inspired algorithm that recovers the leaves corresponding to a subtree $T_u$ of the tree $T$. At a high level, the algorithm iteratively recovers such sets to reconstruct a tree close to the ground truth. To achieve this, it constructs a graph $G'(L',E')$ using the quartet dataset, $L'$ here corresponds to the set of active leaves. The algorithm then uses a semidefinite program to approximately recover the subset of $L'$ that are leaves of that subtree. The size of the recovered subtree is approximately $\delta n$, where $\delta$ is a small constant dependent on the accuracy parameter $\varepsilon$. Our algorithm recovers the intended set up to an error of $O(\delta^2 n)$. We introduce some notation that we use throughout this section.

\paragraph{Notation.} For $a,b\in L$, $\mathcal{Q}_{ab}$ denotes the quartet constraints satisfied by $T$ that place $a$ and $b$ together. $\mathcal{Q}_{L'}$ is used for the subset of $\mathcal{Q}$ that only involves leaves from $L'$, and $\mathcal{Q}_{ab|L'}$ is the subset of $\mathcal{Q}_{ab}$ that only contains elements of $L'$.  Finally, as mentioned in the Technical Overview even though the ground-truth tree $T$ is initially unrooted, we can introduce a root $r$ by subdividing one of the edges of $T$ so that the two subtrees adjacent to $r$ contain between one-third and two-thirds of the leaves. From this point onward, we view $T$ as a rooted tree, and we use $\mathcal{L}$ and $\mathcal{R}$ to denote the leaves of the left and right subtrees of $T$. For vertices $u,v\in L'$ we use $\delta_{uv|L'}$ to denote $|L_{uv}\cap L'|/|L'|$ and for a set $Z\subseteq L'$ we use $\delta_{Z|L'}=|L'\cap L_{\lca(Z)}|/|L'|$, where $\lca(Z)$ is the lowest common ancestor of all nodes in a set $Z$.

\subsection{The \texorpdfstring{$Q(n,\lambda,\eta)$}{Q(n,λ,η)}-model}
\label{sec:Qnpeta definition}
In this section we give a slightly different random model for the problem. A set of quartets $Q\sim Q(n, \lambda, \eta)$ is generated according to the following procedure. For every set $\set{a,b,c,d}\subseteq \binom{L}{4}$, the number of times a quartet on these leaves is included, is determined by an independent Poisson random variable with parameter $\lambda$. Each quartet involving set $\set{a,b,c,d}$ is labeled as $T(a,b,c,d)$ with probability $1-\eta$ and has a random incorrect label with the remaining probability. Note that similarly to the RCN model, the generated set $Q$ could in fact be a multi-set, this is not an issue for us but also the probability of that event happening is $o(1)$ for the sparse parameter regime that is most interesting to us\footnote{If we are in a denser regime, the algorithm can always subsample so that the input is sparse as studied here.}.  The RCN model we are considering corresponds to $Q(n,\lambda ,\eta)$ with $\lambda=\frac{c'}{n^3}$. It is not hard to show that the two models are equivalent, in the sense that an algorithm can use an instance from one to simulate an instance from the other (up to $o(1)$ total variation distance).

Here, for completeness, we describe how an algorithm can use samples from the $\mathcal{RCN}(\eta)$ model to simulate an instance from the $Q(n,\lambda, \eta)$ model. Let for convenience, $\bar{m}={\binom{n}{4}}\cdot\lambda$ (the expected number of quartets in an instance generated from $Q(n,\lambda, \eta)$. We draw $2\bar{m}$ samples from the $\mathcal{RCN}(\eta)$ model and generate a Poisson random variable $X$ with parameter $\lambda \cdot{\binom{n}{4}}$. Note that by Chernoff bound (e.g. Exercise 2.3.5 in \cite{vershynin}), with probability $1-o(1)$, $X\leq 2\bar{m}$. If that is the case then we sample random $X$ constraints among $2\bar{m}$ and provide them to the algorithm, otherwise we give all $2\bar{m}$ constraints to the algorithm. Using Poisson conditioning (e.g. Theorem 3.7.8 in \cite{durrett2019probability}) we get that the quartet sample set generated from the above procedure has $o(1)$ total variation distance from an instance generated according to $Q(n,\lambda, \eta)$.

\subsection{Constructing a Graph}\label{sec:algo_graph}

In this section, we describe how the algorithm constructs a multi-graph $G'=(L',E')$ from an instance $Q$ of quartets drawn from the $Q(n, \frac{c'}{n^3}, \eta)$ model described in \Cref{sec:Qnpeta definition}.
The vertex set is a set of leaves, $L'\subseteq L$, has size $|L'| = \gamma n$, and, roughly speaking, corresponds to leaves that have not yet been recovered by the algorithm, i.e. they are still active. Throughout, we maintain the invariant that $\gamma \geq \frac{1}{4}$.

\paragraph{Construction of the graph.}The construction of $G'$ \iffalse\dnotem{We can construct the graph upfront and then remove the edges that correspond to leaves that are no longer relevant} \fi is straightforward: For every $ab|cd\in Q$ with $\set{a,b,c,d}\subseteq L'$, we add an edge  connecting $a$ and $b$ with probability $1/2$ and an edge connecting $c,d$ otherwise. Notice that the algorithm can construct the graph once and then remove edges coming from constraints that are no longer relevant. In fact, we can construct an initial graph $G$, using all quartets. The graph $G'$ that corresponds to a set of leaves $L'$ can be derived from $G$ by removing all edges coming from a quartet involving at least one item that is not in $L'$.  The following lemma shows that the graph $G'$ constructed has independent edges and gives estimates for the probability of an edge between $a$ and $b$ appearing in terms of the relative position of $a$ and $b$ in the tree.
\begin{lemma}
    \label{lem:edge_prob2}
    The distribution over graphs generated from the above procedure has the following properties.
    \begin{enumerate}
        \item Each edge is included in the graph independently
        \item For $c=\frac{\gamma^2 c'}{4}$, up to $o\left(\frac{1}{n}\right)$ terms:\footnote{We say that $x\leq y$ up to $o(\frac{1}{n})$ terms if $x\leq y+o(\frac{1}{n})$}
            \begin{equation}
                \label{eq:generic_lower_bounds_for_edge_prob}
                \frac{c}{n}\left(\left(1-2\delta_{ab|L'}+\delta^2_{ab|L'}\right)\left(1-\frac{3\eta}{2}\right)+ \frac{\eta}{2}\right)\leq \prob{(a, b)\in E'}.
            \end{equation}
            and also (again up to $o\left(\frac{1}{n}\right)$ terms):
            \begin{equation}
                \label{eq:generic_upperr_bounds_for_edge_prob}
                \prob{(a, b)\in E'}\leq \frac{c}{n}\left(\left(1-2\delta_{ab|L'}+2\delta^2_{ab|L'}\right)\left(1-\frac{3\eta}{2}\right)+ \frac{\eta}{2}\right).
            \end{equation}
        \item If $a\in \mathcal{L}\cap L'$ and $b\in \mathcal{R}\cap L'$ then for $\alpha=\frac{|\mathcal{L}\cap L'|}{|L'|}$ and $\beta=\frac{|\mathcal{R}\cap L'|}{|L'|}$ then up to $o\left(\frac{1}{n}\right)$ terms:
            \begin{equation}
                \label{eq:edge_prob_root_general}
                \prob{(a,b)\in E'}\leq \frac{c}{n}\left((1-2\alpha\beta)\left(1-\frac{3\eta}{2}\right)+\frac{\eta}{2}\right).
            \end{equation}
    \end{enumerate}
\end{lemma}

At a high level, the second item of Lemma \ref{lem:edge_prob2} states that while $\delta_{ab|L'}$ is small the probability of an edge between $a$ and $b$ decreases as the size of the subtree rooted at their lowest common ancestor increases. Note however that this does not hold when $\delta_{ab|L'}$ is large and in fact the upper bound can be very loose in that case. This is especially problematic when the lowest common ancestor of $a$ and $b$ is the root. We complement our bound with \Cref{eq:edge_prob_root_general}. We prove Lemma \ref{lem:edge_prob2} in \Cref{app:distribution of graph}.

\subsection{The SDP }\label{sec:algo_sdp}
In this section, we describe the SDP  used by the algorithm to recover a set $S$ that corresponds to a subtree of the current planted tree $T'$, where $T'$ is the irreducible tree with leaves $L'\subseteq L $. 
We construct the graph $G'=(L',E')$ as described in the previous section, by removing from graph $G$ the edges that correspond to quartets containing at least one element not in $L'$.

We will fix a set $S$ of size $\delta n\leq |S|\leq 2\delta n$ and show that the algorithm recovers a set $\hat{S}$
that is close to $S$ up to error $O(\delta^2n) $ with constant probability at least $p_\delta={\poly(\delta)}$.
We assume that we know the size of $|S|$ up to an additive error of $\poly(\delta)\cdot n $, by successfully guessing it with probability ${\poly(\delta)}$. In other words we assume that we have some estimate $\hat{\delta}_{S}$ (selected from a grid of size $\poly(\frac{1}{\delta})$ and step $\poly(\delta)$) of the value $\delta_{S|L'}=\frac{|S|}{|L'|}$ such that $\delta_{S|L'}\leq \hat{\delta}_{S}$ and $\deltaest-\delta_{S|L'}\ll \delta_{S|L'}^3$.
We can assume, without loss of generality, that $S\subseteq \mathcal{L}$. We furthermore assume that $|\mathcal{L}\cap L'|\geq 48\delta n$ as well as that $|L'\cap \mathcal{R}|\geq \frac{n}{4} $. Note that the last inequality gives us that $\delta\leq\delta_{S|L'}\leq 8\delta$, where $\delta$ and thus $\delta_{S|L'}$ and $\deltaest$ are small enough constants.

We will assume that $n$ is large enough so that the $o(\frac{1}{n})$ terms in Lemma \ref{lem:edge_prob2} are bounded in absolute value by some small enough error term. In particular, we let $\err\ll (1-\frac{3\eta}{2})\delta_{S|L'}^3$ and define functions $\rho^{+}(x)$ and $\rho^{-}(x)$ as follows.
We let $\rho^{+}(x)=\frac{c}{n}\left(\left(1-2x+ 2 x^2\right)(1-\frac{3\eta}{2})+\frac{\eta}{2}+ \err \right)$ and $\rho^{-}(x)=\frac{c}{n}\left(\left(1-2x+  x^2\right)(1-\frac{3\eta}{2})+\frac{\eta}{2}- \err \right))$. We assume that $n$ is large enough so that by Lemma \ref{lem:edge_prob2}, for $a,b\in L'$, we have:
\begin{equation}   \label{eq:prob_in_terms_of_rho}
    \rho^{-}(\delta_{ab|L'})\leq \prob{(a,b)\in G'}\leq \rho^{+}(\delta_{ab|L'}).
\end{equation}

Observe that for small values of $x$, $\rho^{-}(x)\approx\rho^{+}(x)$, so for a high level understanding of the analysis, one can think of these two values applied on $\delta_{ab|L'}$ as being equal, and characterizing the probability that an edge between $a$ and $b$ appears in the graph.

Our solution will try to exploit the fact that the probability that two vertices in $S$ are connected is higher than the probability that a vertex in $S$ is connected with a vertex outside of $S$. In other words, the graph has structure similar to that of the stochastic block model.

We can now describe the SDP. In the SDP  there is one vector for every vertex in $G'$.
Our goal is to show that the vectors that correspond to set $S$ are well clustered together and far from most other vectors. We consider the following program:
\begin{align}
    &\text{max }\quad \sum_{(u,v)\in L'\times L'}\left(\ind{(u,v)\in E'}-q\right)\langle\bar{u}, \bar{v}\rangle\label{eq:sdpobj}\\
    &\text{subject to:}\notag\\
    &\quad \quad\quad\quad\langle\bar{u}, \bar{v}\rangle\geq 0 \hspace{1.9cm}\text{for all } u,v\in L'\label{eq:sdp1}\\
    &\quad \quad\quad\quad \|\bar{u}\|^2=1 \hspace{1.9cm}\text{ for all }u\in L'\label{eq:sdp2}\\
    & \quad \quad\quad\quad \sum_{v\in L'}\langle\bar{u},\bar{v}\rangle\leq \hat{\delta}_{S}\cdot |L'| \hspace{0.63cm}\text{ for all } u\in L'.\label{eq:sdp3}
\end{align}

Where we want $q$ to be smaller than the edge probability between any two vertices in $S$ by a small margin $\Theta(\deltaest^2)$ (see Lemma \ref{lem:entries_of_exp_for_vertices_in_S}). In particular, we let $q=\rho^{-}\left(\hat{\delta}_{S}+4\hat{\delta}_{S}^2\right)$, recall that $\rho^{-}(\delta_{S|L'})$ is a lower bound for the edge probabilities in $S$. We let $A\in \mathbb{R}^{|L'|\times |L'|}$ be a matrix such that $A_{u,v}=\ind{(u,v)\in E'}-q$. Using this notation, the objective function of the SDP can also be written as $\sum_{u,v\in L'\times L'}A_{u,v}\inner{\bar{u},\bar{v}}$. Notice that $A\in \mathbb{R}^{|L'|\times |L'|}$ is a random matrix, the randomness of $A$ comes from the randomness in $G$ which in turn comes from the randomness in the instance but also from the randomness of the algorithm in the selection of edges. We will use $\bar{A}\in \mathbb{R}^{|L'|\times |L'|}$ to denote the expected value of $A$.
Notice that $\bar{A}$ is a deterministic matrix, parameterized by the parameters of the model, $c'$ and $\eta$.
\subsection{Concentration: Expected Instance vs. Actual Instance}\label{sec:algo_grothendieck}
In this section we will use the randomness of the quartet set $Q$ to prove structural properties of the SDP. We want to use the same set of quartet constraints in every call of the QED recovery procedure to ensure polynomial in $1/\epsilon$ samples.
To do that, we will prove structural properties for all possible SDPs that the algorithm could encounter.
Observe that we get a different semidefinite program for every different valid set $L'$ and estimate $\deltaest$.
There are at most $2^n$ valid options for the set $L'$ while the different options for $\deltaest$ are bounded by $\poly({1}/{\delta})$. We show that with high probability over the randomness in $Q$ and the internal randomness of the algorithm, for all options of $L'$ and $\deltaest$, the value of any solution to the SDP on the observed (random) instance $A$ is close to its value on the deterministic ``expected'' instance $\bar{A}$. This will allow us to prove structural properties of the SDP by working with $\bar{A}$. Our analysis follows ideas from \cite{guedon2016community}. 

Recall here that $G$ is the multi-graph constructed from the entire set of quartets $Q$ where for every $ab|cd\in Q$ an edge is added between $a$ and $b$ with probability $1/2$ and between $c$ and $d$ with the remaining probability. To get graph $G'=(L', E')$, we remove every edge coming from a quartet involving at least one item not in $L'$. 

\begin{lemma}
    \label{lem:Actual_to_expected_bound}
    With probability at least $1-\poly(1/\delta)\exp(-1.5n)=1-o(1)$ over the randomness in $G$ which comes from the randomness in the algorithm and the randomness in $Q\sim Q(n,\frac{c'}{n^3}, \eta)$, the following holds. For every valid $L'$, for every valid $\deltaest$ and for every $\set{\bar{u}}_{u\in L'}$ satisfying constraints \ref{eq:sdp1}, \ref{eq:sdp2} and \ref{eq:sdp3} parameterized by $L'$ and $\deltaest$ (recall $q=\rho^{-}(\deltaest+4\deltaest^2)$):
    \begin{equation}            
        \left|\sum_{u,v\in L'}A^{L',\deltaest}_{u,v}\inner{\bar{u},\bar{v}}-\sum_{u,v\in L'}\bar{A}^{L',\deltaest}_{u,v}\inner{\bar{u},\bar{v}}\right|\leq 10 K_{G}\sqrt{c'}n,\label{eq:good_quart_condition}
    \end{equation}
    where $K_G\le 1.783$ is the Grothendieck constant and   $A^{L',\deltaest}$ is used to denote the matrix of the SDP parameterized by $L'$ and $\deltaest$.

\end{lemma}
\begin{proof} Fix $L'$ and $\deltaest$, $L'$ is a subset of $L$ with $|L'|=\gamma n$ and $\deltaest$ is in the grid of $\poly(\delta)$ possible options.
    For convenience we let $\Delta A_{u,v}=A_{u,v}-\bar{A}_{u,v}=\ind{(u,v)\in E'}-\prob{(u,v)\in E'}$. We want to bound the probability of the event 
    \[
    \mathcal{A}=\set{\exists \set{\bar{u}}_{u\in V}:\left|\sum_{u,v}\Delta A_{u,v}\inner{\bar{u},\bar{v}}\right |>10K_G\sqrt{c'}n }.
    \]
    By Grothendieck's inequality (Theorem \ref{thm: Grothendieck}) we get that the event $\mathcal{A}$ is a subset of the following event:
    \[
        \mathcal{B}=\set{\exists x,y\in \set{-1,1}^{\gamma n}:\left| \sum_{u,v}\Delta A_{u,v}x_uy_v\right |>10\sqrt{c'}n}.
    \]
    We now focus on bounding the probability of event $\mathcal{B}$. We fix $x,y$ and bound the probability that $\left|\sum_{u,v}\Delta A_{u,v}x_uy_v\right|> 10\sqrt{c'}n $. We will then take a union bound over all $x,y\in \set{-1,1}^{\gamma n}$. We write $\sum_{u,v} \Delta A_{u,v}x_uy_v=\sum_{u<v}\Delta A_{u,v}(x_uy_v+x_vy_u)$ and observe that this is a sum of $\binom{\gamma n}{2}$ independent random variables. We will apply Bernstein's inequality to this sum. For the expected value, we have that:
    \begin{align*}
        \Exp{\Delta A_{u,v}(x_uy_v+x_vy_u)}&=(x_uy_v+x_vy_u)\Exp{\Delta A_{u,v}}\\
        &=(x_uy_v+x_vy_u)\Exp{\ind{(u,v)\in E}-\prob{(u,v)\in E}}\\
        &=0.
    \end{align*}
    For the variance we have that:
    \begin{align*}
        \Var{\Delta A_{u,v}(x_uy_v+x_vy_u)}&=\left(x_uy_v+x_vy_u\right)^2\Var{\Delta A_{u,v}}\\
        &\leq 4\cdot \prob{(u,v)\in E'}\\
        &\leq \frac{4c'}{n},
    \end{align*}
    where we have used that $\Var{\Delta A_{u,v}}=\Var{\ind{(u,v)\in E'}}\leq \prob{(u,v)\in E'}$. We finally observe that $\left|\Delta A_{u,v}(x_uy_v+x_vy_u)\right|\leq 2$. We can now apply Bernstein's inequality, \Cref{thm:Bernstein}:
    \[
        \prob{\left|\sum_{u,v}\Delta A_{u,v}x_uy_v\right|\geq 10\sqrt{c'}n}\leq 2\exp\left(-\frac{50c'n}{4c'+20\sqrt{c'}/3}\right)
        \leq 2\exp\left(-4.5n \right),
    \]
    using that $c'\geq 1$. We now apply the union bound over all $x,y\in \set{-1,1}^{\gamma n}$ and that $\gamma\leq 1$:
    \[
        \prob{\exists x,y\in \set{-1,1}^{n}: \sum_{u,v}\Delta A_{u,v}x_uy_v>10\sqrt{c}n}\leq 2^{2n}2\exp\left(-4.5n\right).
    \]
    Taking a union bound over all $L'$ and $\deltaest$ we have that with probability at least $$1-\poly({1}/{\delta})2^{3n}\exp(-4.5n)\geq1-\poly(1/\delta)\exp(-1.5n)=1-o(1),$$ \Cref{eq:good_quart_condition} holds for all $L'$ and $q$.
\end{proof}
We say that a graph $G$ is good if it satisfies the condition of Lemma \ref{lem:Actual_to_expected_bound}. Notice that we don't ask for $G$ to be generated from the model of section \cref{sec:Qnpeta definition} in order for the graph constructed to be good, the parameters of the model, $c'$ and $\eta$  are implicit in the definition of matrices $\bar{A}^{L',q}$ which are deterministic matrices. To make this explicit, we will say that a graph $G$ is $(c',\eta)$-good.
%Thus, an (inefficient) algorithm can check, given enough time, whether a quartet instance is $(c',\eta, T)$-good without any knowledge of where the instance came from, by constructing every matrix $A^{L',q}$ and verifying that $ \left|\sum_{u,v\in L'}A^{L',q}_{u,v}\inner{\bar{u},\bar{v}}-\sum_{u,v\in L'}\bar{A}^{L',q}_{u,v}\inner{\bar{u},\bar{v}}\right|\leq 10 K_{G}\sqrt{c'}n$ for all SDP solutions. 
A graph $G$ generated from quartet set $Q\sim Q(n, \frac{c'}{n^3}, \eta)$ according to the procedure described above is $(c',\eta)$-good with very high probability and for the remaining of this section we will condition on that event.
\subsection{QED Description and Theorem}
In this section we describe in detail the QED procedure and state the main guarantee of the procedure. The algorithm is given a quartet instance $Q$, a subset of the leaves $L'$ and an estimate $\deltaest$ for $\frac{|S|}{|L'|}$ (which is successfully guessed with probability $\poly(\delta)$). The algorithm constructs the graph described in \Cref{sec:algo_graph} and solves the corresponding semidefinite program to obtain vectors $\set{\bar{v}}_{v\in L'}$. Finally the algorithm selects a vector $\bar{u}$ uniformly at random in $\set{\bar{v}}_{v\in L'}$ and returns $\hat{S}=\set{v\in L': \|\bar{u}-\bar{v}\|^2\leq 1}$. See Algorithm \ref{alg:recover_S} for pseudocode.
\begin{algorithm}[ht]
\caption{$\textsc{QED}(Q, L', \deltaest)$}
\label{alg:recover_S}
    \Input{Quartet instance: $Q$, subset of leaves: $L'\subseteq L$, estimate $\deltaest$ (the algorithm guesses this up to $\poly(\delta)$ error)}
    \Output{Estimate set $\hat{S}$}
    Construct the graph $G'$ and solve the SDP (Eq.~\ref{eq:sdpobj}, \ref{eq:sdp1}, \ref{eq:sdp2}, \ref{eq:sdp3}, see \Cref{sec:algo_graph}, ~\ref{sec:algo_sdp})\\

    Let $u$ be a uniformly random vertex in $L'$\\
    Let $\hat{S}=\set{v\in L': \|\bar{u}-\bar{v}\|^2\leq 1}$\\
    Return $\hat{S}$.
\end{algorithm}
The procedure enjoys the following guarantee, which roughly speaking says that if the graph $G$ is good then the algorithm recovers the set with $\Omega(\delta)$ probability.
\begin{theorem}[QED-guarantee]
\label{thm:qed-guarantee} 
    Let $G$ be a $(c',\eta)$-good graph (see \Cref{sec:algo_grothendieck}) with 
\[
c' \geq \left(\frac{40960 \cdot K_G}{\left(1-\frac{3\eta}{2}\right)\delta^5}\right)^2.
\]
Let $L' \subseteq L$ be a subset of the leaves, and let $S$ be a set satisfying $\delta n \leq |S| \leq 2\delta n$, where $\delta$ is smaller than a positive absolute constant $\delta^*$.
Assume that there exists a vertex $u \in T$ such that $S = L_u \cap L'$.
Furthermore, assume that $|L'\cap \mathcal{L}|\geq 48\delta n$ and that $|L'\cap \mathcal{R}|\geq \frac{n}{4}$ as well as that $\delta_{S|L'}=\frac{|S|}{|L'|}\leq \deltaest\leq \delta_{S|L'}+\poly(\delta)$ for a small enough polynomial in $\delta$. Then, for sufficiently large~$n$, Algorithm~\ref{alg:recover_S} returns, with probability at least $\frac{\delta_{S|L'}}{2} \geq \frac{\delta}{2}$, a set $\hat{S}$ such that
    \[
        |S\triangle \hat{S}|=O(\delta^2 n).
    \]
\end{theorem}
We prove the theorem in Section \ref{sec:algo_recover_S}. The algorithm  ``guesses" a successful estimate $\deltaest$ of $\delta_{S|L'}$ ($\delta_{S|L'}\leq \deltaest\leq \delta_{S|L'}+\poly(\delta)$) with probability ${\poly(\delta)}$ meaning that the success probability of QED procedure is $\Omega(\poly(\delta))$.

\subsection{A Probability Gap on the Expected Instance}\label{sec:algo_in_out_gap}

By Lemma \ref{lem:Actual_to_expected_bound}, with high probability, the value of every solution on the actual, sampled instance described by matrix $A$, will be close to its value on the ``expected'' instance, described by matrix $\bar{A}$. This allows us to argue about the expected instance, $\bar{A}$, for which we have a better handle. Specifically, to that end, we will use Lemma \ref{lem:edge_prob2} to understand the value of $\bar{A}_{u,v}$ for different leaves $u,v\in L'$. For the probabilities of edges between vertices of $S$, we have the following lemma:
\begin{lemma}
    \label{lem:entries_of_exp_for_vertices_in_S}
    Under the assumptions of \Cref{thm:qed-guarantee}, for every $u,v\in S$, we have that:
    \begin{align*}
        \frac{c}{n}\left(1-\frac{3\eta}{2}\right)\hat{\delta}_{S}^2\leq\bar{A}_{u,v}= \prob{(u,v)\in G'}-q.
    \end{align*}
\end{lemma}
\begin{proof}
    Fix $u,v\in S$, by the definition of $\delta_{S|L'}$:
    \[        \delta_{uv|L'}\leq \delta_{S|L'}\leq \hat{\delta}_{S}.
    \]
    Using the fact that $\rho^{-}(x)$ is monotonically decreasing in the interval $(-\infty,1]$, we have that it suffices to show that $\rho^{-}(\hat{\delta}_{S})-q\geq \frac{c}{n}\hat{\delta}_{S}^2$. We have that:
    \begin{align*}
        \rho^{-}\left(\hat{\delta}_{S}\right)-q&=\rho^{-}\left(\deltaest\right)-\rho^-\left(\deltaest+4\deltaest^2\right)\\
        &=\frac{c}{n}\left(1-\frac{3\eta}{2}\right)\left(8\deltaest^2-O\left(\deltaest^3\right)\right)\\
        &\geq \frac{c}{n}\left(1-\frac{3\eta}{2}\right)\deltaest^2,
    \end{align*}
    where we have used that $\deltaest$ is a small enough constant.
\end{proof}

We now provide a lemma which upper bounds $\bar{A}_{u,v}$ for $u\in S$ and $v\notin S$:
\begin{lemma}
    \label{lem:edge_prob_u_in_S_v_not_inS}
    Let $u\in S$ and $v\in L'\setminus S$ and suppose that the assumptions of \Cref{thm:qed-guarantee} hold. Then, we have that:
    \begin{equation}
        \bar{A}_{u,v}=\prob{(u,v)\in G'}-q\leq 10\left(1-\frac{3\eta}{2}\right) \frac{c}{n}\deltaest^2.
    \end{equation}
    Furthermore, if it holds that $\delta_{uv|L'}\geq \deltaest+11\deltaest^2$ then we have that:
    \[
        \bar{A}_{u,v}\leq 0.
    \]
\end{lemma}
\begin{proof}
    We first show the second part of the lemma. We distinguish between cases for $v$. Consider the case where $v\in \mathcal{R}$ (recall $u\in \mathcal{L}$), we show that under the assumptions of the lemma, $\bar{A}_{u,v}\leq 0$. For that we use Lemma \ref{lem:edge_prob2} and in particular \Cref{eq:edge_prob_root_general}. We have that $\alpha\geq 48\delta\geq 6\delta _{S|L'}$ and that $\beta\geq \frac{1}{4}$, the bound that we get is:
    \begin{align*}        \bar{A}_{u,v}&=\prob{(u,v)\in G'}-q\\
        &\leq \frac{c}{n}\left((1-2\alpha\beta)\left(1-\frac{3\eta}{2}\right)+ \frac{3\eta}{2}+o(1)\right)-\rho^{-}\left(\deltaest+4\deltaest^2\right)\\
        &\leq \frac{c}{n}\left(1-\frac{3\eta}{2}\right)\left(-\delta_{S|L'}+O(\deltaest^2)\right)\\
        &\leq 0,
    \end{align*}
    where we have assumed that $\delta_{S|L'}$  is a sufficiently small constant and $\deltaest$ is sufficiently close to $\delta_{S|L'}$. We now consider the case where $v\in \mathcal{L}$. Note that using our assumption that $|\mathcal{R}\cap L'|\geq \frac{n}{4}$, we have that $\delta_{uv|L'}\leq 3/4$. We will use the monotonicity of $\rho^{+}(x)$, namely, that $\rho^{+}(x)$ is monotonically decreasing in the interval $(-\infty, 1/2]$ and monotonically increasing in the interval $[1/2, +\infty)$. Using monotonicity, it suffices to show that (i) $\rho^{+}\left(\deltaest+11\deltaest^2\right)-q\leq 0$ and that (ii) $\rho^{+}(\frac{3}4)-q\leq 0$. To prove (i), we have that:
    \[
        \rho^{+}(\deltaest +11\deltaest^2)-q=\frac{c}{n}\left(-13\deltaest^2+O\left(\deltaest^3\right)\right)
        \leq 0.
\]
    To prove (ii) we have:
    \[
        \rho^+(3/4)-q=\frac{c}{n}\left(1-\frac{3\eta}{2}\right)\left(-0.375+O\left(\deltaest\right)\right)
        \leq 0.
    \]
    We now show that for any $u\in S$ and $v\notin S$, we have that:
    \[
        \bar{A}_{u,v} \leq 10\frac{c}{n}\left(1-\frac{3\eta}{2}\right)\deltaest^2.
    \]
    We have already shown that $\bar{A}_{u,v}\leq 0$ if $\delta_{uv|L'}\geq \frac{1}{2}$. Observe that for $u\in S$ and $v\notin S$ we have that $\delta_{uv|L'}\geq \delta_{S|L'}$ and using again the monotonicity of $\rho^{+}(x)$ in the interval $(-\infty, 1/2]$ we get that:
    \begin{align*}
        \bar{A}_{u,v}&\leq\rho^{+}(\delta_{uv|L'})-q\\&
        \leq \rho^{+}(\delta_{S|L'})-\rho^{-}(\deltaest+4\deltaest^2)\\
        &=\frac{c}{n}\left(1-\frac{3\eta}{2}\right)\left(9\deltaest^2+O(\deltaest^3)\right)\\
        &\leq \frac{10c}{n}\left(1-\frac{3\eta}{2}\right)\deltaest^2.
    \end{align*}
\end{proof}

Lemma \ref{lem:entries_of_exp_for_vertices_in_S} shows that for $u,v\in S$, we have that $\bar{A}_{u,v}\geq 0$ (in fact this is true by some margin) it might, however, be the case that there are vertices outside the set $S$ that have large edge probabilities with vertices in $S$. In particular, there might exist $v\in L'\setminus S$ such that there exist $u\in S$ with $\bar{A}_{u,v}\geq 0$. We will use Lemma \ref{lem:edge_prob_u_in_S_v_not_inS} to show that there are not many such vertices. Consider the set:
\[
    S_e=\set{v\in L': \exists u
    \in S, \bar{A}_{u,v}\geq 0}.
\]
Notice that by Lemma \ref{lem:entries_of_exp_for_vertices_in_S}, $S\subseteq S_e$, and let $B=S_e\setminus S$. We will bound the size of the set $S_e$ and this will give us a bound on the size of the set $B$.

\begin{figure}[ht]
    \centering
    \begin{overpic}[width=0.3\textwidth]{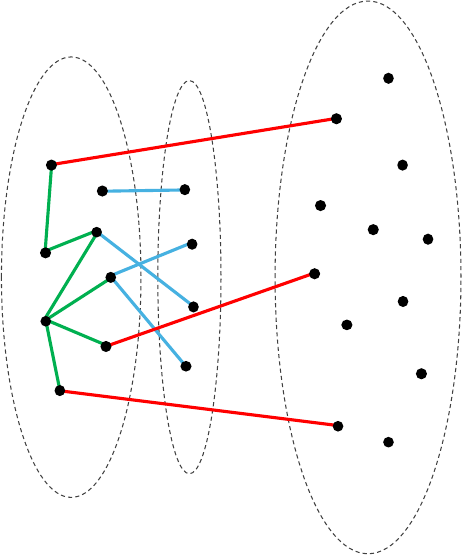}
        \put(18,90){$S$}
        \put(38,85){$B$}
        \put(80,90){$L' \setminus (S \cup B)$}
    \end{overpic}
    \caption{Edges between $S$ and $L'\setminus (S \cup B)$ (red edges) occur with probability at most $q$. Edges within $S$ (green edges) occur with probability at least $q+\frac{c}{n}\left(1-\frac{3\eta}{2}\right)\deltaest^2$.}
    \label{fig:sbm}
\end{figure}

\begin{lemma}
    \label{lem:size of sets S_e and B}
    Suppose that the assumptions of \Cref{thm:qed-guarantee} hold. Consider the set $S_e$ defined as above. For the size of the set $S_e$, we have that:
    \begin{align*}
        |S_e|\leq \left(\deltaest+11 \deltaest^2\right)|L'|
    \end{align*}
    Which in particular implies that:
    \begin{align*}
        |B|\leq 12 \deltaest^2|L'|
    \end{align*}
\end{lemma}
\begin{proof}
    By Lemma \ref{lem:edge_prob_u_in_S_v_not_inS}, we have that if $b\in S_e$, then there exists a vertex $a\in S$ such that $\delta_{ab|L'}\leq \deltaest+11\deltaest^2$.
    Consider the set $\mathcal{S}_e=\set{b\in L': \exists a\in S \text{, } \delta_{ab|L'}\leq \deltaest+11\deltaest^2}$, which by Lemma \ref{lem:edge_prob_u_in_S_v_not_inS} is a superset of the set $S_e$. Observe that $\mathcal{S}_e$ can also be written as:
    \begin{align*}
        \mathcal{S}_e=S\cup \set{b\in L'\setminus S: \forall a \in S\text{, }\delta_{ab|L'}\leq \deltaest+11\deltaest^2},
    \end{align*}
    this follows from the fact that if $b\notin S$ then the $\lca$ of b with any vertex in $S$ is the same (recall that $S$ corresponds to the leaves of a subtree).
    We now consider the subtree rooted at the $\lca$ of the vertices in $\mathcal{S}_e$ and show that the number of leaves of this subtree is bounded by $\left(\deltaest+11\deltaest^2\right)|L'|$, this implies that the set $\mathcal{S}_e$ is bounded by this value, which in turn implies that the size of $S_e$ is also bounded by $\left(\deltaest+11\deltaest^2\right)|L'|$, since $|S_e|\leq |\mathcal{S}_e|$.
    Formally, let $w=\lca(\mathcal{S}_e)$ and let $a$ be an arbitrary vertex in $S$, we bound the size of the set ${L}_{w}\cap L'$, which is a superset of $\mathcal{S}_e$. By Lemma \ref{lem:getting_the_set_lca_by_two}, there exists a $b\in \mathcal{S}_e$ such that $\lca(a,b)=w$. If $b\in S$ then we have that:
    \[        |{L}_{w}\cap L'|=\delta_{ab|L'}|L'|
    \leq \delta_{S|L'}|L'| \leq (\deltaest+11\deltaest^2)|L'|.
    \]
    and the Lemma follows. If on the other hand, $b\in \mathcal{S}_e\setminus S$, then we have that $\delta_{ab|L'}\leq \deltaest+11\deltaest^2$ and we have that:
    \[
    |{L}_{w}\cap L'|=\delta_{ab|L'}|L'|
    \leq (\deltaest+11\deltaest^2)|L'|.
    \]
    We now observe that:
    \[
    |S_e| \leq |\mathcal{S}_e| \leq |{L}_{w}|,
    \]
    and the lemma follows.
\end{proof}
\subsection{\texorpdfstring{Geometry of the SDP Solution: $S$ is a tight Cluster}{Geometry of the SDP Solution: S is a Tight Cluster}}\label{sec:algo_geometry_S}

In this section, we analyze the geometry of the optimal solution $\set{\bar{u}}_{u\in L'}$ to the SDP described by matrix $A$.
In particular, our goal is to show that the vectors corresponding to set $S$ are well clustered together. To that end, we let $\alpha=\Avg_{\substack{u,v\in S\\u\neq v}}\inner{\bar{u},\bar{v}}$ and $\beta=\Avg_{\substack{u\in S, v\in L'\setminus S}}\inner{\bar{u},\bar{v}}$ our analysis shows that $\alpha$ is close to $1$ and $\beta$ is small, if the graph is $(c,\eta)$-good.
At a high level the following lemma shows that we can upper bound $1-\alpha$ by $O(\delta_{S
|L'})$ and $\beta$ by $O(\delta_{S|L'}^2)$ if we assume $c$ to be sufficiently large polynomial in $\frac{1}{\delta_{S|L'}(1-\frac{3\eta}{2})}$:
\begin{lemma}[Bound on $1-\alpha$ and $\beta$]
    \label{lem:bound_on_1-a}
    Suppose that the conditions of \Cref{thm:qed-guarantee} hold, then we have that for $\alpha$:
    \[
        1-\alpha \leq \frac{160K_G}{\delta^4_{S|L'}\gamma^4}\cdot \frac{1}{\sqrt{c'}\left(1-\frac{3\eta}{2}\right)}+482\delta_{S|L'}.
    \]
    Furthermore for $\beta$, we have that
    \[
        \beta\leq \frac{1}{1-\delta_{S|L'}}(\delta_{S|L'}(1-\alpha)+\delta_{S|L'}^2).
    \]
\end{lemma}
\begin{proof}
    We will apply the optimal solution for $A$ to the expected instance $\bar{A}$ and then  examine how much the objective value would increase if all vectors corresponding to set $S$ were equal and orthogonal to everything else. Formally we consider the following SDP solution. Let $e$ be a unit vector orthogonal to all vectors in $\set{\bar{u}}_{u\in L'}$, we consider the solution $\set{\bar{u}_e}_{u\in L'}$ defined as follows:
    \begin{align*}
        \bar{u}_e=
        \begin{cases}
            e \text{, if } u\in S\\
            \bar{u} \text{, if }u\in L'\setminus S,
        \end{cases}
    \end{align*}
    in other words, we move all vectors corresponding to vertices in $S$ to the same vector $e$, which is orthogonal to everything else.
    We compare the values of the two solutions. We have that:
    \begin{align}
        \sum_{u,v\in L'}\bar{A}_{u,v}\inner{\bar{u}_e,\bar{v}_e}-\sum_{u,v\in L'}\bar{A}_{u,v}\inner{\bar{u},\bar{v}}&=\sum_{u,v\in L'}\bar{A}_{u,v}\left(\inner{\bar{u}_e,\bar{v}_e}-\inner{\bar{u},\bar{v}}\right)\notag\\
        &=\sum_{u\in S, v\in S}\bar{A}_{u,v}\left(\inner{\bar{u}_e,\bar{v}_e}-\inner{\bar{u},\bar{v}}\right)\label{eq:first_term}\\&+2\sum_{u\in S, v\in B}\bar{A}_{u,v}\left(\inner{\bar{u}_e,\bar{v}_e}-\inner{\bar{u},\bar{v}}\right)\label{eq:second_term}\\
        &+2\sum_{\substack{u\in S,\\ v\in L'\setminus{B\cup S}}}\bar{A}_{u,v}\left(\inner{\bar{u}_e,\bar{v}_e}-\inner{\bar{u},\bar{v}}\right)\label{eq:third_term}\\
        &+\sum_{u,v\in L'\setminus S}\bar{A}_{u,v}\left(\inner{\bar{u}_e, \bar{v}_e}-\inner{\bar{u},\bar{v}}\right)\label{eq:fourth_term}.
    \end{align}
    We analyze each term in the sum separately, we first start with term (\ref{eq:fourth_term}). We have that for $u\in L'\setminus S$: $\bar{u}_e=\bar{u}$ and $\bar{v}_e=\bar{v}$, therefore:
    \begin{align*}
        \sum_{u,v\in L'\setminus S}\bar{A}_{u,v}\left(\inner{\bar{u}_e, \bar{v}_e}-\inner{\bar{u},\bar{v}}\right)= \sum_{u,v\in L'\setminus S}\bar{A}_{u,v}\left(\inner{\bar{u}, \bar{v}}-\inner{\bar{u},\bar{v}}\right)=0
    \end{align*}
    For the term (\ref{eq:third_term}), we show that it is non-negative. We have that for $u\in S, v\in L'\setminus(B\cup S) $, $\inner{\bar{u}_e, \bar{v}_e}=\inner{e, \bar{v}}=0$ and by constraint (\ref{eq:sdp1}) in the semidefinite program, we have that $\inner{\bar{u},\bar{v}}\geq 0$. Finally we have that for $v\in L'\setminus(B\cup S)$ it holds that $\bar{A}_{u,v}\leq 0$ so we can conclude that:
    \[
        \sum_{u\in S,v\in L'\setminus {(S\cup B)}}\bar{A}_{u,v}\left(\inner{\bar{u}_e, \bar{v}_e}-\inner{\bar{u},\bar{v}}\right)=-\sum_{u\in S,v\in L'\setminus {(S\cup B)}}\bar{A}_{u,v}\inner{\bar{u},\bar{v}}
        \geq 0.
    \]
    We now focus our attention to  term (\ref{eq:first_term}). We use Lemma \ref{lem:entries_of_exp_for_vertices_in_S} and the fact that by constraint \ref{eq:sdp2}, $\inner{\bar{u},\bar{v}}\leq 1$:
    \begin{align*}
        \sum_{u\in S, v\in S}\bar{A}_{u,v}\left(\inner{\bar{u}_e,\bar{v}_e}-\inner{\bar{u},\bar{v}}\right)&=\sum_{u\in S, v\in S}\bar{A}_{u,v}\left(1-\inner{\bar{u},\bar{v}}\right)\\
        &\geq \frac{{c}}{n}\deltaest^2\left(1-\frac{3\eta}{2}\right)\sum_{\substack{u\in S, v\in S\\u\neq v}}(1-\inner{\bar{u},\bar{v}})\\
        &=\frac{{c}}{n}\deltaest^2\left(1-\frac{3\eta}{2}\right)(1-\alpha)\cdot |S|\cdot (|S|-1)\\
        &\geq \frac{{c}}{2n}\deltaest^2\left(1-\frac{3\eta}{2}\right)\left(1-\alpha\right)\delta_{S|L'}^2|L'|^2,
    \end{align*}
    where for the last inequality we assumed that $|S|\geq 2$. Finally, we turn our attention to term \ref{eq:second_term}:
    \begin{align*}
        2\sum_{u\in S, v\in B}\bar{A}_{u,v}\left(\inner{\bar{u}_e, \bar{v}_e}-\inner{\bar{u}, \bar{v}}\right)&=-2\sum_{u\in S,v\in B}\bar{A}_{u,v}\inner{\bar{u},\bar{v}}\\
        &\geq -20\left(1-\frac{3\eta}{2}\right)\frac{c}{n}\deltaest^2\sum_{u\in S, v\in B}\inner{\bar{u},\bar{v}}\\
        &\geq -20\left(1-\frac{3\eta}{2}\right)\frac{c}{n}\deltaest^2|B|\cdot |S|,
    \end{align*}
    were in the first line we have used that for $u\in S$ and $v\in B$, $\inner{\bar{u}_e, \bar{v}_e}=\inner{e, \bar{v}}=0$ and from the first to the second line we have used Lemma \ref{lem:edge_prob_u_in_S_v_not_inS}. Finally, from the second to the third line we have used that $\inner{\bar{u},\bar{v}}\leq 1$.  We now use our bound for $|B|$ from Lemma \ref{lem:size of sets S_e and B} and the fact that $|S|=\delta_{S|L'}|L'|$:
    \[
        2\sum_{u\in S, v\in B}\bar{A}_{u,v}\left(\inner{\bar{u}_e,\bar{v}_e}-\inner{\bar{u},\bar{v}}\right)\geq -20\left(1-\frac{3\eta}{2}\right)\frac{c}{n}\deltaest^2|B|\cdot |S|
        \geq -240\frac{c}{n}\deltaest^4\left(1-\frac{3\eta}{2}\right)\delta_{S|L'}|L'|^2.
\]
    Putting everything together, we have that:
    \begin{align*}
        \sum_{u,v\in L'}\bar{A}_{u,v}\left(\inner{\bar{u}_e,\bar{v}_e}-\inner{\bar{u},\bar{v}}\right)&\geq \frac{{c}}{2n}\left(1-\frac{3\eta}{2}\right)\deltaest^2(1-\alpha)\delta_{S|L'}^2\cdot |L'|^2  \\
        &{\phantom{aaaaaaaaaaaaa}}-240 \cdot\frac{{c}}{n}\left(1-\frac{3\eta}{2}\right) \cdot \deltaest^4\delta_{S|L'}|L'|^2\\
        &\geq\frac{{c}}{n}\left(1-\frac{3\eta}{2}\right)\left(\frac{1}{2}\delta^4_{S|L'}(1-\alpha)-241\delta^5_{S|L'}\right)|L'|^2\\
        &={c}\left(1-\frac{3\eta}{2}\right)\left(\frac{1}{2}\delta_{S|L'}^4(1-\alpha)-241\delta_{S|L'}^5\right)\gamma^2n.
    \end{align*}
    Where we have assumed that $\deltaest$ is sufficiently close to $\delta_{S|L'}$.

    We will now use this to get a bound on $\alpha$. We want to look at the difference that we would get if we applied solution $\set{\bar{u}_e}_{u\in L'}$ instead of $\set{\bar{u}}_{u\in L'}$ to the realized instance given by matrix $A$. Recall that $\set{\bar{u}}_{u\in L'}$ is assumed to be the optimal solution for $A$. Here we use our assumption that $G$ is $(c',\eta)$-good. On the one hand this condition tells us that for every SDP solution the value of the solution on the expected instance and the actual instance are very close. On the other hand, we have that we can improve the objective value by putting vectors that correspond to vertices in $S$ together. This gives us a bound on how large $1-\alpha$ can be. We have that:
    \begin{align*}
        \sum_{u,v}A_{u,v}\inner{\bar{u}_e,\bar{v}_e}-\sum_{u,v}A_{u,v}\inner{\bar{u},\bar{v}}&\geq \sum_{u,v}\bar{A}_{u,v}\inner{\bar{u}_e,\bar{v}_e}-10K_G\sqrt{c'}\cdot n- \sum_{u,v}\bar{A}_{u,v}\inner{\bar{u},\bar{v}}-10K_{G}\sqrt{c'}\cdot n\\
        &=\sum_{u,v}\bar{A}_{u,v}\left(\inner{\bar{u}_e,\bar{v}_e}-\inner{\bar{u},\bar{v}}\right)-20K_{G}\sqrt{c'}\cdot n\\
        &\geq {{c}}\left(1-\frac{3\eta}{2}\right)\left(\frac{1}{2}\delta^4_{S}(1-\alpha)-241\delta^5_{S}\right)\gamma^2n- 20K_{G}\sqrt{c'}\cdot n.
    \end{align*}
    By virtue of the fact that $\set{\bar{u}}_{u\in L'}$ is an optimal solution for the instance given by matrix $A$ we have that:
    \begin{align*}
        \sum_{u,v}A_{u,v}\inner{\bar{u}_e,\bar{v}_e}-\sum_{u,v}A_{u,v}\inner{\bar{u},\bar{v}}\leq 0.
    \end{align*}
    Solving for $1-\alpha$, we get that:
    \begin{align*}
        1-\alpha\leq \frac{40 K_{G}}{\delta^4_{S}\gamma^2}\cdot \frac{\sqrt{c'}}{{c\left(1-\frac{3\eta}{2}\right)}} +482\delta_{S|L'}.
    \end{align*}
    We now replace ${c}=\frac{\gamma^2 c'}{4}$ and get that:
    \begin{align*}
        1-\alpha \leq \frac{160K_G}{\delta^4_{S}\gamma^4}\cdot \frac{1}{\sqrt{c'}\left(1-\frac{3\eta}{2}\right)}+482\delta_{S|L'}.
    \end{align*}
    We now turn our attention to bounding $\beta$. We consider \Cref{eq:sdp3}
    of the semidefinite program for every $u\in S$. Summing over all vertices $u\in S$ we have that:
    $$
    \sum_{u\in S}\sum_{v\in L'}\inner{\bar{u},\bar{v}}\leq |S|\cdot \deltaest \gamma n  =\delta_{S|L'}\deltaest (\gamma n)^2.
    $$
    We now use that:
    \begin{align*}
        \sum_{u\in S}\sum_{v\in L'\setminus S}\inner{\bar{u},\bar{v}}&=\sum_{u\in S}\sum_{v\in L'}\inner{\bar{u},\bar{v}}-\sum_{u\in S}\sum_{v\in S}\inner{\bar{u},\bar{v}}\\
        &\leq \delta_{S|L'}\deltaest (\gamma n)^2-\alpha|S|^2\\
        &=\left(\delta_{S|L'}\deltaest-\delta_{S|L'}^2\alpha\right)(\gamma n)^2.
    \end{align*}
   We will now use that $\delta_{S|L'}$ and $\deltaest$ are sufficiently close (for example $\deltaest\leq \delta_{S|L'}+\delta_{S|L'}^2$). In particular, we have that:
    \begin{align*}
        \sum_{u\in S}\sum_{v\in L'\setminus S}\inner{\bar{u},\bar{v}}&\leq \left(\delta_{S|L'}\deltaest-\delta^2_{S}\alpha\right)(\gamma n)^2\\
        &\leq \left(\delta_{S|L'}(\delta_{S|L'}+\delta_{S|L'}^2)-\delta_{S|L'}^2\alpha\right)(\gamma n)^2\\
        &=\delta_{S|L'}^2(1-\alpha)(\gamma n)^2+\delta_{S|L'}^3(\gamma n)^2.
    \end{align*}
    We now use the definition of $\beta$, to get that:
    \begin{align*}
        \beta&=\Avg_{u\in S, v\in L'\setminus S}\left(\inner{\bar{u},\bar{v}}\right)\\
        &=\frac{1}{|S|}\frac{1}{|L'\setminus S|}\sum_{u\in S}\sum_{v\in L'\setminus S}\inner{\bar{u},\bar{v}}\\
        &\leq \frac{1}{1-\delta_{S|L'}}(\delta_{S|L'}(1-\alpha)+\delta_{S|L'}^2).
    \end{align*}
\end{proof}
\subsection{\texorpdfstring{Proof of \Cref{thm:qed-guarantee}: Recovering $S$ up to $O(\delta^2n)$ error}{Recovering S up to O(delta²|L'|) Error}}
\label{sec:algo_recover_S}
In this section, we prove \Cref{thm:qed-guarantee}. We start with the following lemma.
\begin{lemma}
    \label{lem:bound_on_average_square_dist}
    Under the assumptions of \Cref{thm:qed-guarantee}:
    \[
        \Avg_{u,v\in S}\left(\|\bar{u}-\bar{v}\|^2\right)\leq 2(1-\alpha)\leq 966\cdot \delta_{S|L'},
    \]
    and for $\beta$:
    \[
        \beta\leq 485\delta_{S|L'}^2
    \]
\end{lemma}
\begin{proof}
    We have that:
    \[
        \Avg_{u,v\in S}\left(\|\bar{u}-\bar{v}\|^2\right)=\Avg_{u,v\in S}\left(\|\bar{u}\|^2+\|v\|^2-2\inner{\bar{u},\bar{v}}\right)
        =\Avg_{u,v\in S}\left(2-2\inner{\bar{u},\bar{v}}\right)
        \leq2(1-\alpha).
    \]
    We now use our bound on $1-\alpha$ from Lemma \ref{lem:bound_on_1-a} and our assumption on $c'$ (recall that $\delta_{S|L'}\ge \delta$ and that $\gamma\ge 1/4$) to get the first inequality. In particular, we have that $c'\geq  \left(\frac{40960\cdot K_G}{\left(1-\frac{3\eta}{2}\right)\delta^5}\right)^2$, replacing this into the bound that we have on $1-\alpha$ gives us that:
    \begin{align*}
        2(1-\alpha)&\leq 2\frac{1}{\sqrt{c'}}\cdot \frac{160K_G}{\delta_{S|L'}^4\gamma^4(1-\frac{3\eta}{2})}+964\delta_{S|L'}\\
        &\leq2\frac{(1-\frac{3\eta}{2})\delta^5}{40960K_G}\cdot\frac{40960K_G}{\delta_{S|L'}^4(1-\frac{3\eta}{2})}+964\delta_{S|L'}\\
        &\leq\frac{\delta^5}{\delta_{S|L'}^4}+964\delta_{S|L'}\\
        &\leq 966\cdot \delta_{S|L'},
    \end{align*}
    where in the first line we have used on $1-\alpha$ from Lemma \ref{lem:bound_on_1-a}, in the second line we have used our assumption on $c'$ and that $\gamma\geq \frac{1}{4}$ and in the fourth line we have used that \[\delta_{S|L'}=\frac{|S|}{|L'|}\geq \frac{|S|}{n}\geq  \delta.
    \]
    For the second inequality we will use the bound on $\beta$ in Lemma \ref{lem:bound_on_1-a}, and that $\delta_{S|L'}$ is a small enough constant. In particular, assuming that $\frac{1}{1-\delta_{S|L'}}\leq \frac{485}{484}$, we have that:
    \begin{align*}
        \beta&\leq \frac{1}{1-\delta_{S|L'}}(\delta_{S|L'}(1-\alpha)+\delta_{S|L'}^2)\\
        &\leq \frac{485}{484}484\delta_{S|L'}^2\\
        &\leq 485\delta_{S|L'}^2.
    \end{align*}
\end{proof}

We now prove \Cref{thm:qed-guarantee}.

\begin{proof}[Proof of \Cref{thm:qed-guarantee}]
    We first define $\mu$ to be the center of points in $S$:
    \[        \mu=\frac{1}{|S|}\sum_{u\in S}\bar{u}.
    \]
    We will show that if we sample a point that is sufficiently close to $\mu$ and whose average squared distance to points in $L'\setminus S$ is large, then $\hat{S}$ satisfies the desired properties. We bound the average squared distance between points in $S$ and $\mu$.
    By Lemma \ref{lem:Avg dist between points to avg dist between points and center} and Lemma \ref{lem:bound_on_average_square_dist}, we have that:
    $$
    \Avg_{u\in S}(\|\bar{u}-\mu\|^2)= \frac{1}{2}\Avg_{u\in S, v\in S}\left(\|\bar{u}-\bar{v}\|^2\right)     \leq 483\delta_{S|L'}.
    $$
    We now define the set $\Gamma$ of good points that we want to sample $u$ from ($u$ is the random vertex sampled in Line 2 of Algorithm \ref{alg:recover_S}), to that end we let
    \[
        S_{B_1}=\set{u\in S: \|\bar{u}-\mu\|^2\geq 4\cdot 483 \delta_{S|L'}}.
    \]
    and also
    \[
        S_{B_2}=\set{u\in S: \Avg_{v\in L'\setminus S}\|\bar{u}-\bar{v}\|^2\leq 2-8\cdot 485\delta_{S|L'}^2}.
    \]
    We let $\Gamma=S\setminus(S_{B_1}\cup S_{B_2})$. We argue about the size of set $\Gamma$ by arguing about the size of sets $S_{B_1}$ and $S_{B_2}$. We let $U$ be a uniformly selected vertex in $S$ and $W$ be a uniformly selected vertex in $L'\setminus S$. By Markov's inequality, we have that:
    $$
        \frac{|S_{B_1}|}{|S|}=\prob{\|\bar{U}-\mu\|^2\geq 4\cdot 483\delta_{S|L'}}
        \leq \frac{1}{4},
    $$
    meaning that $|S_{B_1}|\leq \frac{1}{4}|S|$. For set $S_{B_2}$ we have that:
    \begin{align*}
        \frac{|S_{B_2}|}{|S|}&=\prob{\Avg_{v\in L'\setminus S}\left(\|\bar{U}-\bar{v}\|^2\right)\leq 2-8\cdot 485\delta_{S|L'}^2}\\
        &=\prob{\Exp{\|\bar{U}-\bar{W}\|^2|U}\leq 2-8\cdot 485\delta_{S|L'}^2}\\
        &=\prob{2-\Exp{\|\bar{U}-\bar{W}\|^2|U}\geq 8\cdot 485\delta_{S|L'}^2}.
    \end{align*}

    We now use that for any $u,v\in L'$, $\inner{\bar{u},\bar{v}}\geq 0$, therefore $\|\bar{u}-\bar{v}\|^2\leq 2$. This implies that the random variable $2-\Exp{\|\bar{U}-\bar{W}\|^2|U}$ is non-negative allowing us to apply Markov's inequality. We get that:
    \begin{align*}
        \frac{|S_{B_2}|}{|S|}&\leq \frac{\Exp{2-\Exp{\|\bar{U}-\bar{W}\|^2|U}}}{8\cdot 485\delta_{S|L'}^2}\\
        &\leq \frac{2-\Exp{\|\bar{U}-\bar{W}\|^2}}{8\cdot 485\delta_{S|L'}^2}\\
        &=\frac{2\Avg_{u\in S, v\in L'\setminus S}\left(\inner{\bar{u},\bar{v}}\right)}{8\cdot 485\delta_{S|L'}^2}\\
        &=\frac{2\beta}{8\cdot 485\delta_{S|L'}^2}\\
        &\leq \frac{1}{4},
    \end{align*}
    where we have used our bound on $\beta$ from Lemma \ref{lem:bound_on_average_square_dist}. This in turn means that $|S_{B_2}|\leq \frac{1}{4}|S|$. Putting everything together, we get that $|\Gamma|=|S\setminus \left(S_{B_1}\cup S_{B_2}\right)|\geq |S|-\frac{1}{2}|S|=\frac{1}{2}|S|$. This gives us that the probability that $u$ is in $\Gamma$ is at least $\frac{1}{2}\delta_{S|L'}$. We now show that if $u\in \Gamma$ then for $\hat{S}$ it holds that $|S\triangle \hat{S}|=O(\delta^2n)$, notice that $u\in \Gamma$ implies that
    \begin{align}
        \label{eq:dist_p_mu}
        \|\bar{u}-\mu\|^2\leq 4\cdot 483\delta_{S|L'},
    \end{align}
    and also
    \begin{align}
        \Avg_{v\in L'\setminus S}\left(2-\|\bar{v}-\bar{u}\|^2\right)\leq8\cdot 485\delta_{S|L'}^2.
    \end{align}
    We first focus on the size of $S\setminus \hat{S}$. By Markov's inequality and relaxed triangle inequality:
    \begin{align*}
        \frac{|S\setminus \hat{S}|}{|S|}&=\prob{\|\bar{U}-\bar{u}\|^2\geq 1}\\
        &\leq {\Exp{\|\bar{U}-\bar{u}\|^2}}{}\\
        &\leq \Exp{2\|\bar{U}-\mu\|^2+2\|\bar{u}-\mu\|^2}\\
        &\leq {10}\cdot 483\delta_{S|L'}.
    \end{align*}
    Meaning that $|S\setminus \hat{S}|\leq4830 \delta_{S|L'}^2\gamma n=O(\delta^2n )$ (since $\delta_{S|L'}\leq 8\delta $). We now turn our attention to bounding the size of the set $\hat{S}\setminus S$. We have that
    $$
    \frac{|\hat{S}\setminus S|}{|L'\setminus S|}=\prob{\|\bar{W}-\bar{u}\|^2\leq 1}
    = \prob{2-\|\bar{W}-\bar{u}\|^2\geq 1}.
    $$
    We can now use that $2-\|\bar{W}-\bar{u}\|^2$ is a non-negative random variable to apply Markov's inequality:
    \begin{align*}
        \frac{|\hat{S}\setminus S|}{|L'\setminus S|}&=\prob{2-\|\bar{W}-\bar{u}\|^2\geq 1}\\
        &\leq {\Exp{2-\|\bar{W}-\bar{u}\|^2}}\\
        &={\Avg_{v\in L'\setminus S}\left(2-\|\bar{v}-\bar{u}\|^2\right)}\\
        &\leq {8}\cdot 485\delta_{S|L'}^2.
    \end{align*}
    For the size of the set $\hat{S}\setminus S$, we get that $|\hat{S}\setminus S|\leq {8}\cdot 485\delta_{S|L'}^2|L'|={8}\cdot 485\delta_{S|L'}^2\gamma n=O(\delta^2 n)$.
\end{proof}

\section{Natarajan Dimension of Quartets}\label{sec:nat-dimension}

Here we derive the bound on the Natarajan dimension $N_{\mathrm{quart}}(n)$ for learning phylogenetic trees from quartets. Recall, that the Natarajan dimension governs learnability for multiclass settings, similar to how VC dimension governs learnability in binary classification (see also the Multiclass Fundamental Theorem 29.3, Chapter 29 in~\citep{shalev2014understanding}). Since our Theorem~\ref{thm:natarajan_quart} here guarantees that $N_{\mathrm{quart}}(n) \;\le\; 10n$, this means that a linear in $n$ sample of quartets suffices for generalization, and so our previous tree reconstruction algorithm finds a tree $\widehat T$ with low generalization error.

We now formally define the Natarajan dimension for learning quartets. In our definition we naturally view $\binom{L}{4}$ (subsets of $L$ of size $4$), as the domain set and quartet constraints as the label set. The set of binary trees that have $L$ as leaves is viewed as the hypothesis class: each tree $T$ is interpreted as a function that on input a set of $4$ elements $\set{a,b,c,d}$ outputs the unique quartet constraint (among three possible), $T(a,b,c,d)$, that the tree satisfies. Our definitions are similar to the definitions of \cite*{avdiukhin2023tree} but applied to quartet reconstruction. We start by defining Natarajan Shattering:
\begin{definition}[Natarajan Shattering]
    \label{def: Natarajan shattering}
    Let $S=\set{\set{a_1,b_1,c_1,d_1}, \ldots, \set{a_m,b_m,c_m,d_m}}\subseteq \binom{L}{4}$ be a collection of $4$-element sets of size $m$. We say that $S$ is Natarajan shattered if for every $\set{a_i,b_i, c_i,d_i}\in S$ there exist two distinct constraints $f_0(a_i,b_i,c_i,d_i)$ and $f_1(a_i,b_i,c_i,d_i)$ on $\set{a_i,b_i,c_i,d_i}$, such that for every function $g:[m]\rightarrow \set{0,1}$, the instance 
    $$\mathcal{I}_{g}=\set{f_{g(i)}(a_i,b_i,c_i,d_i): i\in [m]}$$
    is satisfiable.
\end{definition}

Using this definition, we can define the Natarajan dimension:
\begin{definition}[Natarajan dimension]
    The Natarajan dimension of quartets, denoted $N_{\mathrm{quart}}(n)$, is defined as the size of the largest cardinality of a set which can be Natarajan shattered.
\end{definition}

We prove that the Natarajan dimension of quartets, is linear i.e. $N_{\mathrm{quart}}(n)=\Theta(n)$. Our argument uses the following theorem of \cite{warren1968lower} which was also used by \cite*{alon2023optimal} to give bounds on the Natarajan dimension for contrastive learning.
\begin{theorem}[\cite{warren1968lower}]
    Let $m\geq l\geq 2$ be integers, and let $p_1,\ldots,p_m$ be real polynomials on $l$ variables, each of degree $\leq k$. Let
    \begin{align*}
        U(p_1,\ldots, p_m)=\set{\vec{x}\in \mathbb{R}^l\text{ }|\text{ } p_i(\vec{x})\neq 0\text{ } \forall  i\in [m] }.
    \end{align*}
    be the set of points in $x\in \mathbb{R}^{l}$ which are non-zero in all polynomials. Then the number of connected components in $U(p_1,\ldots, p_m)$ is at most $(4ekm/l)^l$.
\end{theorem}
In particular we use the following corollary.
\begin{corollary}
    \label{cor: Warren}
    Let $m\geq l\geq 2$ be integers, and let $p_1,\ldots,p_m$ be real polynomials on $l$ variables, each of degree $\leq k$. Then we have that
    \begin{align*}
        \left|\set{\xi \in \set{-1,1}^{m}: \exists \vec{x} \text{ such that for all }i \text{, } \sign(p_i(x))=\xi(i)}\right|\leq \left(\frac{4ekm}{l}\right)^l.
    \end{align*}
\end{corollary}

To apply this corollary, we begin with the following embedding lemma.
\begin{lemma}
    \label{lem:embedding_Tree_to_line}
    Let $T$ be a binary tree and $0<\gamma<\frac{1}{2}$. There exists a mapping $\phi$ of the leaves of $T$ to the interval $[0,1]$ with the following property. Let $v_i$ and $v_j$ be two leaves of the tree and let $\lambda$ be the depth of their lowest common ancestor. For the points $x_i=\phi(v_i)$ and $x_j=\phi(v_j)$ we have that
    \begin{align*}
        (1-2\gamma)\gamma^{\lambda}\leq|x_i-x_j|\leq \gamma^{\lambda}.
    \end{align*}
\end{lemma}
\begin{proof}
    We describe a recursive algorithm that provides the mapping $\phi$. Given an interval $\mathcal{I}=[a,b]$ and tree $T$, the algorithm divides between cases. If $T$ has only one node then this node is mapped to $\frac{a+b}{2}$. Otherwise, $\mathcal{I}$ is divided into the intervals $[a,a+\gamma(b-a)]$, $[a+\gamma(b-a), b-\gamma(b-a)]$ and $[b-\gamma(b-a),b]$, we recursively map all the leaves of the left subtree to the interval $[a,a+\gamma(b-a)]$ and all leaves of right subtree to interval $[b-\gamma(b-a), b]$. The mapping $\phi$ is the output of the algorithm on interval $[0,1]$ and tree $T$. It is easy to show inductively that every subtree whose root is at depth $\lambda$ is mapped to an interval of length $\gamma^{\lambda}$. Now consider points $x_i$ and $x_j$. The leaves of the tree rooted at their lowest common ancestor are mapped to an interval $\mathcal{I}=[a,b]$ of length $\gamma^{\lambda}$, this in particular gives us the upper bound on $|x_i-x_j|$. For the lower bound, observe that one of $x_i,x_j$ will be in the interval $\mathcal{I}_1=[a,a+\gamma(b-a)]$ while the other will be in the interval $\mathcal{I}_2=[b,b-\gamma(b-a)]$. The distance between any point in $\mathcal{I}_1$ and any point in $\mathcal{I}_2$ is at least $(1-2\gamma)(b-a)=(1-2\gamma)\gamma^{\lambda}$, giving us the lower bound on $|x_i-x_j|$.
\end{proof}
Given the embedding guaranteed by \Cref{lem:embedding_Tree_to_line}, we will describe a polynomial that is positive if a quartet constraint is satisfied and negative otherwise. We prove the following lemma
\begin{lemma}\label{lem:polynomial_quart}
    Let $T$ be a tree with leaves $v_1,\ldots, v_n$ and $\phi(v_i)=x_i$ where $\phi$ is the mapping of \Cref{lem:embedding_Tree_to_line} and let $\gamma > 0$ be small enough so that $(1-2\gamma)^8>\max(\gamma^2,\gamma^4,\gamma^8)$ (e.g., $\gamma \le \frac{1}{16}$). Consider leaves $v_i,v_j,v_k,v_l$ and the polynomial:
    \begin{align*}
        p(x_i,x_j,x_k,x_l)=\left((x_i-x_k)(x_i-x_l)(x_j-x_k)(x_j-x_l)\right)^2-\left((x_i-x_j)(x_k-x_l)\right)^4.
    \end{align*}
    We have the following.
    \begin{enumerate}
        \item If the quartet constraint $v_iv_j|v_kv_l$ is satisfied in $T$ then $p(x_i,x_j,x_k,x_l)>0$.
        \item Otherwise, if $v_iv_j|v_kv_l$ is not satisfied in $T$ then $p(x_i,x_j,x_k,x_l)<0$.
    \end{enumerate}
\end{lemma}
\begin{proof}
    We first prove item 1 of the lemma. By symmetry, it suffices to only consider the case where (a) $v_i$ and $v_j$ are in one subtree of $T$ and $v_k$, $v_l$ are on another disjoint subtree (as in \Cref{fig:case_a}) and the case (b) where $\lca(v_j,v_l)$ is an ancestor of $\lca(v_k,v_l)$ and $\lca(v_i,v_j)$ is an ancestor of $\lca(v_j,v_k,v_l)$ (as in \Cref{fig:case_b}).

    For case (a), let $\depth(\lca(v_i,v_j))=\lambda_1$, $\depth(\lca(v_k,v_l))=\lambda_2$ and  
    \[\depth(\lca(v_i,v_j,v_k,v_l))=\lambda_3.
    \]
    Note that $\lambda_1\geq \lambda_3+1$ and also $\lambda_2\geq \lambda_3+1$. Using the lower bound of \Cref{lem:embedding_Tree_to_line} we have that
    \begin{align*}
        \left((x_i-x_k)(x_i-x_l)(x_j-x_k)(x_j-x_l)\right)^2\geq (1-2\gamma)^8\gamma^{8\lambda_3}.
    \end{align*}
    On the other hand, using the upper bound of \Cref{lem:embedding_Tree_to_line} we get that
    $$
    \left((x_i-x_j)(x_k-x_l)\right)^{4}\leq \gamma^{4\lambda_1+4\lambda_2}
    \leq \gamma^{8\lambda_3+8}.
    $$
    Where we have used that both $\lambda_1$ and $\lambda_2$ are greater than $\lambda_3$. We get the following bound on $p(x_i,x_j,x_k,x_l)$:
    \[
        p(x_i,x_j,x_k,x_l)\geq (1-2\gamma)^{8}\gamma^{8\lambda_3}-\gamma^{8\lambda_3+8}
        \geq \gamma^{8\lambda_3}\left((1-2\gamma)^8-\gamma^8\right)
        >0,
    \]
    where we have used that $(1-2\gamma)^{8}>\gamma^8$.
    
    For case (b), let $\depth(\lca(v_k,v_l))=\lambda_1$, $\depth(\lca(v_j,v_k))=\lambda_2$ and $\depth(\lca(v_i,v_j))=\lambda_3$. Using \Cref{lem:embedding_Tree_to_line}, we have that
    \begin{align*}
        \left((x_i-x_k)(x_i-x_l)(x_j-x_k)(x_j-x_l)\right)^2\geq(1-2\gamma)^8\gamma^{4\lambda_3}\gamma^{4\lambda_2}.
    \end{align*}
    While on the other hand:
    $$
    \left((x_i-x_j)(x_k-x_l)\right)^{4}\leq \gamma^{4\lambda_3+4\lambda_1}     \leq \gamma^{4\lambda_3+4\lambda_2+4}.
    $$
    Where we have used that $\lambda_1\geq\lambda_2+1$. We can now bound the polynomial $p(x_i,x_j,x_k,x_l)$:
    \begin{align*}
        p(x_i,x_j,x_k,x_l)\geq \gamma^{4\lambda_2+4\lambda_3}\left((1-2\gamma)^8-\gamma^4\right)>0
    \end{align*}
    where we have used that $(1-2\gamma)^8>\gamma^4$.

    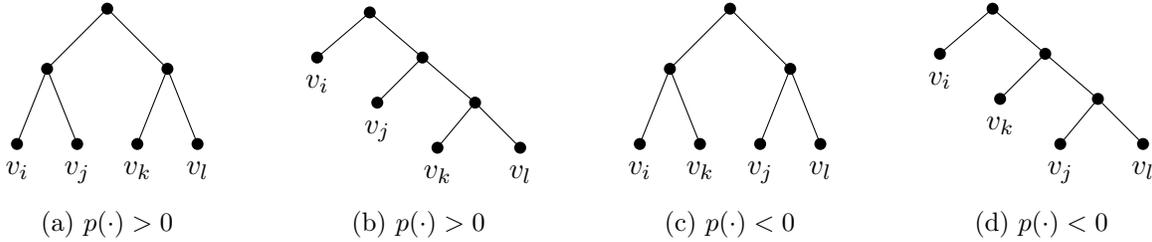
\begin{figure}[ht]\label{fig:quart}
        \centering
        {
            \tikzset{
                dot/.style={circle,fill,inner sep=1.6pt}
            }

            \begin{subfigure}[t]{0.24\textwidth}
                \centering
                \begin{tikzpicture}
                    \coordinate (r)  at (0,1.8);
                    \coordinate (a)  at (-0.8,1.0);
                    \coordinate (b)  at ( 0.8,1.0);
                    \coordinate (vi) at (-1.2,0);
                    \coordinate (vj) at (-0.4,0);
                    \coordinate (vk) at ( 0.4,0);
                    \coordinate (vl) at ( 1.2,0);

                    \draw (r) -- (a) -- (vi);
                    \draw (a) -- (vj);
                    \draw (r) -- (b) -- (vk);
                    \draw (b) -- (vl);

                    \foreach \p in {r,a,b,vi,vj,vk,vl}{\node[dot] at (\p) {};}

                    \node[below=3pt] at (vi) {$v_i$};
                    \node[below=3pt] at (vj) {$v_j$};
                    \node[below=3pt] at (vk) {$v_k$};
                    \node[below=3pt] at (vl) {$v_l$};
                \end{tikzpicture}
                \caption{$p(\cdot) > 0$\label{fig:case_a}}
            \end{subfigure}\hfill
            \begin{subfigure}[t]{0.24\textwidth}
                \centering
                \begin{tikzpicture}
                    \coordinate (r)  at (-0.1,1.8);
                    \coordinate (a)  at ( 0.6,1.2);
                    \coordinate (b)  at ( 1.3,0.6);

                    \coordinate (vi) at (-0.8,1.2);
                    \coordinate (vj) at ( 0,0.6);
                    \coordinate (vk) at ( 0.8,0);
                    \coordinate (vl) at ( 1.9,0);

                    \draw (r) -- (a) -- (b);
                    \draw (r) -- (vi);
                    \draw (a) -- (vj);
                    \draw (b) -- (vk);
                    \draw (b) -- (vl);

                    \foreach \p in {r,a,b,vi,vj,vk,vl}{\node[dot] at (\p) {};}

                    \node[below=3pt] at (vi) {$v_i$};
                    \node[below=3pt] at (vj) {$v_j$};
                    \node[below=3pt] at (vk) {$v_k$};
                    \node[below=3pt] at (vl) {$v_l$};
                \end{tikzpicture}
                \caption{$p(\cdot) > 0$\label{fig:case_b}}
            \end{subfigure}\hfill
            \begin{subfigure}[t]{0.24\textwidth}
                \centering
                \begin{tikzpicture}
                    \coordinate (r)  at (0,1.8);
                    \coordinate (a)  at (-0.8,1.0);
                    \coordinate (b)  at ( 0.8,1.0);
                    \coordinate (vi) at (-1.2,0);
                    \coordinate (vk) at (-0.4,0);
                    \coordinate (vj) at ( 0.4,0);
                    \coordinate (vl) at ( 1.2,0);

                    \draw (r) -- (a) -- (vi);
                    \draw (a) -- (vk);
                    \draw (r) -- (b) -- (vj);
                    \draw (b) -- (vl);

                    \foreach \p in {r,a,b,vi,vj,vk,vl}{\node[dot] at (\p) {};}

                    \node[below=3pt] at (vi) {$v_i$};
                    \node[below=3pt] at (vk) {$v_k$};
                    \node[below=3pt] at (vj) {$v_j$};
                    \node[below=3pt] at (vl) {$v_l$};
                \end{tikzpicture}
                \caption{$p(\cdot) < 0$\label{fig:unsat case_a}}
            \end{subfigure}\hfill
            \begin{subfigure}[t]{0.24\textwidth}
                \centering
                \begin{tikzpicture}
                    \coordinate (r)  at (-0.1,1.8);
                    \coordinate (a)  at ( 0.6,1.2);
                    \coordinate (b)  at ( 1.3,0.6);

                    \coordinate (vi) at (-0.8,1.2);
                    \coordinate (vk) at ( 0,0.6);
                    \coordinate (vj) at ( 0.8,0);
                    \coordinate (vl) at ( 1.9,0);

                    \draw (r) -- (a) -- (b);
                    \draw (r) -- (vi);
                    \draw (a) -- (vk);
                    \draw (b) -- (vj);
                    \draw (b) -- (vl);

                    \foreach \p in {r,a,b,vi,vj,vk,vl}{\node[dot] at (\p) {};}

                    \node[below=3pt] at (vi) {$v_i$};
                    \node[below=3pt] at (vk) {$v_k$};
                    \node[below=3pt] at (vj) {$v_j$};
                    \node[below=3pt] at (vl) {$v_l$};
                \end{tikzpicture}
                \caption{$p(\cdot) < 0$\label{fig:unsat case_b}}
            \end{subfigure}
        }
        \caption{The configuration of $v_i,v_j,v_k,v_l$ determines the sign of $p(x_i,x_j,x_k,x_l)$.}
    \end{figure}

    We now prove item 2 of the lemma. Again by symmetry, it suffices to consider the following two cases. The first case where (c) $v_i$ and $v_k$ are in one subtree of $T$ and $v_j,v_l$ are on another disjoint tree (as in \Cref{fig:unsat case_a}) and the second case (d) where $\lca(v_k,v_l)$ is an ancestor of $\lca(v_j,v_l)$ and $\lca(v_i,v_k)$ is an ancestor of $\lca(v_j,v_k,v_l)$ (as in \Cref{fig:unsat case_b}).
    We first consider case (c) and let $\depth(\lca(v_i,v_k))=\lambda_1$, $\depth(\lca(v_j,v_l))=\lambda_2$ and $\depth(\lca(v_i,v_j))=\lambda_3$. We again use the bounds of \Cref{lem:embedding_Tree_to_line} to get that
    \begin{align*}
        \left((x_i-x_k)(x_i-x_l)(x_j-x_k)(x_j-x_l)\right)^2\leq \gamma^{2\lambda_1+2\lambda_2+4\lambda_3},
    \end{align*}
    while on the other hand, we have that
    \begin{align*}
        \left((x_i-x_j)(x_k-x_l)\right)^4\geq \left(1-2\gamma\right)^{8}\gamma^{8\lambda_3}.
    \end{align*}
    We get the following upper bound on the value of $p(x_i,x_j,x_k,x_l)$:
    \[
        p(x_i,x_j,x_k,x_l)\leq \gamma^{2\lambda_1+2\lambda_2+4\lambda_3}-(1-2\gamma)^{8}\gamma^{8\lambda_3}
        \leq \gamma^{8\lambda_3}(\gamma^{4}-(1-2\gamma)^{8})
        <0.
    \]
    where we have used that $\lambda_1,\lambda_2\geq \lambda_3+1$ and $\gamma^{4}<(1-2\gamma)^{8}$. Finally for case (d), let 
    \[
\depth(\lca(v_j, v_l)) = \lambda_1, \quad
\depth(\lca(v_j, v_k)) = \lambda_2, \quad
\depth(\lca(v_i, v_k)) = \lambda_3.
\]
We use \Cref{lem:embedding_Tree_to_line} for the final time to get that
    \begin{align*}
        \left((x_i-x_k)(x_i-x_l)(x_j-x_k)(x_j-x_l)\right)^2\leq\gamma^{2\lambda_1+2\lambda_2+4\lambda_3}
    \end{align*}
    while on the other hand:
    \begin{align*}
        \left((x_i-x_j)(x_k-x_l)\right)^4\geq(1-2\gamma)^{8}\gamma^{4\lambda_3+4\lambda_2}.
    \end{align*}
    Meaning that
    \begin{align*}
        p(x_i,x_j,x_k,x_l)&\leq \gamma^{2\lambda_1+2\lambda_2+4\lambda_3}-(1-2\gamma)^{8}\gamma^{4\lambda_3+4\lambda_2}\\
        &=\gamma^{2\lambda_2+4\lambda_3}(\gamma^{2\lambda_1}-(1-2\gamma)^{8}\gamma^{2\lambda_2})\\
        &\leq\gamma^{4\lambda_2+4\lambda_3}(\gamma^2-(1-2\gamma)^{8})\\
        &<0.
    \end{align*}
    Where we have used that $\lambda_1\geq \lambda_2+1$ and that $\gamma^2<(1-2\gamma)^{8}$.
\end{proof}

\begin{theorem}\label{thm:natarajan_quart}
    For the Natarajan dimension of quartets we have
    \[
        n-3 \;\le\; N_{\mathrm{quart}}(n) \;\le\; 10n.
    \]
\end{theorem}
\begin{proof}
    \emph{Upper bound.} Fix any $0<\gamma\le \frac{1}{16}$ and the embedding $\phi$ from \Cref{lem:embedding_Tree_to_line}.
    Let $S=\{\{a_i,b_i,c_i,d_i\}\}_{i=1}^m$ be Natarajan-shattered. For each $i$, let $f_0(a_i,b_i,c_i,d_i)$ and $f_1(a_i,b_i,c_i,d_i)$ be the two distinct quartet labels promised by shattering, and define the degree-$8$ polynomial
    \[
        p_i(x_{a_i},x_{b_i},x_{c_i},x_{d_i})
    \]
    as in \cref{lem:polynomial_quart} but with the “positive” side chosen to be $f_0(a_i,b_i,c_i,d_i)$.
    By that lemma, for \emph{any} function $g:[m]\to\{0,1\}$, the tree $T_g$ that satisfies $\mathcal{I}_g=\{f_{g(i)}(a_i,b_i,c_i,d_i)\}_{i=1}^m$ yields reals $x_1,\dots,x_n$ (via $\phi$) with
    \[        \sign\!\left(p_i(x_{a_i},x_{b_i},x_{c_i},x_{d_i})\right)=(-1)^{g(i)}\qquad\forall i\in[m].
    \]
    Thus the number of sign patterns realized by $\{p_i\}_{i=1}^m$ in $\mathbb{R}^n$ is at least $2^m$.
    Applying \cref{cor: Warren} with $l=n$, $m$ polynomials, and degree $k=8$ gives that the number of realizable sign patterns is at most $\left(\tfrac{4ekm}{l}\right)^l=\left(\tfrac{32 e\, m}{n}\right)^n$. Hence
    \[
        2^m \le \left(\tfrac{32 e\, m}{n}\right)^n.
    \]
    Writing $m=tn$ this is $\left(2^t\right)^n \le \left(32 e\, t\right)^n$, i.e., $2^t \le 32 e\, t$.
    The function $t\mapsto \frac{2^t}{t}$ is increasing for $t\ge 2$, and $2^{10}/10>32e$.
    Therefore $2^t>32 e\, t$ for all $t\ge 10$, implying $m\le 10n$.

    \emph{Lower bound.} For rooted $k$-tuple relations, checking their satisfiability is polynomial-time solvable, contrary to our quartet setting where checking satisfiability is $\mathsf{NP}$-complete. The Natarajan dimension for rooted $k$-tuple relations satisfies $N_\mathrm{rooted}(H_k)\ge n-k+1$, as proved in~\cite{avdiukhin2023tree}. For unrooted quartets the result doesn't follow directly, since two distinct rooted labels may collapse to the same unrooted label after forgetting the root. To keep the two labels distinct, we adapt their construction as follows. Fix three leaves $A=\{a_1,a_2,a_3\}\subset V$ and let
    $B=V\setminus A$ so $|B|=n-3$.
    For each $b\in B$ define the 4-set
    $q_b=A\cup\{b\}$ and choose two quartets $f_0(q_b)=a_1 b|a_2 a_3$ and $f_1(q_b)=a_2 b|a_1 a_3$,
    which are distinct in the unrooted sense.
    Given any labeling $g:B\to\{0,1\}$,
    construct a binary tree by first forming a 3-leaf base on
    $\{a_1,a_2,a_3\}$ and, for each $b\in B$,
    attaching $b$ as a sibling of $a_1$ if $g(b)=0$
    or of $a_2$ if $g(b)=1$.
    Each operation determines the quartet on $q_b$
    without affecting any other $q_{b'}$.
    Hence every labeling $\{f_{g(b)}(q_b):b\in B\}$ is realizable,
    and these $n-3$ quartets are Natarajan-shattered, proving $N_{\mathrm{quart}}(n)\ge n-3$.
\end{proof}

\section{Acknowledgment}
The authors thank the anonymous STOC reviewers for valuable feedback, and the organizers and participants of Dagstuhl Seminar 25211, ``The Constraint Satisfaction Problem: Complexity and Approximability,'' for helpful discussions on phylogenetic CSPs.
\newpage
\bibliography{references.bib}
\newpage

\appendix

\section{Inequalities}\label{app:inequalities}

\begin{theorem}[Bernstein inequality, Theorem 2.7 in \cite{mcdiarmid1998concentration}]
    \label{thm:Bernstein}
    Let the random variables $X_1, \ldots, X_n$ be independent, with $X_k - \mathbb{E}(X_k) \leq b$ for each $k$. Let $S_n = \sum X_k$, and let $S_n$ have expected value $\mu$ and variance $V$ (the sum of the variances of the $X_k$). Then for any $t \geq 0$,
    \begin{align*}
        \Pr(S_n - \mu \geq t)
        &\leq e^{-\left( \frac{V}{b^2} \right)\left( (1+\epsilon) \ln(1+\epsilon) -\epsilon \right)} \quad \text{where } \epsilon = bt/V \\
        &\leq e^{-\frac{t^2}{2V \left( 1 + \frac{bt}{3V} \right)}}.
    \end{align*}
\end{theorem}

\begin{theorem}[Grothendieck inequality, \cite{grothendieck1956resume}; see also \cite{krivine1978constantes, braverman2013grothendieck}]
\label{thm: Grothendieck}
    For every $n\times n$ matrix $M$, the following inequality holds:
    \begin{align*}
        \max_{\|U_i\|,\|V_j\|=1}\left|\sum_{i,j}M_{ij}\inner{U_i,V_j}\right|\leq K_{G}\cdot \max_{x,y\in \set{-1,1}^{n}}\sum_{i,j=1}^{n}M_{ij}x_iy_j,
    \end{align*}
    where $K_G\leq 1.783$ is the Grothendieck constant. The first maximum is over all unit vectors $U_1,\ldots, U_n$ and $V_1,\ldots,V_n$.
\end{theorem}

\section{Mapping Quartet Relations on a Rooted Tree}
\label{app:mapping-lemma}
In this section, we show that the mapping $\phi_U$ preserves the quartet relation $(ab \mid cd)$ under mild conditions.
Recall that we consider a rooted tree $T$ and a subset of its vertices $U$.
We define a mapping $\phi_U : V(T) \to U$ by
\[
    \phi_U(u) = \text{the closest ancestor of } u \text{ that lies in } U.
\]
Equivalently, $\phi_U(u)$ is the first vertex in $U$ encountered on the path from $u$ to the root $r$.
Below, we denote by $P(x, y)$ the (unique) path from vertex $x$ to vertex $y$ in the tree,
and by $\mathrm{LCA}(x, y)$ the least common ancestor of $x$ and $y$.

\begin{lemma}\label{lem:disjoint-f}
    Let $a,b,c,d \in V$ be distinct vertices satisfying:
    \begin{enumerate}
        \item[\textnormal{(i)}] The paths $P(a,b)$ and $P(c,d)$ are vertex-disjoint.
        \item[\textnormal{(ii)}] For every distinct triple $x,y,z \in \{a,b,c,d\}$, we have $\phi_U(x) \neq \phi_U(\lca(y,z))$.
    \end{enumerate}
    Then the paths $P(\phi_U(a),\phi_U(b))$ and $P(\phi_U(c),\phi_U(d))$ are vertex-disjoint.
\end{lemma}

\begin{proof}
    Suppose, for contradiction, that the paths $P(\phi_U(a),\phi_U(b))$ and $P(\phi_U(c),\phi_U(d))$ intersect.
    Let $x$ be the lowest vertex in their intersection, that is, the vertex farthest from the root $r$
    that lies on both paths.
    Since $P(a,b)$ and $P(c,d)$ are disjoint, $x$ does not belong to one or both of them.
    Without loss of generality, assume that $x \notin P(a,b)$.

    Because $x \in P(\phi_U(a),\phi_U(b))$, at least one of $\phi_U(a)$ or $\phi_U(b)$ lies in the subtree rooted at $x$, denoted $T_x$.
    If only one of them is in $T_x$, assume it is $\phi_U(a)$.
    If both are in $T_x$ but not both equal to $x$, assume $\phi_U(a)\neq x$.
    If $\phi_U(a)=\phi_U(b)=x$, we make no further assumptions.
    Similarly, since $x\in P(\phi_U(c),\phi_U(d))$, we may assume that $\phi_U(c)\in T_x$.

    We now prove two claims.

    \begin{claim}\label{claim:f-lca-ac}
        We have
        \[
            \phi_U(\lca(a,c)) = \phi_U(x).
        \]
    \end{claim}

    \begin{claim}\label{claim:f-lca-b}
        We have
        \[
            \phi_U(b) = \phi_U(x).
        \]
    \end{claim}

    Together, these claims imply $\phi_U(b) = \phi_U(\lca(a,c))$, contradicting condition~(ii).
    Hence our assumption that the two paths intersect was false, and
    the paths $P(\phi_U(a),\phi_U(b))$ and $P(\phi_U(c),\phi_U(d))$ are vertex-disjoint.

    \begin{proof}[Proof of Claim~\ref{claim:f-lca-ac}]
        Both $\phi_U(a)$ and $\phi_U(c)$ lie in $T_x$, and hence $\lca(\phi_U(a),\phi_U(c))$ also lies in $T_x$.
Thus, $\lca(\phi_U(a),\phi_U(c))$ lies on both paths $P(\phi_U(a),x)$ and $P(\phi_U(c),x)$.
Since these paths are contained in $P(\phi_U(a),\phi_U(b))$ and $P(\phi_U(c),\phi_U(d))$, respectively, it follows that
\[
    \lca(\phi_U(a),\phi_U(c)) \in P(\phi_U(a),\phi_U(b)) \cap P(\phi_U(c),\phi_U(d)).
\]
As $x$ is the lowest vertex in this intersection, it follows that $\lca(\phi_U(a),\phi_U(c)) = x$.

We now consider two cases.
\medskip

\noindent I. If neither $\phi_U(a)$ nor $\phi_U(c)$ is an ancestor of the other, i.e.,      
\[x = \lca(\phi_U(a),\phi_U(c)) \notin \{\phi_U(a),\phi_U(c)\},\]
then
        $\lca(a,c) = \lca(\phi_U(a),\phi_U(c)) = x$, and thus $\phi_U(\lca(a,c)) = \phi_U(x)$.

\medskip

\noindent II. Otherwise, 
$$x = \lca(\phi_U(a),\phi_U(c)) \in \{\phi_U(a),\phi_U(c)\}.$$
In this case, $\phi_U(x)=x$ since $x\in\{\phi_U(a),\phi_U(c)\}\subset U$. Suppose, for instance, $$\lca(\phi_U(a),\phi_U(c)) = \phi_U(c).$$
Then both $a$ and $c$ lie in the subtree rooted at $\phi_U(c)$, and hence so does $\lca(a,c)$.
        Thus, the path from $\lca(a,c)$ to $\phi_U(c)$ is a subpath of the path $P(c,\phi_U(c))$.
        Since $\phi_U(c)$ is the first vertex from $U$ encountered when moving from $c$ to the root, it follows that $\phi_U(c)$ is also the first vertex from $U$ encountered when moving from $\lca(a,c)$ to the root.
        Therefore,
        \[
            \phi_U(\lca(a,c)) = \phi_U(c) = x = \phi_U(x).
        \]
        The case $\phi_U(a)=x$ is analogous.
    \end{proof}

    \begin{proof}[Proof of Claim~\ref{claim:f-lca-b}]
        If $\phi_U(b)=x$, then $x\in U$ and hence $\phi_U(x)=x$, giving $\phi_U(b)=\phi_U(x)$ as required.
        Assume now that $\phi_U(b)\neq x$. By the earlier assumptions, this also means $\phi_U(a)\neq x$.

        Since $\phi_U(a)$ lies below $x$, the vertex $a$ itself must also lie in the subtree rooted at $x$.
        Because $x \notin P(a,b)$, both $a$ and $b$ must lie in the same child subtree of $x$ -- otherwise $P(a,b)$ would pass through $x$. Denote that child by $y$.

        Observe that since $\phi_U(a)\in T_x$ but $\phi_U(a)\neq x$, the vertex $\phi_U(a)$ must be in $T_y$.
        The path from $b$ to the root decomposes into two parts: the path from $b$ to $y$ and the path from $x$ to $r$ (we use that $y$ is a child of $x$).
        If $\phi_U(b)$ were in the first part, then $P(\phi_U(a),\phi_U(b))$ would not contain $x$ (since both $\phi_U(a)$ and $\phi_U(b)$ would lie in $T_y$), contradicting our choice of $x$.
        Hence $\phi_U(b)$ must lie on the segment from $x$ to $r$.
        Since $\phi_U(b)\neq x$, it is the first vertex from $U$ encountered when moving upward from $x$, i.e., $\phi_U(b)=\phi_U(x)$.
        This completes the proof of Lemma~\ref{lem:disjoint-f}.
    \end{proof}
\end{proof}

\section{Auxiliary Lemmas used in Section~\ref{sec:community-detection}}\label{app:auxiliary_proofs}

\begin{lemma}
    \label{lem:Avg dist between points to avg dist between points and center}
    Let $P\subseteq \mathbb{R}^{d}$ be a set of points and $\mu$ be the average of those points, $\mu=\frac{1}{|P|}\sum_{p\in P}p$, then we have that:
    \[
        \Avg_{p,q\in P}\left(\|p-q\|^2\right)=2\cdot\Avg_{p\in P}\left(\|\mu-p\|^2\right)
    \]
\end{lemma}
\begin{proof}
    We first show the following intermediate claim
    \begin{claim}
        \label{clm: cluster_cost_diff_center}
        Let $s\in \mathbb{R}^{d}$ be a point, it holds that:
        \[
            \sum_{p\in P}\|s-p\|^2=\sum_{p\in P}\|p-\mu\|^2+|P|\cdot \|s-\mu\|^2.
        \]
    \end{claim}
    \begin{proof}
        Consider a random variable $X$ distributed uniformly in $P$, then we have that:
        \begin{align*}
            \sum_{p\in P}\|s-p\|^2&=|P|\cdot \Exp{\|s-X\|^2}\\
            &=|P|\cdot \Exp{\|s-\mu+\mu-X\|^2}\\
            &=|P|\cdot\Exp{\|s-\mu\|^2+\|\mu-X\|^2+2\inner{s-\mu, \mu-X}}\\
            &=|P|\cdot \|s-\mu\|^2+\sum_{p\in P}\|\mu-p\|^2+2|P|\cdot \inner{s-\mu, \Exp{\mu-X}}\\
            &=|P|\cdot\|s-\mu\|^2+\sum_{p\in P}\|\mu-p\|^2
        \end{align*}
    \end{proof}
    We now have that:
    \begin{align*}
        \Avg_{p,q\in P}\left(\|p-q\|^2\right)&=\frac{1}{|P|^2}\sum_{p\in P}\sum_{q\in P}\|p-q\|^2\\
        &=\frac{1}{|P|}\sum_{p\in P}\left(\frac{1}{|P|}\sum_{q\in P}\|p-q\|^2\right)
    \end{align*}
    For a fixed $p$, we apply \Cref{clm: cluster_cost_diff_center} to $\sum_{q\in P}\|p-q\|^2$, we have that:
    \begin{align*}
        \Avg_{p,q\in P}\left(\|p-q\|^2\right)&=\frac{1}{|P|}\sum_{p\in P}\left(\frac{1}{|P|}\cdot \sum_{q\in P}\|q-\mu\|^2+\|p-\mu\|^2\right)\\
        &=\frac{1}{|P|}\sum_{q\in P}\|q-\mu\|^2+\frac{1}{|P|}\sum_{p\in P}\|p-\mu\|^2\\
        &=2\Avg_{p\in P}\|\mu-p\|^2
    \end{align*}
\end{proof}

We give a proof to the following simple lemma that will be useful for our analysis.
\begin{lemma}
    \label{lem:subtree_of_size_alpha}
    Let $T$ be a rooted binary  tree with $n$ leaves and let $\alpha \geq \frac{1}{n}$. Then there is a node $u$ in the tree such that:
    \begin{align*}
        \alpha n\leq |L_u|<2\alpha n
    \end{align*}
\end{lemma}
\begin{proof}
    Let $u$ be a vertex that is farthest away from the root and furthermore, $|L_u|\geq \alpha n$ (by our assumption that $\alpha n\geq 1$ such a vertex exists).
    If $u$ is a leaf then $|L_u|=1\leq \alpha n$, which together with our assumption that $|L_u|\geq \alpha n$ gives us that $|L_u|=\alpha n< 2\alpha n$ and the claim holds.
    Otherwise, we will again show that $|L_u|< 2\alpha n$, suppose for the sake of contradiction that this is not the case and let $l$ and $r$ be the children of $u$. We have that $|L_u|=|L_r|+|L_l|$, our assumption that $|L_u|\geq 2\alpha n$ implies that either $|{L}_r|\geq \alpha n$ or $|{L}_l|\geq \alpha n$. In either case the assumption that $u$ is a vertex that is farthest away from the root such that $|{L}_u|\geq \alpha n$ is contradicted.
\end{proof}
\begin{lemma}
    \label{lem:getting_the_set_lca_by_two}
    Let $T$ be a binary rooted tree and let $S$ be a subset of the leaves. Let $w=\lca(S)$, then for every $a\in S$ there exists $b\in S$ such that $\lca(a,b)=w$.
\end{lemma}
\begin{proof}
    Consider the subtree rooted at $w$ and assume, without loss of generality, that $a$ lies in the left subtree of $w$. Set $b\in S$ to be any vertex in the right subtree of $w$ and observe that $\lca(a,b)=w$.
\end{proof}

\section{Distribution of graph \texorpdfstring{$G'$}{G'}}
We use the following claim that follows directly from Theorem 8.12 in \cite{mitzenmacher2017probability}.
\begin{claim}[Poisson thinning]
    \label{clm:Poisson-thinning}
    Let $X$ be a random variable following the Poisson distribution with parameter $\lambda$ and let $S$ be a set of size $X$. Suppose that we label each element in $S$ independently as: $0$ with probability $p$ and $1$ with probability $1-p$. Let $X_0, X_1$ be the random variables denoting the number of elements labeled $0$ and $1$ respectively. Then, $X_0\sim \Poi(p\lambda)$, $X_1\sim \Poi((1-p)\lambda)$ and furthermore, $X_0$ and $X_1$ are independent.
\end{claim}
The claim follows directly from the definition of a Poisson process at time $t$. We now prove the main lemma about the graph constructed by the algorithm.
\begin{proof}[Proof of Lemma \ref{lem:edge_prob2}]
    First notice that the $Q(n,\lambda, \eta)$ model is equivalent to the following model parameterized by $n$, $\lambda$ and $q$, for $q=\frac{3}{2}\eta$.
    For every set $\set{a,b,c,d}\subseteq \binom{L}{4}$ the number of quartet constraints involving these variables is a random variable $X_{abcd}\sim \Poi(\lambda)$. In each of these occurrences the quartet constraint is $T(a,b,c,d)$ with probability $1-q$, which we label as a clean quartet, and a uniformly random constraint with the remaining probability, we label such quartets as random quartets. This is only for the sake of analysis, the algorithm doesn't know which quartets are clean and which are random.
    By Claim \ref{clm:Poisson-thinning}, for a set $\set{a,b,c,d}$, the number of quartets labeled clean follows $\Poi(\lambda (1-q))$ and the number of quartets labeled random follows $\Poi(\lambda q)$, moreover, the two random variables are independent. We can therefore view the (multi)-set $Q$ as $Q=Q_{c}\cup Q_r$, where $Q_c$ contains all the quartet constraints labeled clean, $Q_r$ contains all the quartet constraints labeled random and the two sets are independent.
    Furthermore, we can view the graph constructed by the algorithm as the union of two independent random (multi)-graphs constructed by $Q_c$ and $Q_r$ respectively:
    \begin{align*}
        G'=G'_{c}\cup G'_r.
    \end{align*}

    Using Poisson thinning again, we have that for every $a,b\in L'$ the number of quartet constrains that contribute an edge $(a,b)$ to $G'_c$ is an independent Poisson random variable with parameter $\frac{1}{2}\lambda(1-q)|\mathcal{Q}_{ab|L'}|$. Recall that $\mathcal{Q}_{ab|L'}$ is the subset of quartet constraints involving elements of $L'$ that place $a$ and $b$ together, the factor $\frac{1}{2}$ comes from the fact that for a constraint $ab|cd$ we add an edge $(a,b)$ with probability $\frac{1}{2}$. Similarly for $G'_r$, the number of quartets that contribute an edge between $a$ and $b$ is an independent random variable following the Poisson distribution with parameter 
    $$\frac{1}{6}\lambda q \binom{|L'\setminus \set{a,b}|}{2}=\frac{1}{6}\frac{c'}{n^3} q\binom{\gamma n-2}{2}.$$
    This directly gives us the first item of the lemma, namely that edges are independent.

    We now focus on giving estimates for the probabilities of edges. We will analyze graphs $G'_c$ and $G'_r$ separately and then use that the two graphs are independent. We start with the easier, namely graph $G_r'$. The probability that there is at least one edge between $a,b$ is:
    \begin{align*}
        \prob{(a,b)\in G'_r}=\prob{X_{ab}>0},
    \end{align*}
    where $X_{ab}\sim \Poi\left(\frac{1}{6}\frac{c'}{n^3}q\binom{\gamma n-2}{2}\right)$. Using Taylor's theorem we have that:
    \begin{align*}
        \prob{X_{ab}>0}&=1-\prob{X_{ab}=0}\\
        &=\frac{1}{6}\frac{c'}{n^3}q\binom{\gamma n-2}{2}+O\left(\frac{1}{n^2}\right)\\
        &=\left(1+o(1)\right)\frac{c'\gamma^2q}{12n}.
    \end{align*}
    We now turn our attention to graph $G'_c$. Using similar calculations we can easily get that:
    \begin{align}
        \label{eq:edge prob characterization in Gc}
        \prob{(a,b)\in G'_c}=\left(1+o(1)\right)\frac{c'(1-q)}{2n^3}\cdot|\mathcal{Q}_{ab|L'}|
    \end{align}
    We will estimate the probability $\prob{(a,b)\in G'_c}$ by estimating $|\mathcal{Q}_{ab|L'}|$. Let $\set{C,D}$ be a random element of $\binom{L'\setminus \set{a,b}}{2}$. We have that:
    \begin{align*}
        |\mathcal{Q}_{ab|L'}|=\prob{ab|CD=T(a,b,C,D)}\cdot\binom{\gamma n-2}{2}.
    \end{align*}
    We can now turn our attention to estimating $\prob{ab|CD=T(a,b,C,D)}$. Recall that we are using the notation $\delta_{uv|L'}=\frac{|L_{uv}\cap L'|}{|L'|}$ ($\delta_{uv|L'}$ is the fraction of leaves in $L'$ that belongs to $L_{uv}$). To derive a lower bound observe that  if neither of $C$ and $D$ are in the subtree of $\lca(a,b)$ then $ab|CD=T(a,b,C,D)$:
    \begin{align*}
        \prob{ab|CD=T(a,b,c,d)}&\geq \prob{C,D\not \in L_{ab}\cap L'}\\
        &\geq (1-\delta_{ab|L'})(1-\delta_{ab|L'}+o(1))\\
        &=1-2\delta_{ab|L'}+\delta_{ab|L'}^2+o(1),
    \end{align*}
    where the $o(1)$ terms are introduced because of the fact that the variables are sampled without replacement. This gives us that:
    \begin{align}
        \label{eq:QabL lower bound}
        |\mathcal{Q}_{ab|L'}|\geq(1-2\delta_{ab|L'}+\delta_{ab|L'}^2+o(1))\cdot \binom{\gamma n -2}{2}.
    \end{align}
    For an upper bound on $|\mathcal{Q}_{ab|L'}|$ we notice that if $C$ is in the subtree of $\lca(a,b)$, while $D$ is not, then $ab|CD\neq T(a,b,C,D)$. Similarly if $C\not\in L_{ab}\cap L'$ while $D\in L_{ab}\cap L'$. Since these two events are disjoint, we get that:
    \begin{align*}
        \prob{ab|CD\neq T(a,b,C,D)}&\geq   \prob{C\in L_{ab}\cap L' \text{ and } D\notin L_{ab}\cap L'}\\
        &\quad +\prob{D\in L_{ab}\cap L' \text{ and } C\notin L_{ab}\cap L'}\\
        &\geq 2\delta_{ab|L'}(1-\delta_{ab|L'}+o(1))\\
        &= 2\delta_{ab|L'} - 2 \delta_{ab|L'}^2 + o(1),
    \end{align*}
    where again $o(1)$ terms are due to the fact that we are sampling without replacement. We have that:
    \begin{align*}
        \prob{ab|CD=T(a,b,C,D)}\leq 1-2\delta_{ab|L'}+2\delta_{ab|L'}^2+o(1),
    \end{align*}
    which in turn gives us that:
    \begin{align*}
        |\mathcal{Q}_{ab|L'}|\leq \left(1-2\delta_{ab|L'}+2\delta_{ab|L'}^2+o(1)\right)\cdot {\binom{\gamma n-2}{2}}
    \end{align*}

    We finally derive a bound on $|\mathcal{Q}_{ab|L'}|$ that corresponds to the third item of the lemma. Suppose that $a\in \mathcal{L}$ and $b\in \mathcal{R}$ (here, we are using $\mathcal{L}$ and $\mathcal{R}$ to denote the leaves of left and right subtree respectively) and let $\alpha=\frac{|\mathcal{L}\cap L'|}{L'}$ and $\beta=\frac{|\mathcal{R}\cap L'|}{|L'|}$. Observe that if $C\in \mathcal{L}\cap L'$ while $D\in \mathcal{R}\cap L'$ (or vice versa) then $ab|CD\neq T(a,b,C,D)$. We have that:
    \begin{align*}
        \prob{ab|CD\neq T(a,b,C,D)}\geq& \prob{C\in \mathcal{L}\cap L'\text{ and }D\in \mathcal{R}\cap L'}\\
        &+\prob{D\in \mathcal{L}\cap L'\text{ and }C\in \mathcal{R}\cap L'}\\
        \geq&2\alpha \beta+o(1)
    \end{align*}
    This in turn gives us that
    \begin{align}
        \prob{ab|CD=T(a,b,c,d)}\leq 1-2\alpha \beta +o(1).
    \end{align}
    We now go back to $G'$, we have that:
    \begin{align*}
        \prob{(a,b)\in G'}&=\prob{(a,b)\in G_c'}+\prob{(a,b)\in G'_r}-\prob{(a,b)\in G'_c}\cdot \prob{(a,b)\in G'_r}\\
        &=\prob{(a,b)\in G_c'}+\prob{(a,b)\in G'_r}+O\left(\frac{1}{n^2}\right).
    \end{align*}
Having bounded $|\mathcal{Q}_{ab|L'}|$, we can now prove the second and third items of the lemma. For the lower bound:
    \begin{align*}
        \prob{(a,b)\in G'_c}+\prob{(a,b)\in G'_r}&=(1+o(1))\cdot\left(\frac{c'(1-q)}{2n^3}|\mathcal{Q}_{ab|L'}|+\frac{c'\gamma^2 q}{12n}\right)\\
        &\geq \frac{c\gamma^2}{4}(1-q)\left(1-2\delta_{ab|L'}+\delta_{ab|L'}^2\right)+\frac{c'\gamma^2}{12}q\frac{1}{n}+o\left(\frac{1}{n}\right),
    \end{align*}
    this directly implies that:
    \begin{align*}
        \prob{(a,b)\in G'}\geq \frac{c'\gamma^2}{4}(1-q)\left(1-2\delta_{ab|L'}+\delta_{ab|L'}^2\right)+\frac{c'\gamma^2}{12}q\frac{1}{n}+o\left(\frac{1}{n}\right),
    \end{align*}
    replacing $\eta=\frac{3}{2}q$ and ${c}=\frac{c'\gamma^2}{4}$ gives us the lower bound of the second item. The other bounds follow similarly.
\end{proof}

\newpage
\label{app:distribution of graph}

\section{Quartet Query Complexity}\label{app:query-complexity}

In this section we extend the results of \cite{emamjomeh2018adaptive} on the query complexity of learning a hierarchical clustering to the quartet (unrooted) setting. In particular, they showed that adaptive algorithms can exactly reconstruct a rooted binary hierarchy from $O(n\log n)$ \textit{triplet} queries, while any non-adaptive algorithm requires $\Omega(n^3)$ \textit{triplet} queries, even without noise. We show that both results naturally extend to quartets.

\begin{lemma}[Adaptive quartets]\label{lem:adaptive-quartets}
    Assume that the responses to quartet queries are correct independently with probability $p>1/2$, and adversarially incorrect otherwise. There exists an adaptive algorithm that learns the correct unrooted binary tree with probability at least $1-\delta$, using at most $O(n\log n+n\log(1/\delta))$ quartet queries. In the absence of noise ($p=1$), the algorithm exactly reconstructs the tree using $O(n\log n)$ quartets.
\end{lemma}

\begin{proof}
    Consider the rooted tree $T_r$ obtained by rooting $T$ at any leaf $r$.
    For each triplet $(a,b,c)$, the quartet $(r,a,b,c)$ yields the same information:
    if the answer is $ra|bc$ then $(b,c)$ is the closest pair in $(a,b,c)$,
    and so on.
    Thus, one can simulate each triplet query
    in the algorithm of \cite{emamjomeh2018adaptive} by a single quartet query with $r$.
    The resulting insertion-based algorithm reconstructs $T_r$, hence $T$, exactly.
    In the independent noise model, the same repetition and majority-vote
    mechanism used in their Theorem~5.1 applies directly, preserving the
    $O(n\log n+n\log(1/\delta))$ query complexity.
\end{proof}

\begin{lemma}[Non-adaptive lower bound for quartets]\label{lem:nonadaptive-quartets}
    Any non-adaptive algorithm that reconstructs an unrooted binary tree on $n$ leaves, even without noise, requires $\Omega(n^3)$ quartet queries.
\end{lemma}

\begin{proof}
    Analogous to \cite{emamjomeh2018adaptive}, let $C = \{a,b,c,d\}$ be one of these disjoint 4-leaf subtrees; it attaches to the rest of the tree by a single edge.
    Consider two global trees that are identical outside $C$ but differ only in the internal structure within $C$, e.g., $ab|cd$ vs. $ac|bd$.

    Then, any quartet query using $\le 2$ leaves from $C$ (and the rest outside $C$) has the \textit{same answer} under both trees. Therefore, to distinguish the two internal structures of $C$, a non-adaptive algorithm must include at least one query containing $\ge 3$ leaves from $C$ (either a 3-of-4 or the 4-of-4 query).

    For a uniformly random 4-set query on the $n$ leaves,
    \[
        \prob{\text{query hits exactly 3 of } C}
        = \frac{\binom{4}{3}\binom{n-4}{1}}{\binom{n}{4}}
        = \frac{4(n-4)}{\binom{n}{4}},
    \]
    and
    \[
        \prob{\text{query hits all 4 of } C}
        = \frac{1}{\binom{n}{4}}.
    \]
    Hence,
    \[
        \prob{\text{query hits}\ge 3\text{ from } C}
        = \frac{4(n-4)+1}{\binom{n}{4}}
        = \frac{(4n-15)\cdot 24}{n(n-1)(n-2)(n-3)}
        = \Theta\!\left(\frac{1}{n^3}\right).
    \]

    There are $n/4$ disjoint size-4 clusters. With $k$ non-adaptive quartet queries, the expected number of clusters whose internal structure is ``touched'' is at most
    \[
        \frac{n}{4} \cdot k \cdot
        \frac{4(n-4)+1}{\binom{n}{4}}
        = k \cdot \frac{24n - 90}{(n-1)(n-2)(n-3)}
        = \Theta\!\left(\frac{k}{n^2}\right).
    \]

    To recover the full tree, the algorithm must determine the internal structure for almost all of these $n/4$ disjoint clusters, so the expected number of touched clusters must be $\Omega(n)$.
    This requires
    \[
        \Theta\!\left(\frac{k}{n^2}\right) = \Omega(n)
        \quad \Longrightarrow \quad
        k = \Omega(n^3).
    \]

\end{proof}

\section{Figures}\label{app:fig}

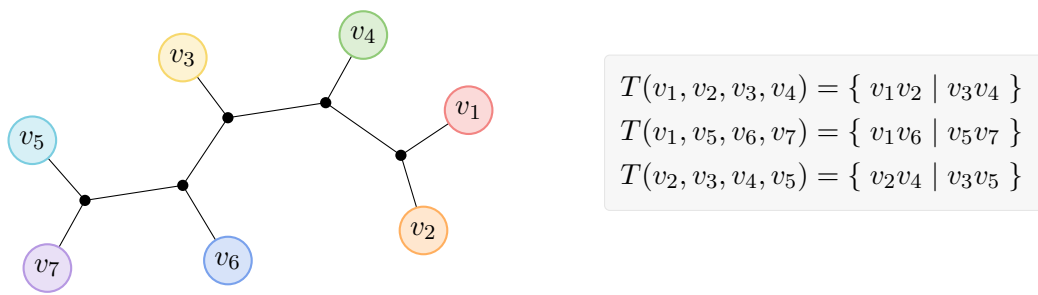
\begin{figure}[h]
    \centering
    {
        \begin{tikzpicture}
            \tikzset{
                dot/.style={circle,fill=black,inner sep=1.5pt},
                leaf/.style={circle,minimum size=18pt,inner sep=0pt,draw=black,thick}
            }

            \definecolor{vone}{RGB}{242,132,130}
            \definecolor{vtwo}{RGB}{255,175,95}
            \definecolor{vthree}{RGB}{247,216,108}
            \definecolor{vfour}{RGB}{144,202,119}
            \definecolor{vfive}{RGB}{122,205,224}
            \definecolor{vsix}{RGB}{121,161,236}
            \definecolor{vseven}{RGB}{182,152,226}

            \coordinate (I0) at (3.0, 2.7);
            \coordinate (I1) at (4.3, 2.9);
            \coordinate (I2) at (2.4, 1.8);
            \coordinate (I3) at (1.1, 1.6);
            \coordinate (I4) at (5.3, 2.2);

            \coordinate (v1) at (6.2, 2.8);
            \coordinate (v2) at (5.6, 1.2);
            \coordinate (v3) at (2.4, 3.5);
            \coordinate (v4) at (4.8, 3.8);
            \coordinate (v5) at (0.4, 2.4);
            \coordinate (v6) at (3.0, 0.8);
            \coordinate (v7) at (0.6, 0.7);

            \draw (I0) -- (v3);
            \draw (I0) -- (I1);
            \draw (I0) -- (I2);
            \draw (v1) -- (I4);
            \draw (v7) -- (I3);
            \draw (I1) -- (v4);
            \draw (I1) -- (I4);
            \draw (I2) -- (v6);
            \draw (I2) -- (I3);
            \draw (I3) -- (v5);
            \draw (I4) -- (v2);

            \foreach \p in {I0,I1,I2,I3,I4} {\node[dot] at (\p) {};}

            \node[leaf,fill=vone!30,draw=vone]    at (v1) {$v_1$};
            \node[leaf,fill=vtwo!30,draw=vtwo]    at (v2) {$v_2$};
            \node[leaf,fill=vthree!30,draw=vthree]at (v3) {$v_3$};
            \node[leaf,fill=vfour!30,draw=vfour]  at (v4) {$v_4$};
            \node[leaf,fill=vfive!30,draw=vfive]  at (v5) {$v_5$};
            \node[leaf,fill=vsix!30,draw=vsix]    at (v6) {$v_6$};
            \node[leaf,fill=vseven!30,draw=vseven]at (v7) {$v_7$};

            \node[anchor=west, align=left, inner sep=6pt, rounded corners=2pt,
            fill=black!3, draw=black!10] at (8.0,2.5) {
                $
                \begin{aligned}
                    T(v_1,v_2,v_3,v_4) &= \{\; v_1 v_2 \mid v_3 v_4 \;\}\\
                    T(v_1,v_5,v_6,v_7) &= \{\; v_1 v_6 \mid v_5 v_7 \;\}\\
                    T(v_2,v_3,v_4,v_5) &= \{\; v_2 v_4 \mid v_3 v_5 \;\}
                \end{aligned}$
            };
        \end{tikzpicture}
    }
    \caption{A phylogenetic tree $T$ on $7$ taxa, showing how quartet examples are derived from $T$.}
    \label{fig:quartet_label}
\end{figure}

\end{document}